\documentclass[reprint,superscriptaddress,amsmath,amssymb,longbibliography,twocolumn]{revtex4-2}

\usepackage{amsmath}
\usepackage{mathtools}
\usepackage{amssymb}
\usepackage{graphicx}
\usepackage{bm}
\usepackage[colorlinks, linkcolor=blue, citecolor=blue, urlcolor=blue,breaklinks=true]{hyperref}
\usepackage{color,dsfont,multirow,booktabs}
\usepackage{braket}
\usepackage{tabularx}
\usepackage{lipsum,booktabs}
\usepackage{subfiles}
\usepackage{subcaption}
\usepackage{caption}
\usepackage{ragged2e}
\usepackage{gensymb}
\usepackage{orcidlink}
\usepackage[T1]{fontenc}
\usepackage[normalem]{ulem}
\usepackage{amsthm}  
\newtheorem{theorem}{Theorem}        
\newtheorem{lemma}{Lemma}            
\newtheorem{proposition}{Proposition} 
\newtheorem{corollary}{Corollary}[section]
\newcommand\tr{\operatorname{Tr}}

\newcommand{\commentold}[1]{}
\DeclareMathSymbol{:}{\mathpunct}{operators}{"3A}
\newcommand{\dicke}[2]{\ket{\smash{D_{#2}^{#1}}}}
\newcommand{\dickeD}[2]{\bra{\smash{D_{#2}^{#1}}}}

\newcommand\numberthis{\addtocounter{equation}{1}\tag{\theequation}}

\begin{document}
\date{\today}

\title{Harnessing Environmental Noise for Quantum Energy Storage}

\author{Borhan Ahmadi\orcidlink{0000-0002-2787-9321}}
\email{borhan.ahmadi@ug.edu.pl}
\address{ International
  Centre for Theory of Quantum Technologies (ICTQT), University of Gdańsk, Jana Bażyńskiego 8, 80-309 Gdańsk, Poland}
\author{Aravinth Balaji Ravichandran\orcidlink{0000-0001-5379-7525}}
\address{ International
  Centre for Theory of Quantum Technologies (ICTQT), University of Gdańsk, Jana Bażyńskiego 8, 80-309 Gdańsk, Poland}
\email{aravinth.ravichandran@ug.edu.pl}
\author{Paweł Mazurek\orcidlink{0000-0003-4251-3253}}
\email{pawel.mazurek@ug.edu.pl}
\address{Institute of Informatics, Faculty of Mathematics, Physics and Informatics, University of Gdańsk, Wita Stwosza 63, 80-308 Gdańsk, Poland}
\author{Shabir Barzanjeh}
\address{Department of Physics and Astronomy, University of Calgary, Calgary, AB T2N 1N4 Canada}
\author{Paweł Horodecki}
\address{ International
  Centre for Theory of Quantum Technologies (ICTQT), University of Gdańsk, Jana Bażyńskiego 8, 80-309 Gdańsk, Poland}

\begin{abstract}
Quantum hardware increasingly relies on energy reserves that can later be converted into useful work; yet, most battery-like proposals demand coherent drives or engineered non-equilibrium resources, limiting practicality in noisy settings. We develop an autonomous charging paradigm in which an ensemble of identical two-level units, collectively coupled to a thermal environment, acquires work capacity without any external control. The common bath mediates interference between emission and absorption pathways, steering the many-body state away from passivity and into a steady regime with nonzero extractable work. The full charging dynamics and closed-form expressions are obtained for the steady-state, showing favorable scaling with the number of cells that approach the many-body optimum. We show that the mechanism is robust to local noise: under a convex mixture of collective and local dissipation, non-zero steady-state ergotropy persists, exhibits counterintuitive finite-temperature optima, and remains operative when the collective channel is comparable to or stronger than the local one. We show that environmental fluctuations can be harnessed to realize drive-free, scalable quantum batteries compatible with circuit- and cavity-QED platforms. Used as local work buffers, such batteries could potentially enable rapid ancilla reset, bias dissipative stabilizer pumps, and reduce syndrome-extraction overhead in fault-tolerant quantum computing.
\end{abstract}

\maketitle
\section*{Introduction}

The ability to store energy and later convert it into useful work is a core requirement in modern quantum technologies, including computing~\cite{RevModPhys.80.885, PRXQuantum.1.020101}, secure communication~\cite{Chen_2021}, and quantum sensing~\cite{RevModPhys.89.035002}. In all of these platforms, storage is only valuable insofar as the energy remains extractable — i.e., available as work on demand \cite{Allahverdyan_2004}. Quantum processors benefit from local energy reservoirs that support fast gate operations \cite{PhysRevLett.124.067701,PhysRevResearch.6.033215,RevModPhys.93.025005}, active reset \cite{kobayashi2023feedback}, and error correction \cite{sutcliffe2025distributed}; quantum networks require controlled energy transfer for entanglement distribution and repeater operation; and quantum sensors rely on stable energy flows to reach sensitivities beyond classical limits. The common challenge is to ensure that the stored energy remains ergotropic, meaning it can be converted to work even in the presence of decoherence and noise.

As we know from classical thermodynamics, thermal equilibrium states are passive: no work can be extracted from them by unitary dynamics alone~\cite{Cengel1989ThermodynamicsA}. The quantum refinement of this statement uses ergotropy—the maximal work obtainable from a quantum state under cyclic unitary operations~\cite{Allahverdyan_2004}. Passive states (including Gibbs states) have zero ergotropy; non-passive states have non-zero ergotropy enabled by population inversion, coherence, or correlations that can be rearranged by a unitary. This distinction between passive and non-passive states motivates the concept of quantum batteries (QBs): engineered quantum systems that store energy in a form convertible to work (ergotropy)~\cite{PhysRevE.87.042123, PhysRevLett.122.047702, PhysRevResearch.5.013155,Tomas2023, Rodríguez_2024, PhysRevA.107.032218, binder2015quantacell, kamin2023steady, kamin2020non, PhysRevA.107.042419, quach2022superabsorption, barra2022quantum, PhysRevLett.131.240401, bv4w-jr6q, PhysRevLett.134.220402, PhysRevA.110.052404, hadipour2025nonequilibrium}.

Despite rapid theoretical progress, many QB designs rely on external coherent drives, engineered non-thermal reservoirs, or specially prepared nonequilibrium states to induce non-passivity \cite{RevModPhys.96.031001}. These methods are powerful but often impractical in uncontrolled or high-temperature settings where coherence is fragile and precise control is limited. This raises a basic question: Can a quantum battery be charged using only thermal energy, with no external coherent drive? Demonstrating this would establish a route to autonomous and scalable devices.

Recent results in many-body quantum thermodynamics suggest a promising path. When several quantum units (such as qubits) couple collectively to a common environment, dissipation can become cooperative, modifying relaxation pathways relative to independent baths \cite{PhysRevA.2.883}. Symmetric (Dicke-like) coupling permits interference between indistinguishable emission and absorption processes, building correlations that are inaccessible to independently coupled subsystems. Such collective effects have been shown to enhance the performance bounds of quantum thermal machines and widen the conditions under which useful work can be extracted~\cite{PhysRevLett.124.170602, Niedenzu_2018, Boeyens_2025, PhysRevLett.132.210402, PhysRevApplied.23.024010}. Importantly, collective coupling can produce steady-state coherences or correlated populations even with thermal reservoirs, thereby creating non-passive reduced states of the working medium without external coherent control.

Here we introduce a model of an autonomous, self-charging quantum battery that acquires ergotropy directly from a thermal reservoir via collective dissipation. The battery consists of an ensemble of identical two-level systems (cells) symmetrically coupled to a bosonic bath. Despite the bath's incoherent nature, interference between collective absorption and emission channels drives the battery to a non-passive steady-state with non-zero ergotropy. We derive the full charging dynamics, characterize the steady-state ergotropy as a function of reservoir temperature and system size, and identify the parameter regimes where thermal energy is efficiently converted into ergotropic work. In particular, we show that (i) the extractable work increases with the number of cells due to cooperative effects; (ii) within the validity of the collective-coupling model, the steady-state ergotropy can approach the maximal value allowed for the corresponding constraints; and (iii) counter to conventional expectations, higher reservoir temperatures can enhance charging by activating additional collective channels. Moreover, they are robust to imperfections. Even with a convex mixture of collective and local dissipation, steady-state ergotropy persists and exhibits finite-temperature optima and a stable activation lobe at moderate collectivity, tolerating added dephasing.

These results show a thermodynamic regime in which environmental noise—mediated by collective coupling—becomes a resource for energy storage. The mechanism yields a practical design principle for scalable, autonomous QBs that function in resource-limited or high-noise conditions. Our proposed approach is compatible with existing experimental platforms such as circuit and cavity QED, where ensembles of qubits or atoms naturally experience partial collective coupling to common (and often thermalized) modes. Starting from the most accessible laboratory condition—a product of local thermal states—the battery autonomously develops non-passivity, with extractable work that systematically grows with system size. Our analysis thus provides both a conceptual advance in quantum thermodynamics and a feasible blueprint for implementing autonomous, thermally charged quantum batteries to power quantum processors.

\section*{Results}
\subsection*{Theoretical Model and Hamiltonian}

As shown in Fig.~\ref{Model}, we model the QB as an ensemble of $N$ non-interacting two-level atoms (qubits), each acting as an individual cell of the battery. The free Hamiltonian is (with $\hbar = 1$) 
\begin{equation}
    H_B = \sum_{i=1}^N \omega\, \sigma_+^{(i)} \sigma_-^{(i)},
\end{equation}
where $\omega$ denotes the energy splitting between the ground and excited states, and $\sigma_+^{(i)}$ ($\sigma_-^{(i)}$) are the raising (lowering) operators for the $i$-th atom. The battery is initialized in a fully uncharged state, with all atoms prepared in their ground states.  
To implement charging, the ensemble is coupled to a common thermal reservoir at temperature $T_c$, modeled as a continuum of bosonic modes. The interaction Hamiltonian takes the form \cite{breuer2002theory}
\begin{equation}
    H_I = \sum_k \left( \sum_{i=1}^N g\, \sigma_+^{(i)} \right) \otimes a_k + \text{h.c.},
\end{equation}
where $g$ is the coupling strength and $a_k$ is the annihilation operator of the $k$-th reservoir mode. This form of interaction implies that all atoms couple collectively to each reservoir mode, a condition satisfied when the interatomic spacing is much smaller than the thermal or optical wavelength of the bath \cite{Agarwal1974,Stephen1964,damanet2016competition,PhysRevA.99.052105} (see Supplementary Note I).  

Such a collective coupling regime is experimentally accessible in a variety of architectures, including ensembles of atoms in high-finesse optical cavities~\cite{bohnet2012steady}, superconducting qubits coupled to a common transmission line~\cite{van2013photon,PhysRevA.88.043806}, and dense arrays of spin defects in solid-state hosts~\cite{PhysRevLett.107.060502,PhysRevX.14.041055}. Under these conditions, dissipation occurs through collective jump operators rather than local ones, leading to cooperative energy exchange with the environment~\cite{PhysRevA.2.889,GrossHaroche1982,RevModPhys.95.015002}. This collective dissipation is the essential ingredient that enables the battery to acquire ergotropy directly from an incoherent thermal reservoir.

Under the standard Born--Markov and secular approximations~\cite{breuer2002theory,audretsch2007entangled}, the reduced dynamics of the battery in the interaction picture is governed by the master equation (see Supplementary Note I)
\begin{equation}
    \dot{\rho} = \mathcal{L}[\rho],
    \label{ME}
\end{equation}
where the Liouvillian superoperator $\mathcal{L}$ is given by \cite{Agarwal1974,Stephen1964,damanet2016competition,PhysRevA.99.052105}
\begin{align}\label{Liouvillian}
    \mathcal{L}[\boldsymbol{\cdot}] &= \gamma_c (n_c+1)\left(J_-\,\boldsymbol{\cdot}\, J_+ - \tfrac{1}{2}\{J_+J_-, \boldsymbol{\cdot}\} \right) \nonumber \\
    &\quad + \gamma_c n_c \left(J_+\,\boldsymbol{\cdot}\, J_- - \tfrac{1}{2}\{J_-J_+, \boldsymbol{\cdot}\} \right).
\end{align}
Here, $\gamma_c$ is the decay rate, and \(n_c = ({e^{\frac{\omega}{kT_c}} - 1})^{-1}\) is the Bose--Einstein distribution function, which gives the mean thermal occupation number at frequency $\omega$ and reservoir temperature $T_c$, with $k$ being the Boltzmann constant. The operators $J_\pm = \sum_{i=1}^N \sigma_\pm^{(i)}$ are the collective lowering and raising operators, respectively.  

By expressing the dynamics in terms of the collective operators $J_\pm$, we capture the cooperative nature of the system--reservoir interaction, which arises because all atoms couple \emph{indistinguishably} to the same reservoir modes. The collective form of \eqref{Liouvillian} makes explicit that dissipation proceeds through the permutation–symmetric jumps $J_\pm$, i.e., it is insensitive to which atom emitted or absorbed. In particular, when the emitters are deeply subwavelength so that
\(\xi_{ij}\ll 1\) (with \(\xi_{ij}\equiv k\,|\mathbf r_i-\mathbf r_j|\), $k$ the resonant wavenumber), the radiation field cannot resolve individual sites and no which–atom information is available. Consequently, the dynamics preserves the exchange symmetry of the atomic state and implements indistinguishable decay/absorption pathways through the collective operators \(J_\pm\). This structure of the dissipator reflects that energy exchange with the thermal bath proceeds through collective emission and absorption processes. Importantly, such collective coupling enables the system to evolve into non-passive steady states with non-zero ergotropy, despite the absence of any external coherent driving or measurement.  

We use ergotropy as the central figure of merit in our analysis, as it directly quantifies the maximum work that can be extracted from a quantum state through a cyclic unitary process—one in which the system Hamiltonian returns to its initial form at the end of the evolution, $H(t_0) = H(t_f)$~\cite{Allahverdyan_2004}.
Let $E_B(\rho) := \mathrm{Tr}\{\rho H_B\}$ denote the internal energy of the battery when in state $\rho$. The ergotropy $\mathcal{W}(\rho)$ is defined as~\cite{Allahverdyan_2004} 
\begin{equation}
    \mathcal{W}(\rho) := E_B(\rho) - E_B(\rho_p),
\end{equation}
where $\rho_p$ is the passive state associated with $\rho$, obtained by minimizing the energy over all unitary transformations,
\begin{equation}
   \rho_p = \arg \min_{U \in SU(2^N)} \mathrm{Tr}\{U \rho U^{\dagger} H_B\}.
\end{equation}
Note that the passive state $\rho_p$ has the same spectrum (eigenvalues) as $\rho$, but its populations are arranged so as to minimize the energy with respect to $H_B$. The difference between the energy of $\rho$ and that of $\rho_p$ defines the ergotropy, i.e., the maximum amount of work that can be extracted from the state via unitary operations alone. Thermal (Gibbs) states are always passive and therefore have zero ergotropy \cite{Allahverdyan_2004}, whereas states with population inversion or quantum coherence in the energy eigenbasis yield non-zero ergotropy. In this way, ergotropy provides a rigorous operational measure of the stored, extractable energy in the quantum battery.
\begin{figure}[t!]
\center
\includegraphics[width=8cm]{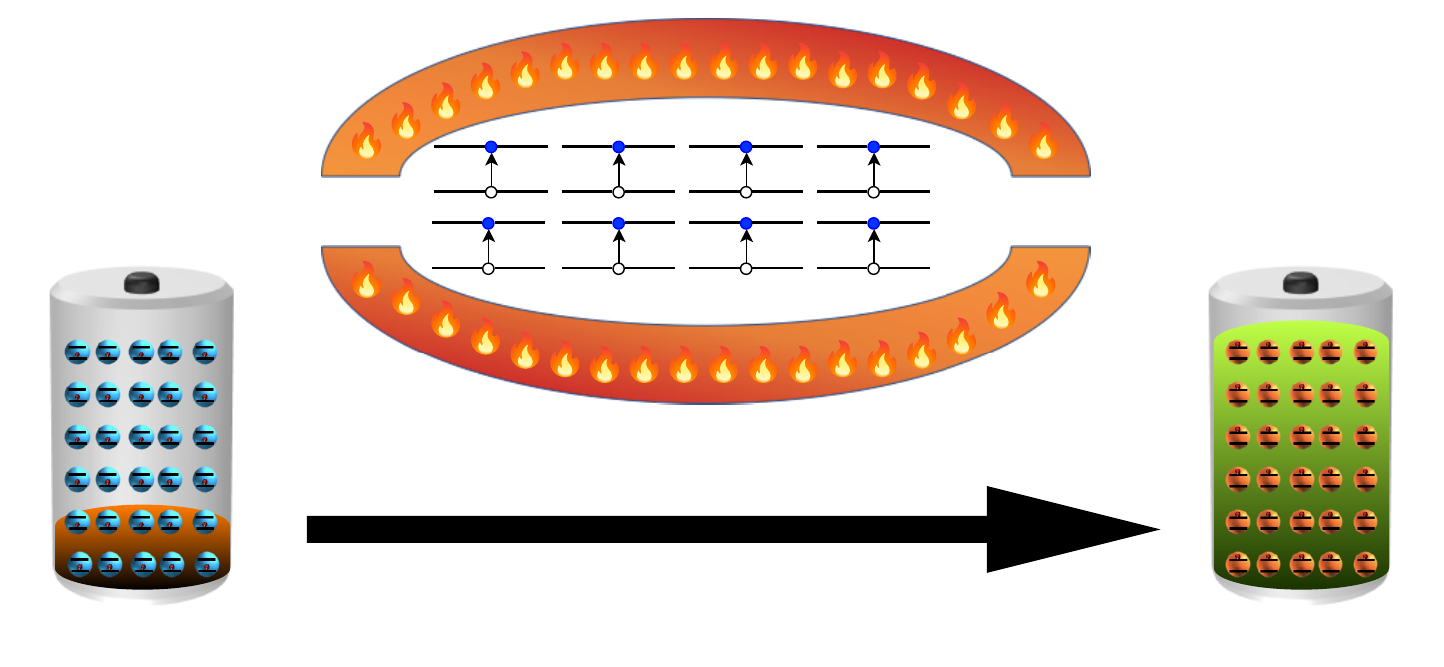}
\captionsetup{justification=justified}
\caption{\textbf{Schematic representation of the charging process of the QB.} On the left is the uncharged QB with each qubit constituting a cell of the QB. Depicted as cavity is the shared reservoir at a finite temperature $T_R$ \justifying that interacts with each cell collectively exciting each individual cell. As a result, the QB is charged as depicted in the right.} 
\label{Model}
\end{figure}

\subsection*{Charging Process of the QB}

To analyze the charging process induced by the collective dissipator, we now derive the transient and steady-state dynamics of the $N$-qubit battery. Our analysis applies to arbitrary initial states $\rho(t_{0})$, including those outside the fully symmetric subspace. Below we shall show how to represent the evolving state $\rho(t)$ in the form that is suitable with respect to the final steady-state. More specifically, it turns out that the Liouvillian, expressed in terms of collective angular momentum operators, naturally decomposes the dynamics into irreducible sectors suitable to represent the state $\rho(t)$. We first note that the positive operators $J_- J_+$ and $J_+ J_-$ commute with each other and with the $J_{z}$ operator.
\begin{equation}
    J_z = \frac{1}{2}\sum_{i=1}^{N} \sigma_z^{(i)},
\end{equation}
where $\sigma_z^{(i)}$ is the Pauli-$z$ operator acting on the $i$-th qubit. As a result, $J_- J_+$, $J_+ J_-$, and $J_z$ share a common eigenbasis. Each eigenstate is associated with a fixed number of excitations, determined by its $J_z$ eigenvalue. These eigenstates can be labeled by quantum numbers $\ket{j, m, \sigma}$, which satisfy  
\begin{align}
    J^2 \ket{j, m, \sigma} &= j(j+1) \ket{j, m, \sigma}, \\
    J_z \ket{j, m, \sigma} &= m \ket{j, m, \sigma},
\end{align}
where $j$ is the total spin quantum number, $m$ the magnetic quantum number, and $\sigma$ indexes the degeneracy within each $(j,m)$ subspace~\cite{Shankar1994,sakurai2020modern}. For even $N$, the allowed values of $j$ range from $0 \leq j \leq N/2$, while for odd $N$ they range from $1/2 \leq j < N/2$. The degeneracy index runs over  
\[
1 \leq \sigma \leq \binom{N}{N/2 - j} - \binom{N}{N/2 - j - 1},
\]  
and $m$ takes integer steps from $-j$ to $j$. Since $[J^2, J_{\pm}] = 0$ and $[J^2, J_z] = 0$, the ladder operators $\hat{J}_{\pm}$ only connect eigenstates of $J_z$ within the same $J^2$ subspace. Consequently, the Liouvillian evolution preserves this block structure, leading to dynamics that are block-diagonal in the so-called Bohr sectors~\cite{davies1974markovian,Baumgartner_2008}. These sectors are labeled by the multi-index  
\[
k = (j, j', \sigma, \sigma', \Delta_{J_z}),
\]  
where $\Delta_{J_z} = m - m'$. The full system density matrix at any moment of time (we drop here the time variable $t$ so that the formulas are not too obscured) takes the form
\begin{equation}\label{10NEW}
    \rho^{k} = \sum_{m,m':\, m - m' = \Delta_{J_z}} 
    \rho_{j,m,\sigma;\, j',m',\sigma'} \;
    \ket{j, m, \sigma}\bra{j', m', \sigma'}.
\end{equation}
Each block $\rho^k$ evolves independently under the master equation. Physically, the Bohr sectors correspond to coherences between states with fixed $j$, $j'$, and $\Delta_{J_z}$. This block-diagonal structure significantly simplifies both analytical treatment and numerical simulations of the dynamics.
For each Bohr sector $\rho^k$, we collect the associated matrix elements into a vector $\vec{p}^k$, with components defined as  
\[
\vec{p}^k_l = \rho_{j,m,\sigma,;j',m',\sigma'},
\]  
where the index $l$ enumerates all values consistent with a fixed pair $(j,j')$ and coherence order $\Delta_{J_z} = m - m'$. The components are arranged in order of increasing $m$, so that $l = 1, 2, \dots, d$. The dimension $d$ of the sector is given by  
\begin{equation}
    d = \min\!\left[\underbrace{\dim \mathcal{H}_J}_{2j+1},\, 
                      \underbrace{\dim \mathcal{H}_{J'}}_{2j'+1}\right] 
        - \max\!\left[|\Delta_{J_z}| - |j - j'|,\, 0\right].
\end{equation}
This decomposition allows the full density matrix to be written as a vector
\begin{equation}
    \vec{p} = \bigoplus_k \vec{p}^k,
\end{equation}
with each $\vec{p}^k$ evolving independently within its respective sector. The dynamics are then governed by a block-diagonal Liouvillian superoperator,
\begin{equation}
    \frac{d\vec{p}}{dt} = L \vec{p}, \qquad L = \bigoplus_k L^k,
\end{equation}
where $L^k$ is the linear operator generating the evolution in the $k$th sector. As we show below, the structure of the master equation~\eqref{ME} leads to two distinct types of sector dynamics:  
1. \textit{Thermalizing sectors}, where the system relaxes to a non-trivial stationary state.  
2. \textit{Leaking sectors}, where coherences decay irreversibly and the steady-state vanishes.
This classification follows directly from the form of $L^k$, which we now derive explicitly. For a sector $k = (j, j', \sigma, \sigma', \Delta_{J_z})$, the master equation~\eqref{ME} leads to coupled linear equations for the matrix elements $\rho_{j,m,\sigma;j'm',\sigma'}$ in that sector:
\begin{align}
    \frac{d}{dt}\rho_{j,m,\sigma;\, j',m',\sigma'} &= 
    d_m^k \,\rho_{j,m,\sigma;\, j',m',\sigma'} \nonumber \\ 
    &\quad + b_{m-1}^k \,\rho_{j,m-1,\sigma;\, j',m'-1,\sigma'} \nonumber\\
    &\quad + c_{m}^k \,\rho_{j,m+1,\sigma;\, j',m'+1,\sigma'}.
\end{align}
This structure implies that $L^k$ is tri-diagonal:
\begin{equation}\label{15NEW}
L^{k} =
\begin{pmatrix}
d_1^k & c_1^k & 0 & 0 & \cdots \\
b_1^k & d_2^k & c_2^k & 0 & \cdots \\
0 & b_2^k & d_3^k & c_3^k & \cdots \\
0 & 0 & b_3^k & d_4^k & \cdots \\
\vdots & \vdots & \vdots & \vdots & \ddots
\end{pmatrix}.
\end{equation}
The coefficients are given by
\begin{align}
    d_m^k &= - \tfrac{\gamma}{2}\Big[n\,(A_{m+1}^k + A_{m'+1}^{'k}) 
    + (n+1)\,(A_m^k + A_{m'}^{'k})\Big], \nonumber\\
    c_m^k &= \gamma (n+1)\sqrt{A_{m+1}^k A_{m'+1}^{'k}},\,
    b_m^k = \gamma n\sqrt{A_{m+1}^k A_{m'+1}^{'k}}, \nonumber
\end{align}
where $A_m^k$ and $A_{m'}^{'k}$ are determined by the collective coupling structure. To analyze the structure of the induced dynamics and distinguish between thermalizing and leaking sectors, we introduce the shorthand notation 
\begin{equation}
\begin{aligned}
    A_{m}^k &= (j + m)(j - m + 1),\\
    A_{m'}^{k'} &= (j' + m')(j' - m' + 1),
\end{aligned}
\end{equation}
which arises from the matrix elements of the collective ladder operators $J_\pm$ acting within the angular momentum basis.  
To build intuition, we first note that the off-diagonal elements of each generator matrix $L^k$ are non-negative. This ensures that the evolution governed by $L^k$ preserves positivity of the density matrix, as required for any valid quantum dynamical map.  
Next, consider the sum of the elements in each column of $L^k$. A key observation is that this sum is always non-positive and takes the form
\begin{align}\label{17NEW}
    &-\gamma(n+1) \left( \tfrac{A_m^k + A_{m'}^{k'}}{2} - \sqrt{A_m^k A_{m'}^{k'}} \right) \nonumber \\
    &\quad - \gamma n \left( \tfrac{A_{m+1}^k + A_{m'+1}^{k'}}{2} - \sqrt{A_{m+1}^k A_{m'+1}^{k'}} \right) \leq 0.
\end{align}
Above, the equality holds if and only if both conditions
\begin{align}
    A_m^k &= A_{m'}^{k'}, \quad 
    A_{m+1}^k = A_{m'+1}^{k'}
\end{align}
are satisfied simultaneously. This occurs precisely when the two subspaces are identical, i.e., $j = j'$ and $\Delta_{J_z} = m - m' = 0$. In all other cases, the column sums are strictly negative, which implies that the corresponding coherences decay in the long-time limit.  
Thus, only the sectors with $j = j'$ and $\Delta_{J_z} = 0$ can sustain non-decaying (stationary) components of the density matrix, corresponding to thermalizing dynamics. All other sectors are purely dissipative and vanish in the steady-state, identifying them as leaking sectors. 

To formalize this distinction, we focus on the column sums of the generator matrix $L^k$. A sector exhibits leakage if there exists at least one column of $L^k$, indexed by $w$, for which the sum of its elements,
\begin{equation}
    S(w) := \sum_i L^k_{i w},
\end{equation}
is strictly negative. In this case, the long-time evolution within the sector yields a vanishing steady-state,
\begin{equation}\label{20NEW}
    \lim_{t \rightarrow \infty} \vec{p}^k(t) = 0.
\end{equation}
The mechanism is visible from the tri-diagonal structure of $L^k$. Each component $\vec{p}^k_l$ couples only to its nearest neighbors $\vec{p}^k_{l \pm 1}$, so the total sector weight $\sum_l \vec{p}^k_l(t)$ evolves as
\begin{equation}\label{22NEW}
    \frac{d}{dt} \sum_l \vec{p}^k_l(t) = \sum_j S(j)\, \vec{p}^k_j(t).
\end{equation}
Because the dynamics is positivity-preserving, it does not mix real and imaginary parts of $\vec{p}^k$, and allows finite-time propagation of information across the sector, the presence of any column with $S(w) < 0$ guarantees that the total weight $\sum_l \vec{p}^k_l(t)$ decays to zero. Consequently, the entire sector leaks.  
Leakage is avoided only when the conditions $j = j'$ and $\Delta_{J_z} = 0$ are simultaneously satisfied. In these exceptional cases, $L^k$ corresponds either to (i) diagonal blocks within a fixed subspace labeled by $(j,\sigma)$, or (ii) coherences between degenerate subspaces with the same total spin $j$ but different degeneracy labels $\sigma \neq \sigma'$, provided the states share the same magnetic quantum number $m$.  
In such \textit{thermalizing sectors}, the column sums of $L^k$ vanish, and the coefficients satisfy
\begin{equation}\label{23NEW}
    d_m^k = c_{m-1}^k + b_m^k, 
    \qquad b_m^k = \alpha_c \, c_m^k,
\end{equation}
where $c_m^k$ and $b_m^k$ denote the de-excitation and excitation rates between neighboring $m$ levels, with their ratio fixed by the thermal Gibbs factor $\alpha_c= \frac{n_c}{(n_c+1)}$. This detailed balance condition leads to a unique steady-state distribution of the form
\begin{equation}\label{eq24}
    \lim_{t \to \infty} \vec{p}^k_l(t) 
    = \alpha_c^l \vec{p}^k_0(t),
\end{equation}
where $\vec{p}^k_0(t)$ is fixed by conservation of probability,  
\begin{equation}
    \sum_l \vec{p}^k_l(t) = \sum_l \vec{p}^k_l(t_0).
\end{equation}
A rigorous proof of the leakage behavior in non-thermalizing sectors is provided in Supplementary Note IIG.

To recast the above structure of the stationary state in a compact form, we invoke the Schur--Weyl duality \cite{fulton1991representation,goodman2009symmetry,KEYL2002431,georgi1983lie}. 
The Liouvillian in Eq.~\eqref{Liouvillian} commutes with the permutation $\mathcal{U}(\pi)$ of every qubit, i.e.
$\mathcal{U}(\pi)\,\mathcal L[\rho]\,\mathcal{U}(\pi)^\dagger=\mathcal L[\mathcal{U}(\pi)\rho\,\mathcal{U}(\pi)^\dagger]$ for all $\pi$ \cite{Shankar1994,sakurai2020modern,georgi1983lie}. Then, in the decomposition $\mathcal H^{(N)}\simeq\bigoplus_j(\mathcal{V}_j\otimes U_j)$, the generator necessarily takes the block form in the Schur basis
\begin{equation}
    \mathcal L \;=\;\bigoplus_{j}\Big(\,\mathbb{I}_{\mathcal{V}_j}\otimes \mathcal L_j\,\Big),
\label{eq:SW-block}
\end{equation}
i.e.\ it acts trivially on the multiplicity (degeneracy) space $\mathcal{V}_j$ and acts non-trivially only on the spin-$j$ carrier $U_j$.
In particular, for each fixed $j$ the reduced state $\chi_j(t)\coloneqq\mathrm{Tr}_{U_j}\rho_j(t)$ is conserved:
$\frac{d}{dt}\chi_j(t)=\mathrm{Tr}_{\mathcal{U}_j}[(\mathbb{I}\otimes\mathcal L_j)[\rho_j(t)]]=0$. Here $\mathcal{V}_j$ carries the irreducible representation (irrep) of the symmetric group $S_N$ (indexed by $\sigma$) with multiplicity $\nu_j$ and $U_j$ carries the $SU(2)$ irrep of dimension $2j+1$ (indexed by $m$). Using the structure \eqref{eq:SW-block} in the sector analysis in Eqs. \eqref{10NEW}-\eqref{15NEW} above, we conclude that only the sectors with $j=j'$ and $\Delta J_z=0$ may host stationary components, while all others leak (Eqs.~\eqref{17NEW}-\eqref{20NEW}). Within each surviving $j$-block, Eq.~\eqref{eq24} determines the unique steady-state $\tau_j$ on $U_j$ which can be presented in the full basis as:
\begin{equation}
    \mathbb{I}_{\nu_{j}} \otimes \tau_j\;=\; \frac{1}{\mathcal Z_j}\sum_{\sigma = 1}^{\nu_{j}}\sum_{m=-j}^{j} \alpha_{c}^{\,j+m}\,\ket{j,m, \sigma}\bra{j,m,\sigma},
\end{equation}
while off-diagonal coherences in $m$ decay. Therefore, the global steady-state factorizes in each $j$ sector as
\begin{equation}\label{SS}
    \rho_{\mathrm{SS}} \;=\; \bigoplus_{j}\Big(\chi_j\otimes \tau_j\Big),
\end{equation}
with $\chi_j\ge 0$, $\sum_j\mathrm{Tr}\,\chi_j=1$ determined by the initial projections onto $\mathcal{V}_j$
\begin{equation}
    \chi_j \otimes \mathbb{I}_{m}= \sum_{\sigma, \sigma' = 1}^{\nu_j} (\chi_j)_{\sigma,\sigma'} \sum_{m = -j}^{j} \ket{j, m, \sigma}\bra{j, m, \sigma'}.
\end{equation}
More specifically, the only remaining information on the initial state of the system $\rho(t_{0})$ is encapsulated in the matrix through the following matrix elements: $(\chi_{j})_{\sigma,\sigma'} = \sum_{m=-j}^{j}\bra{j,m,\sigma'}\rho(t_0)\ket{j,m,\sigma}$ (see Supplementary Note II).

As will be seen in the following the protection of intra-degeneracy coherence provides the microscopic mechanism that enables work extraction from an initially thermal state.

Finally, we emphasize that the collective Liouvillian preserves permutation symmetry and does not mix different total-spin sectors $j$. This non-ergodicity ensures that interference-enabled structures—such as superradiant and subradiant pathways, as well as degeneracy-protected coherences within fixed-$j$ manifolds—survive at stationarity. Even at high reservoir temperatures, the dynamics remain confined to collective ladders; there is no symmetry-breaking “reservoir microscope” capable of washing out these coherences. Note however that local noise is capable of breaking this symmetry, leading to a unique stationary state, as shown in Section: Robustness to local noise. 
We are now well-equipped to numerically compute the steady-state ergotropy for the Gibbs-product initial condition up to a significant high number of qubits.

\subsection*{Ergotropy of the QB}

Having established the charging process of our QB via the collective dissipative model and its symmetry constraints, we now turn to the central experimental motivation of this work. In realistic settings, the most readily available initial states are thermal (Gibbs) states, which are also the most passive states: no work can be extracted from them through unitary dynamics alone. Guided by this consideration, we take as initial state the product of local Gibbs states,
\begin{equation}
    \rho(t_{0})= \rho_{\beta_{\mathrm{q}}}^{\otimes N},
    \qquad
    \rho_{\beta_{q}}=
\frac{1}{Z}\begin{pmatrix}
    1 & 0\\
    0 & e^{-\beta_{q}}
\end{pmatrix}, \quad
    q = e^{-\beta_{\mathrm{q}}},
\end{equation}
which contains no ergotropy at $t=t_0$. Despite this passivity, the collective coupling to a single bosonic reservoir autonomously activates ergotropy.
The steady-state ergotropy $\mathcal{W}$ for $N = 26$ qubits (illustrated in Fig.~\ref{Fig2}) shows a broad parameter regime where the steady-state exhibits non-zero ergotropy, even when the reservoir temperature is finite. In addition, increasing the temperature of the reservoir increases the amount of activated ergotropy. In this plot, the local temperature parameter is defined $q = e^{-\beta_{\mathrm{q}}}$. For clarity, the axes are arranged as $(q,\alpha_c)$.  

Two main trends emerge:
(i) Even a zero-temperature reservoir ($\alpha_c \to 0$) activates the ergotropy for $N>1$ provided the local temperature parameter $q$ is sufficiently high, with the activated ergotropy increasing with the size of the system (see Supplementary Note III).  
(ii) Analytically, in the thermodynamic limit the activation of ergotropy becomes {\it generic} (see Supplementary Note IV):
\begin{equation}
    \lim_{N\to\infty}\mathcal{W}(N,\alpha_c, q) > 0,
    \qquad \text{for all} \quad \alpha_c \neq q,
\end{equation}
demonstrating that only the fine-tuned line $\alpha_c = q$ yields vanishing ergotropy.
This establishes the key result: starting from the most accessible and fully passive product Gibbs state, collective dissipation alone suffices to generate a steady-state with finite ergotropy, whose magnitude grows systematically with $N$.
\begin{figure}[t]
\center
\includegraphics[width=1\columnwidth]{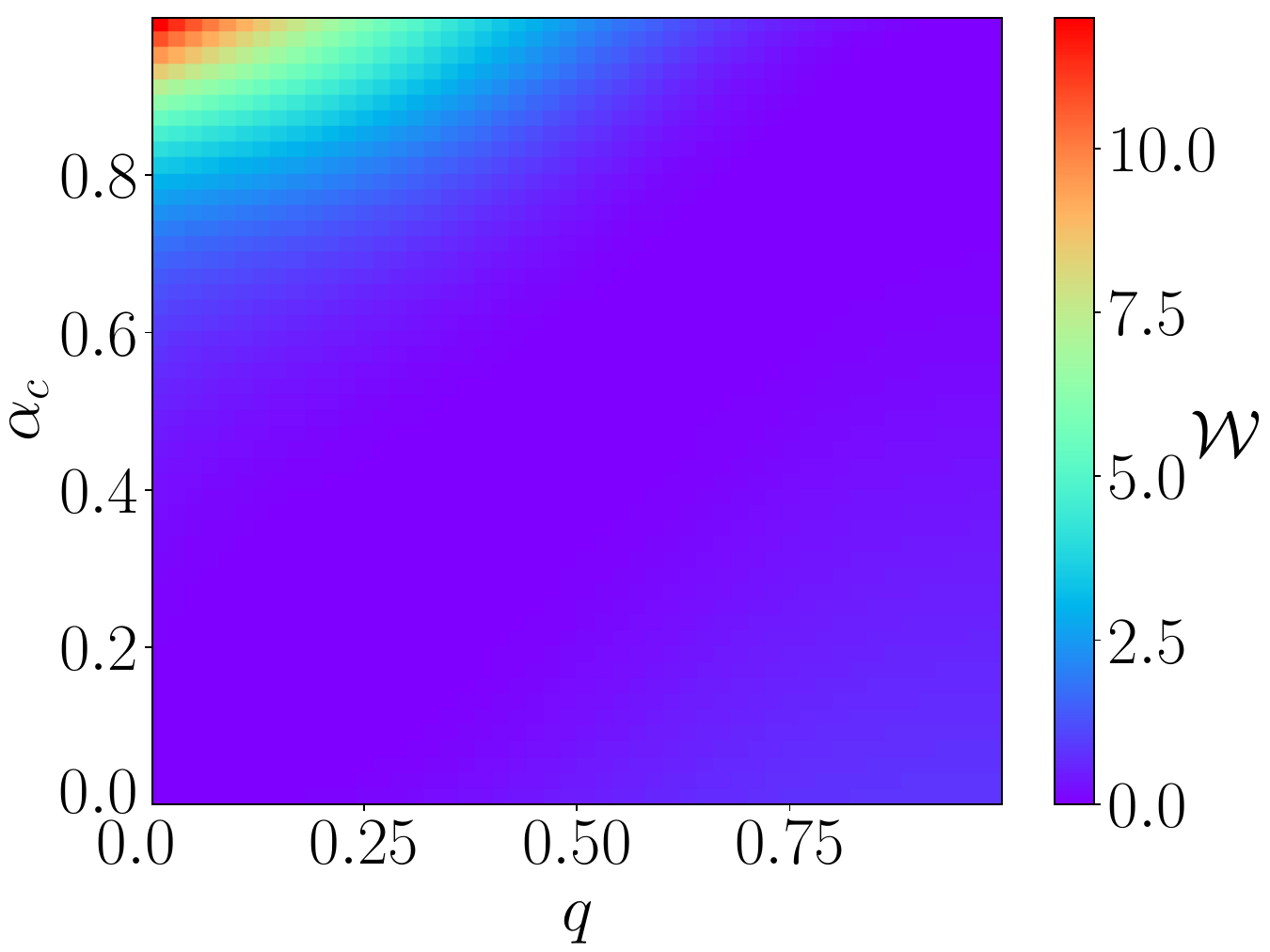}
\caption{\textbf{Steady-state ergotropy under ideal collective dissipation.}
Heat map of the steady-state ergotropy $\mathcal{W}$ for $N=26$ with perfectly collective coupling ($\eta=1$) versus the common-reservoir parameter $\alpha_c$. Away from the fine-tuned line $\alpha_c=q$, $\mathcal{W}>0$ and typically grows as either $\alpha_c\to 1$ or $q\to 1$.\justifying}
\label{Fig2}
\end{figure}
\begin{figure}[t]
    \centering
    \begin{subfigure}{0.49\textwidth}
        \centering
        \includegraphics[width=1\textwidth]{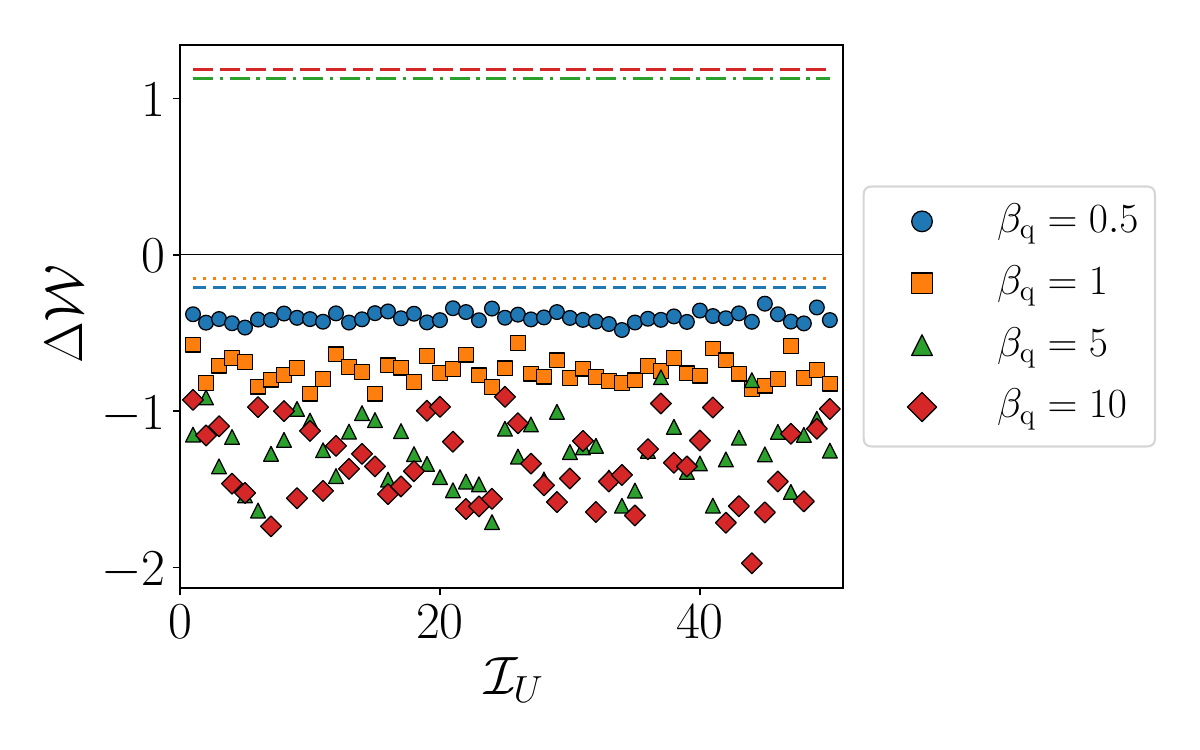}
        \caption{Ergotropic balance for $\beta_{c} = 0.01$}
        \label{DeltaW0.01}
    \end{subfigure}
    \hfill
    \begin{subfigure}{0.49\textwidth}
        \centering
        \includegraphics[width=1\textwidth]{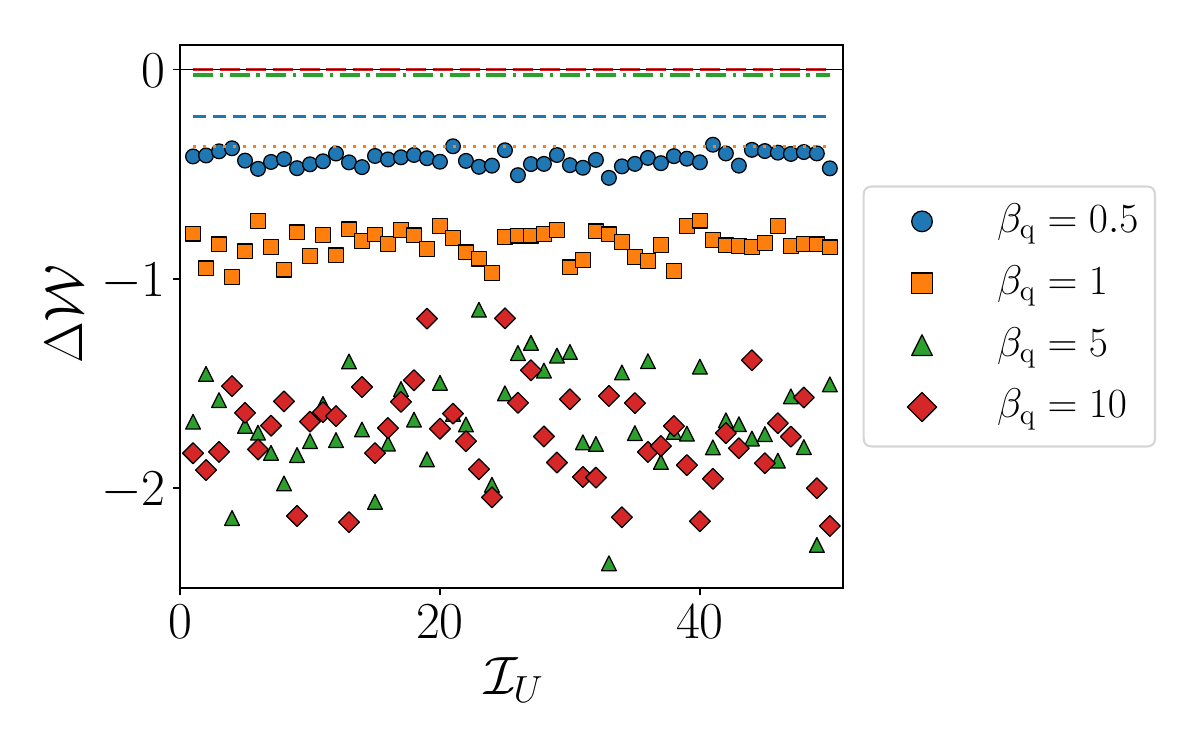}
        \caption{Ergotropic balance for $\beta_{c} = 10$}
        \label{DeltaW010}
    \end{subfigure}
    \caption{Comparison of the ergotropic balance $\Delta \mathcal{W}$ after applying random unitaries. $\Delta \mathcal{W}$ is plotted against the ordinal index of the unitary, denoted by $\mathcal{I}_{U}$, for different system temperatures $\beta_{q}$ and two environment temperatures $\beta_{c}$. In all cases, the ergotropic balance satisfies $\Delta \mathcal{W} < 0$. \justifying}
    \label{DeltaW}
\end{figure}
%

As a particular case, our QB can be initialized fully uncharged, i.e., ground state $|0\rangle^{\otimes N}$. As discussed above for initial states diagonal in the Dicke basis $\rho(t_0)=\sum_{i=0}^{N}p_{i}(t_0)\dicke{N}{i}\dickeD{N}{i}$, the equation of motion \eqref{ME} does not generate any coherences with respect to Dicke basis and since the ground state lies within the block with highest $j=\frac{N}{2}$, the evolution remains within the same block.  Consequently, at any given time, the state can be expressed as $\rho(t)=\sum_{i=0}^{N}p_{i}(t)\dicke{N}{i}\dickeD{N}{i}$. The ergotropy of the stationary state, in this case, reads (see Supplementary Note V) \cite{PhysRevE.102.042111}
\begin{align*}
    \mathcal{W} = N + \frac{\alpha_c}{1-\alpha_c}+\frac{N+\alpha_c}{\alpha_c^{1+N} - 1}.\numberthis
\end{align*}
In infinite temperature limit, $n\rightarrow\infty$ ($q\rightarrow 1$), we have $\lim_{q\rightarrow 1}\mathcal{W} = \frac{N}{2}+\frac{1}{N+1}-1$.

Now we address the problem of identifying the optimal initial state for ergotropy distillation, providing numerical evidence suggesting that not only preparing the system in a product Gibbs state may be experimentally friendly, but also optimal from the point of view of ergotropy extraction.  We initialize the system in a product of Gibbs states of fixed temperature, and apply a random global unitary in order to obtain a random initial state. We compare the ergotropy in the stationary state, diminished by the energetic cost of the unitary rotation, with ergotropy available without the rotation. Numerical analysis of the 4 qubit system case, as shown in Fig. \ref{DeltaW}, suggests that a product of Gibbs states is optimal, as no energetic benefit can be observed from a random unitary (with its energetic cost taken into account). We note that 4 qubit systems can exhibit coherent terms in the stationary state when initialized in a product state and rotated unitarily, therefore the whole structure Eq.~\eqref{SS} was exploitede in the analysis.  
Further details of this analysis are presented in Supplementary Note VI. 
\begin{figure}[t]
\center
\includegraphics[width=1\columnwidth]{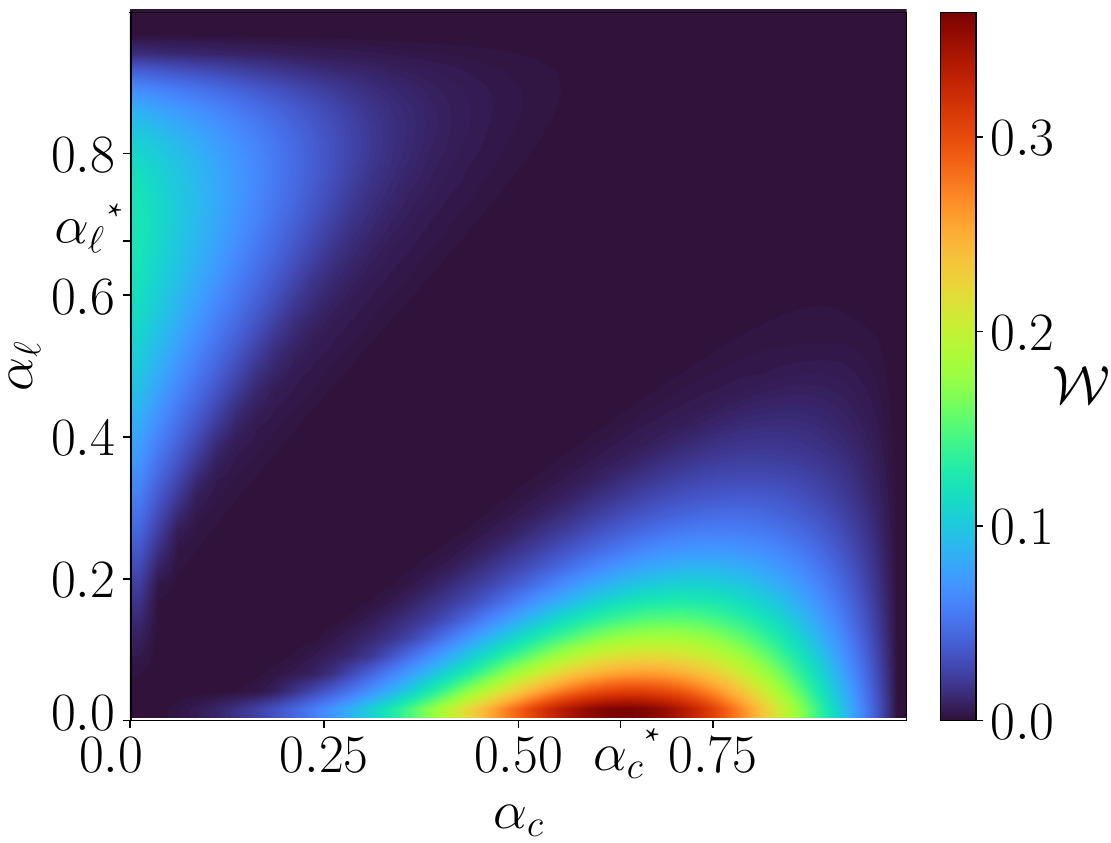}
\caption{\textbf{Finite-temperature optima under partial collectivity.}
Steady-state ergotropy $\mathcal{W}$ for $N=7$, $\gamma_r=1$ with collective fraction $\eta=0.9$ as a function of $(\alpha_c,\alpha_\ell)$. Unlike the ideal case, the maximum occurs at a finite $\alpha_c^\star<1$ and a non-zero $\alpha_\ell^\star>0$, reflecting a trade-off between interference-enabled collective pumping (favored by larger $\alpha_c$) and which-path information injected by local channels (growing with $\alpha_\ell$) that degrades interference.\justifying}
\label{Fig3}
\end{figure}

\subsection*{Robustness to local noise}

In realistic devices, deviations from perfectly collective coupling arise from inhomogeneous qubit–resonator couplings, spatial phase variations, and additional local decay channels. To capture these effects, we parametrize this using a collective coupling fraction $\eta\in[0,1]$, defined such that $\eta=1$ corresponds to perfectly collective dissipation (all qubits coupling identically to the reservoir), while $\eta=0$ corresponds to purely local dissipation with no collective enhancement. We then consider the GKLS generator \cite{breuer2002theory,audretsch2007entangled,gardiner2004quantum}
\begin{equation}\label{Leta}
    \mathcal{L}_\eta \;=\; \eta\,\mathcal{L}_{\mathrm{coll}} \;+\; (1-\eta)\,\mathcal{L}_{\mathrm{local}},
\end{equation}
where $\mathcal{L}_{\mathrm{coll}}$ is the ideal collective dissipator (built from $J_\pm$), and $\mathcal{L}_{\mathrm{local}}$ is a sum of independent single-qubit dissipators. Note that because GKLS generators form a convex cone, any positive linear combination such as Eq. \eqref{Leta} is again a valid physical GKLS generator in the weak-coupling, factorized derivation of the master equation when the qubits are on resonant with each other \cite{Mitchison_2018,Hofer_2017,Gonzalez2017,PhysRevA.97.062124}, and it physically corresponds to simultaneous coupling to two independent reservoirs \cite[Sec.~3.2]{breuer2002theory}. We take the bosonic reservoir parameters in the standard form $\alpha_x=\frac{n_x}{(n_x+1)}$ with $x\in\{c,\ell\}$ denoting the collective and local channels, respectively.

Deviations from perfect collectivity qualitatively reshape the landscape. In Fig.~\ref{Fig3} we observe a finite-temperature optimum in both control parameters: for $N=7$, $\gamma_r\equiv\frac{\gamma_\ell}{\gamma_c}=1$ at fixed $\alpha_\ell$ there exists an $\alpha_c^\star<1$ that maximizes $\mathcal{W}$, whereas at fixed $\alpha_c$ the optimum shifts to a non-zero $\alpha_\ell^\star>0$. The underlying physics of this counterintuitive effect is detailed in the following. Physically, $\alpha_c$ governs the strength of collective pumping along the irrep of $j$, where indistinguishable paths interfere constructively (superradiant channels) and degeneracy-protected coherences survive; this interference structure is the resource that makes ergotropy distillable at stationarity. However, the local channel (weighted by $1-\eta$ and parametrized by $\alpha_\ell$) provides which-path information that erodes these interference effects. Consequently, when $\eta<1$ there is a trade-off: too small $\alpha_c$ underpopulates the symmetric subspace, while too large $\alpha_c$ in the presence of local noise overheats and undermines the interference-protected structure; likewise, a small but non-zero $\alpha_\ell$, with the help of interference caused by the collective reservoir, can facilitate transport across excitation sectors, but excessive local noise suppresses $\mathcal{W}$. The activation lobes in Fig.~\ref{Fig3} at moderate–high $\alpha_c$ and small–moderate $\alpha_\ell$ captures precisely this compromise. Numerics, and the analysis that we refer to Supplementary Note VII reveal a robust activation lobe along the $\alpha_c$-axis in Fig.~\ref{Fig3}: at fixed $\alpha_\ell$, the common reservoir generates essentially a constant amount of ergotropy that persists for all $0<\eta<1$, provided any additional dephasing remains below the collective rate $\gamma_c$. Equivalently, the ergotropy produced by the collective channel at fixed $\alpha_\ell$ is only weakly dependent on the collective fraction $\eta$, demonstrating robustness against disorder and noise. The activation point shifts with collectivity—lower $\eta$ requires higher $\alpha_c$ (higher temperatures)—but for $\eta\gtrsim0.8$ the optimal value $\alpha^*_c$ is only weakly affected  (see Supplementary Note VIIB). We further observe a clear activation lobe along the $\alpha_\ell$-axis at $\alpha_c=0$: for $N\gtrsim 5$ and $\eta\gtrsim 0.5$, a small–to–moderate local temperature ($\alpha_\ell$) robustly activates ergotropy. The activation magnitude increases with $N$, consistent with stronger collective interference within the battery that facilitates energy transfer.

So, why “overheating” is harmless at $\eta=1$ but harmful when $\eta<1$?
The short answer is that overheating only hurts when there exists a channel that turns energy into which-path information. In the perfectly collective limit ($\eta=1$) all dissipation proceeds through the permutation-symmetric jumps $J_\pm=\sum_i\sigma_\pm^{(i)}$. Because the environment never resolves which qubit emitted or absorbed an excitation, the dynamics retains indistinguishability at all temperatures $\alpha_c$, and the interference structure encoded by the bright/dark Dicke channels is preserved. Increasing $\alpha_c$ merely enhances upward transitions (superabsorption) within irreps, without introducing any which-path information that could wash out phases. By contrast, when a local component is present ($\eta<1$), the jumps $\sigma_\pm^{(i)}$ act as which-path–resolving events: they mix total-spin sectors, restore primitivity (unique steady-state), and progressively suppress the super/subradiant contrast. In this regime, making the common reservoir too hot amplifies the rate at which the extra excitations are processed by local channels, thereby undermining interference-protected structure and reducing ergotropy—hence the finite-temperature optimum in Fig.~\ref{Fig3} (see Supplementary Note VIII for analytical detail).

A further observation of the steady-state ergotropy, as presented in detail in Supplementary Note VIIC, is that for weak local dissipation ($\gamma_\ell\ll\gamma_c$), the activation lobe along the $\alpha_\ell$-axis is broad and resilient: suppression sets in only at relatively large local temperatures $\alpha_\ell$; increasing $\gamma_r$ quenches this lobe already at moderate $\alpha_\ell$, while the position and width of the $\alpha_c$ lobe are nearly unchanged. This asymmetry is analytically explained in the Supplementary Note VIII: local, which-path jumps drain weight from the bright sector at a rate $\propto \gamma_\ell(2n_\ell{+}1)$, suppressing the $\alpha_\ell$ lobe, whereas the collective detailed balance that sets the $\alpha_c$ optimum remains largely intact.

As for the transient charging dynamics near and away from the optimum,
Fig.~\ref{ErgoOpt} shows the transient at the optimal collective temperature $\alpha_c^\star$ (fixed $\alpha_\ell$), where the steady-state carries non-zero ergotropy. For $\gamma_c t\!\ll\!1$ the traces nearly overlap: indistinguishable collective jumps $J_\pm$ coherently pump population up the irreps, so the initial rise of $\mathcal W(t)$ is essentially insensitive to $\gamma_r$. At later times, site–resolved jumps $\sigma_\pm^{(i)}$ inject which-path information and gradually erode this interference; larger $\gamma_r$ speeds up the wash-out and slightly lowers the plateau, whereas small $\gamma_r$ leaves a long interference-dominated window before relaxation. (A detailed $N$–dependent study is given in Supplementary Note IX.)

For comparison, Fig.~\ref{Fig4} reports a non-optimal point ($\alpha_c=0.9,\ \alpha_\ell=0.5$) where the steady-state is passive, $\mathcal W(\infty)=0$. The transient still displays the same mechanism—an interference-enabled rise followed by wash-out—but with a striking difference in scale: the optimal case saturates at a robust steady value $\mathcal W_{\rm SS}\!\approx\!0.55$ (Fig.~\ref{ErgoOpt}), whereas the non-optimal case exhibits a large overshoot up to $\mathcal W_{\rm peak}\!\approx\!2.5$ (Fig.~\ref{Fig4}) before decaying to zero. The peak is produced by rapid collective pumping that briefly outpaces local which-path mixing; it is not stationary. However, harvesting such a transient statee would require time-synchronized extraction, while operation near $\alpha_c^\star$ provides a repeatable, timing-insensitive steady ergotropy set by collective detailed balance.
\begin{figure}[t]
\includegraphics[width=1\columnwidth]{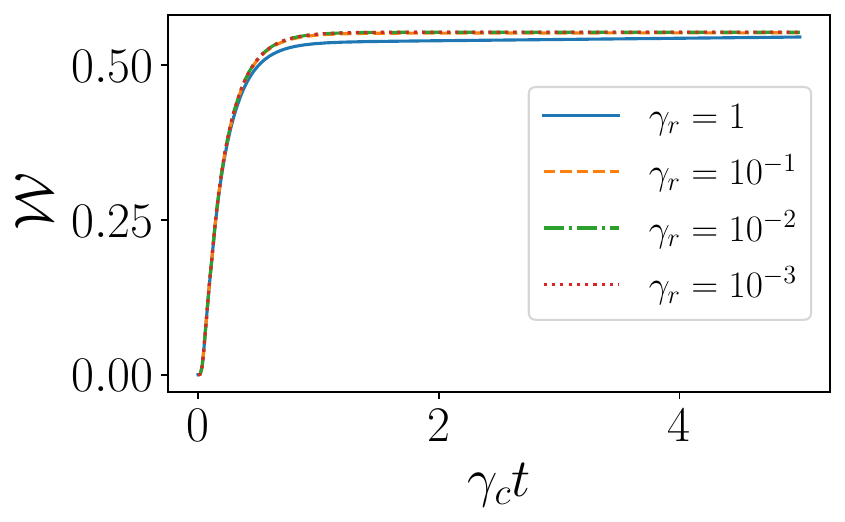}
    \captionsetup{justification=justified}
    \caption{\textbf{Charging dynamics near the activation point.}
    Ergotropy $\mathcal W(t)$ versus rescaled time $\gamma_c t$ for $N=10$, $\eta=0.9$ and various dissipation ratios $\gamma_r$. The collective bath is set close to its optimal value $\alpha_c^\star$; $\alpha_\ell=0$ is fixed across curves.\justifying} 
    \label{ErgoOpt}
\end{figure}
\begin{figure}[t]
\center
\includegraphics[width=1\columnwidth]{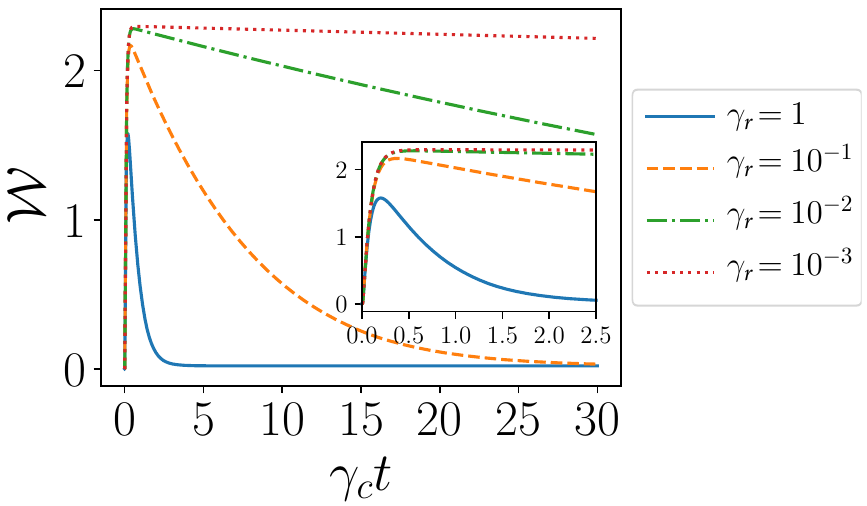}
\caption{\textbf{Transient charging and memory wash-out.}
Time evolution of ergotropy $\mathcal{W}(t)$ for $N = 10, \eta = 0.6$, $\alpha_c = 0.9$, $\alpha_\ell = 0.5$ (non-optimal operating point) for several ratios $\gamma_r$. For small $\gamma_r$ the interference-erasure timescale is long, so the collective channel dominates over an extended window before eventual relaxation.\justifying}
\label{Fig4}
\end{figure}

\subsection*{Discussion}

The present results identify a pathway for charging a QB from a fully passive initial state using only collective coupling to a thermal reservoir. A central question is whether such a mechanism can be realized experimentally with present-day platforms. Circuit QED \cite{RevModPhys.93.025005} provides a particularly promising avenue. Ensembles of superconducting transmon qubits \cite{PhysRevLett.103.083601,krinner2019engineering} coupled to a common microwave resonator routinely achieve collective coupling strengths $\frac{g}{(2\pi)}$ in the range of $10$–$50\,\mathrm{MHz}$, with relaxation rates $\frac{\gamma}{(2\pi)} \sim 0.1$–$1\,\mathrm{MHz}$ \cite{PhysRevApplied.17.044016,KjaergaardAnnual,mlynek2014observation}. The thermal occupation of the resonator mode can be tuned simply by injecting incoherent microwave noise at the resonant frequency, thereby controlling the reservoir parameter $\alpha_c = \frac{n_c}{(n_c+1)}$ over a broad range \cite{PhysRevResearch.6.023262,PhysRevX.10.041054}.

A practical experiment would proceed by first preparing each qubit in a local Gibbs state at inverse temperature $\beta_{\mathrm{q}}$ using engineered thermal reservoirs. The ensemble is then collectively coupled to the noisy resonator for a controlled interaction time, after which the stored energy $E$ can be reconstructed using either quantum state tomography or energy-resolved population measurements. Work extraction can be demonstrated by applying an approximate optimal unitary, implemented through sequences of microwave pulses, to transfer the stored energy into a designated load mode or auxiliary resonator. The difference in measured energy before and after this unitary provides a direct measure of the accessible ergotropy.  

Although our analytical framework focuses on the steady-state, numerical simulations of Eq.~\eqref{ME} show that ergotropy activation occurs on timescales of order $(\gamma N)^{-1}$ in the strongly collective regime. For typical circuit QED parameters, this corresponds to charging within microsecond to sub-microsecond windows, yielding power densities competitive with the fastest QB proposals to date. The absence of a coherent drive makes the protocol particularly attractive for integration into autonomous devices: the same noisy or thermal field that carries a signal—for example in a quantum sensor—could simultaneously charge the QB, which then powers downstream tasks such as readout or error correction. This dual role of the environment is a distinctive advantage in uncontrolled or resource-limited settings, ranging from cryogen-free laboratories to space-based platforms.  

At a conceptual level, the underlying physics of the scheme rests on interference-enabled dissipation. Collective jumps through the operators $J_\pm$ make emission and absorption events indistinguishable, producing super- and subradiant interference in the Dicke basis and preserving coherences within degenerate $j$-manifolds at stationarity. These protected coherences, inaccessible under purely local dissipation, allow a purely thermal reservoir to sustain non-passive steady states with finite ergotropy. Local noise acts as a which-path marker, restoring primitivity of the semigroup and progressively reducing both transient and steady-state ergotropy. The interplay between collective and local channels generates the finite-temperature optima and “memory wash-out’’ behavior observed in Figs.~\ref{Fig3}–\ref{Fig4}.  

For perfectly collective coupling ($\eta = 1$), raising the reservoir temperature enhances upward transitions along the same collective ladders without introducing which-path information, so interference survives and ergotropy can even grow. When $\eta < 1$, local noise mixes total-spin sectors and reduces the contrast between super- and subradiant pathways. The result is a competition: insufficient collective pumping underpopulates bright ladders, while excessive heating in the presence of local noise enhances dephasing and phase randomization. This trade-off yields finite-temperature optima for both common- and local-reservoir parameters. Transiently, $\mathcal{W}(t)$ grows as interference builds and bright ladders fill, but as $\gamma_r$ increases, memory of the initial state is erased more rapidly and both the peak and long-time ergotropy decrease. This “memory wash-out’’ provides a quantitative signature of the erosion of interference by which-path noise.

\textit{Applications—DFSs and QEC.}
Collective coupling naturally supports decoherence-free (and subradiant) subspaces, while the same environment autonomously charges a nearby work buffer. This suggests a hardware-efficient layer for fault tolerance: (i) rapid ancilla reset using the locally stored work \cite{Geerlings2013DDROP,Magnard2018AllMicrowaveReset,Zhou2021ParametricReset}; (ii) dissipative stabilizer mechanisms
\cite{Verstraete2009DissipationDriven,Barreiro2011OpenSystemSimulator,Mueller2011StabilizerPumping,Reiter2016DissipativeEntanglement}; and (iii) stabilization of code spaces/DFSs by operating within collectively protected manifolds \cite{Zanardi1997NoiselessCodes,Lidar1998DFSforQC,Bacon2000UFTonDFS,Altepeter2004TwoQubitDFS,Dicke1954Coherence,Gross1982SuperradianceReview,Guerin2016Subradiance}. Any path of this type  needs a proper control of the time scales in the designed dynamics. 
The robustness we observe
against partial collectivity and dephasing supports this perspective especially  in the last part of the analysis where the observed transient memory wash-out can serve as a   time-scale guide for reset and readout windows. 
Compatibility with circuit- and cavity-QED platforms is well established \cite{Shankar2013AutonomousBell,Leghtas2015TwoPhotonDissipation,Magnard2018AllMicrowaveReset,RevModPhys.96.031001,Campaioli2018QBChapter}.

In summary, these results show that interference can serve as a genuine thermodynamic resource. The incoherent-charging paradigm we propose is compatible with cavity and waveguide QED, atomic arrays, and superconducting circuits, and integrates naturally with reservoir engineering techniques. By exploiting indistinguishability rather than initial coherence as the enabling resource, this approach provides a practical and noise-resilient route to scalable quantum batteries powered directly by simple thermal environments.

\textit{Acknowledgments.-}
BA, ABR and PH acknowledge support from IRA Programme (project no. FENG.02.01-IP.05-0006/23) financed by the FENG program 2021-2027, Priority FENG.02, Measure FENG.02.01., with the support of the FNP. S.B. acknowledges funding by the Natural Sciences and Engineering Research Council of Canada (NSERC) through its Discovery Grant, funding and advisory support provided by Alberta Innovates through the Accelerating Innovations into CarE (AICE) -- Concepts Program, and support from Alberta Innovates and NSERC through Advance Grant. PM acknowledges support by the Polish National Agency for Academic Exchange (NAWA), under Strategic Partnerships Programme, project number BNI/PST/2023/1/00013/U/00001.


\clearpage
\onecolumngrid  

\begin{center}
\textbf{\large Supplementary Note for "Harnessing Environmental Noise for Quantum Energy Storage"}
\end{center}

\section{Microscopic system--bath coupling and conditions for collectivity}
\label{sec:SM_collective}

\subsection{Hamiltonian and mode functions}
We consider $N$ identical two-level systems (TLS) at fixed positions $\{\mathbf r_i\}_{i=1}^N$, transition frequency $\omega_0$, and transition dipole
$\mathbf d=\langle 1|\hat{\mathbf d}|0\rangle$, coupled to a bosonic thermal reservoir at temperature $T_R$.
Here $\hat{\mathbf d}$ denotes the electric-dipole operator and the vector $\mathbf d$ is the corresponding transition matrix element. For consistency we drop hats on operators henceforth, and all quantities without hats should be understood as operators unless explicitly stated otherwise, whereas $\mathbf d$ denotes the transition dipole matrix element defined above.
In the dipole approximation, the light--matter interaction is
\begin{equation}
H_\mathrm{I} \;=\; -\sum_{i=1}^N \mathbf d^{(i)}\!\cdot\!\mathbf E(\mathbf r_i,t)
\;=\; -\sum_{i=1}^N \big(\mathbf d\,\sigma_+^{(i)}+\mathbf d^{\,*}\sigma_-^{(i)}\big)\!\cdot\!
\big(\mathbf E^{(+)}(\mathbf r_i,t)+\mathbf E^{(-)}(\mathbf r_i,t)\big),
\label{eq:SM:HI:dipole}
\end{equation}
where $\mathbf E=\mathbf E^{(+)}+\mathbf E^{(-)}$ is the electric-field operator in the interaction picture with respect to the free-field Hamiltonian, and
$\mathbf E^{(-)}=(\mathbf E^{(+)})^\dagger$.
The TLS operators are defined such that $\sigma_+^{(i)}|0\rangle_i = |1\rangle_i$ and $\sigma_-^{(i)}|1\rangle_i = |0\rangle_i$.
Note that we assume $\hbar = 1$ as in the main text.
In a mode expansion with quantization volume $V$ (assuming periodic boundary conditions),
\begin{equation}
\mathbf E^{(+)}(\mathbf r,t)
= i\sum_{\mathbf k,\lambda}\sqrt{\frac{\omega_k}{2\varepsilon_0 V}}\;
\mathbf e_{\mathbf k\lambda}\,u_{\mathbf k\lambda}(\mathbf r)\,a_{\mathbf k\lambda}\,e^{-i\omega_k t},
\qquad
[a_{\mathbf k\lambda},a_{\mathbf k'\lambda'}^\dagger]=\delta_{\mathbf k\mathbf k'}\delta_{\lambda\lambda'},
\label{eq:SM:Eplus}
\end{equation}
with polarization unit vectors $\mathbf e_{\mathbf k\lambda}$ and mode functions $u_{\mathbf k\lambda}(\mathbf r)$ (for plane waves, $u_{\mathbf k\lambda}(\mathbf r) = e^{i\mathbf k\cdot\mathbf r}$).
Applying the rotating-wave approximation (RWA) to Eq.~\eqref{eq:SM:HI:dipole} gives
\begin{equation}
H_\mathrm{I}
= \sum_{\mathbf k,\lambda}\hbar g_{\mathbf k\lambda}
\Bigg[\sum_{i=1}^N \sigma_+^{(i)}\,u_{\mathbf k\lambda}(\mathbf r_i)\Bigg] a_{\mathbf k\lambda}
+\text{h.c.},
\qquad
g_{\mathbf k\lambda}
= -i\sqrt{\frac{\omega_k}{2\varepsilon_0\hbar V}}\;\mathbf d\!\cdot\!\mathbf e_{\mathbf k\lambda},
\label{eq:SM:HI:RWA}
\end{equation}
which is the standard starting point for the quantum--statistical derivation of the Born--Markov--secular master equation \cite{breuer2002theory,Agarwal1974}. The thermal reservoir at temperature $T_R$ is specified by
$\langle a_{\mathbf k\lambda}^\dagger a_{\mathbf k'\lambda'}\rangle
= n(\omega_k,T_R)\,\delta_{\mathbf k\mathbf k'}\delta_{\lambda\lambda'}$,
with $n(\omega,T_R)=1/(e^{\omega/(k_BT_R)}-1)$ the Bose--Einstein distribution. For $T_R = 0$, we have $n(\omega,0)=0$, corresponding to the vacuum state with spontaneous emission only.

\subsection{Born--Markov--secular master equation at finite temperature}
Starting from the RWA interaction Hamiltonian in Eq.~\eqref{eq:SM:HI:RWA}, tracing out the reservoir under the standard Born--Markov and secular approximations \cite{breuer2002theory} yields the master equation
\begin{align}
\dot\rho
&=-i\,[H_S,\rho]
+\sum_{i,j=1}^N \gamma^{(\downarrow)}_{ij}\!\left(\sigma^{(i)}_-\rho\,\sigma^{(j)}_+ -\tfrac12\{\sigma^{(j)}_+\sigma^{(i)}_-,\rho\}\right)
+\sum_{i,j=1}^N \gamma^{(\uparrow)}_{ij}\!\left(\sigma^{(i)}_+\rho\,\sigma^{(j)}_- -\tfrac12\{\sigma^{(j)}_-\sigma^{(i)}_+,\rho\}\right),
\label{eq:SM_ME}
\end{align}
where $H_S = \omega_0\sum_{i=1}^N \sigma_{+}^{(i)}\sigma_{-}^{(i)}$ is the free system Hamiltonian. The rates $\gamma^{(\downarrow)}_{ij}$ and $\gamma^{(\uparrow)}_{ij}$ describe collective emission and absorption, respectively, with $\gamma^{(\downarrow\uparrow)}_{ii}$ the single-qubit rate and the off-diagonal elements $\gamma^{(\downarrow\uparrow)}_{ij}=\gamma^{(\downarrow\uparrow)}_{ji}$ ($i \neq j$) encoding interference effects. The secular approximation eliminates terms oscillating at frequencies $\pm 2\omega_0$, resulting in the separated emission and absorption channels. These rates are given by
\begin{equation}
    \gamma^{(\downarrow)}_{ij}
    = \frac{2}{\hbar^{2}}\, \operatorname{Re}\!\int_{0}^{\infty}\! d\tau\, e^{i\omega_0\tau}\, \sum_{\alpha,\beta\in\{x,y,z\}}\big\langle E^{(+)}_{\alpha}(\mathbf r_i,\tau)\,E^{(-)}_{\beta}(\mathbf r_j,0)\big\rangle\, d_{\alpha} d_{\beta},\qquad
\gamma^{(\uparrow)}_{ij} = e^{-\hbar\omega_0/(k_BT_R)}\, \gamma^{(\downarrow)}_{ij},\label{eq:Gammaij}
\end{equation}
where $d_\alpha$ is the $\alpha$-th Cartesian component of the transition dipole moment $\mathbf{d}$, and the real part ensures the rates are real and symmetric. The second relation is the detailed balance condition. For the thermal reservoir specified in the previous section, we have \cite{damanet2016competition}
\begin{equation}
\gamma^{(\downarrow)}_{ij}=(n(\omega_0,T_R)\!+\!1)\,g_{ij},\qquad
\gamma^{(\uparrow)}_{ij}=n(\omega_0,T_R)\,g_{ij},
\label{eq:SM_updown}
\end{equation}
where $n(\omega_0,T_R)$ is the Bose--Einstein distribution evaluated at the transition frequency, and $g_{ij}$ is the vacuum contribution defined by
\begin{equation}
g_{ij} = 2\, \operatorname{Re}\!\int_{0}^{\infty}\! d\tau\, e^{i\omega_0\tau}\, \sum_{\alpha,\beta\in\{x,y,z\}}\big\langle E^{(+)}_{\alpha}(\mathbf r_i,\tau)\,E^{(-)}_{\beta}(\mathbf r_j,0)\big\rangle_{\text{vac}}\, d_{\alpha} d_{\beta}.
\label{eq:gij_def}
\end{equation}
All spatial dependence of the dissipator is encoded in the real symmetric matrix $g_{ij}$. The construction and thermal generalization in Eqs.~\eqref{eq:SM_ME}--\eqref{eq:SM_updown} follow the quantum-statistical projection-operator route summarized in Ref.~\cite{Agarwal1974}.

\subsection{Free-space kernels and asymptotic regimes}

For an electromagnetic reservoir in free space, $g_{ij}$ depend only on $x_{ij}\equiv k_0 r_{ij}$ with $r_{ij}=|\mathbf r_i-\mathbf r_j|$, $k_0=\omega_0/c$, and on the angle $\alpha_{ij}$ between $\mathbf r_{ij}$ and $\mathbf d$.
Using the notation of Ref. \cite{damanet2016competition}, one has (with $g_{ii}=\gamma_0$ the single-emitter rate)
\begin{align}
g_{ij} &= \frac{3\gamma_0}{2}\Big[(1-3\cos^2\alpha_{ij})\Big(\frac{\cos x_{ij}}{x_{ij}^2} - \frac{\sin x_{ij}}{x_{ij}^3}\Big)
+ (1-\cos^2\alpha_{ij})\,\frac{\sin x_{ij}}{x_{ij}}\Big]. \label{eq:SM_gij}
\end{align}
with $\gamma_0 = (\omega_0^3 d_{eg}^2)/(3\pi\hbar\epsilon_0 c^3)$ the single-atom spontaneous emission rate. Two limiting regimes control collectivity \cite{damanet2016competition}:
\begin{itemize}
\item \textbf{Small-sample (Dicke) limit:} if $x_{ij}\ll 1$ for all pairs $(i,j)$, then $g_{ij}\to \gamma_0$ and the dissipator \eqref{eq:SM_ME} \emph{closes} on collective operators,
\begin{equation}
    \mathcal D_{\rm coll}[\rho] = \gamma_0(n_R\!+\!1)\Big(J_-\rho J_+ - \tfrac12\{J_+J_-,\rho\}\Big)
    + \gamma_0 n_R\Big(J_+\rho J_- - \tfrac12\{J_-J_+,\rho\}\Big),
    \qquad J_\pm\equiv \sum_{i=1}^N \sigma^{(i)}_\pm.
\label{eq:SM_collectiveD}
\end{equation}
\item \textbf{Large-sample limit:} if $x_{ij}\gg 1$ (or if the angular average is effective), then $g_{ij}\to 0$ for $i\neq j$; emitters relax independently and cooperative pathways vanish.
\end{itemize}

\subsection{Validity conditions and role of dispersion}

The derivation of Eq. \eqref{eq:SM_ME} assumes weak coupling ($\gamma_0\ll\omega_0$), a short reservoir correlation time relative to system timescales (Markov limit), and spectral separation of Bohr frequencies (secular approximation) \cite{Agarwal1974}. As analyzed in Ref. \cite{damanet2016competition}, \emph{finite-size} phase factors and \emph{dipole--dipole} dispersion compete: the former suppress $g_{ij}$ off-diagonals as $x_{ij}$ grows, while the latter introduces coherent exchange that can break the ideal collective (Dicke) symmetry.

\subsection{Mapping to the collective interaction used in the main text}
Under the identical-coupling conditions of the small-sample limit (all $x_{ij}\ll 1$ and negligible inhomogeneity in $f_{ij}$), the mode functions can be taken equal across sites, $u_{k\lambda}(\mathbf r_i)\simeq u_{k\lambda}(\mathbf r_j)$ for all $i,j$, and Eq. \eqref{eq:SM:HI:dipole} reduces to
\begin{equation}
    H_\mathrm{I} \simeq \sum_{\mathbf{k},\lambda}\hbar g_{\mathbf{k}\lambda}\, J_+\, a_{\mathbf{k}\lambda}+\text{h.c.},
\end{equation}
which, when grouped by frequency, yields the collective form used in Eq. (2) of the main text,
\begin{equation}
    H_\mathrm{I} = \sum_{\mathbf k,\lambda}\Big(\sum_{i=1}^N g_{\mathbf k\lambda}\,\sigma^{(i)}_+\Big) a_{\mathbf k\lambda} + \text{h.c.}
\end{equation}
The corresponding dissipator is exactly Eq. \eqref{eq:SM_collectiveD} with upward/downward rates weighted by $n$ and $n+1$. This identifies the parameter regime where the common reservoir acts through \emph{indistinguishable} jump operators $J_\pm$, which is the prerequisite for the interference-enabled charging mechanism studied in the main text.

\section{Steady State Structure}\label{SteadyStructure}
\subsection{Preface}

In this section of the Supplementary Material we provide the analytical expression (and examples) of the steady-state of the system under symmetric dissipative dynamics given by the GKLS master equation 
\begin{eqnarray} \label{MESup}
    \dot{\rho} = \mathcal{L}[\rho].
\end{eqnarray}
We begin by introducing the system Hamiltonian and the collective angular momentum basis \( |j, m, \sigma \rangle \)—which is also the Schur basis arising from Schur-Weyl duality. This symmetry-adapted basis not only respects the permutation invariance of the Hamiltonian but also provides a convenient framework for describing Lindbladian dynamics involving collective jump operators. Subsequently, we present explicit examples of analytically obtained steady states for small values of \( N \) (two-qubit and three-qubit scenarios).

\subsection{The Model}

We consider a system of $N$ identical, non-interacting two-level systems (TLSs). The Hamiltonian is given by:
\begin{equation}
    H_B = \sum^{N}_{i=1}\omega \sigma_{+}^{(i)}\sigma_{-}^{(i)}
\end{equation}
where $\sigma_{+}^{i} = \ket{1_i}\bra{0_i}$ and $\sigma_{-}^{i} = \ket{0_i}\bra{1_i}$ are the raising and lowering operators for the $i$-th TLS, with $\ket{0_i}$ and $\ket{1_i}$ denoting the ground and excited states, respectively. The Hilbert space for each TLS is $\mathbb{C}^2$, and the total Hilbert space is the tensor product $\mathcal{H} = (\mathbb{C}^2)^{\otimes N}$ with dimension $2^N$.

To exploit the inherent permutation symmetry of the system, we adopt the collective angular momentum basis (also known as the \textbf{Extended Dicke basis}) $\{\ket{j,m,\sigma}\}$. This basis simultaneously diagonalizes the total spin operators $\vec{J}^2$ and $J_z$, defined as:
\begin{align}
    J_z &= \frac{1}{2} \sum_{i=1}^{N} \sigma_{z}^{(i)}, \\
    \vec{J}^2 &= J_x^2 + J_y^2 + J_z^2,
\end{align}
where $\sigma_i^z = \ket{1_i}\bra{1_i} - \ket{0_i}\bra{0_i}$ is the Pauli-$z$ operator. We also define the ladder operators (which are our Lindblad operators) $J_{\pm}$ as follows:
\begin{equation}
    J_{\pm}= \sum_{i} \sigma_{\pm}^{i}.
\end{equation}
The basis states are labeled by two quantum numbers: the total angular momentum $j$,
\begin{equation}
    j \in \begin{cases} 
    \{0, 1, \dots, N/2\} & \text{if}\ N\ \text{ even}, \\
    \{1/2, 3/2, \dots, N/2\} & \text{if}\ N\ \text{odd}.
\end{cases}
\end{equation}
the magnetic angular momentum $m$ that takes values
\begin{equation}
    -j \leq m \leq j.
\end{equation}
A basis state is denoted $\ket{j, m, \sigma}$. The latin index $\sigma$ here is not a quantum number, but it denotes the degeneracy ranging $1\leq \sigma\leq \nu_{j}$ where $\nu_j = \binom{N}{\frac{N}{2}-j} - \binom{N}{\frac{N}{2}-j-1}, 
\quad \text{with } \binom{N}{k<0}=0.$ for a given $j$.
The reason this basis $\{\ket{j, m, \sigma}\}$ comes handy in solving the problem roots to the \textbf{Schur-Weyl duality}, which decomposes $\mathcal{H}$ into irreducible representations (irreps) of $\text{SU}(2)$ \cite{fulton1991representation,goodman2009symmetry,KEYL2002431}:
\begin{equation}
    \mathcal{H}^{(N)} \cong \bigoplus_{j} \mathcal{V}_j \otimes \mathcal{U}_j.
\end{equation}
Here $\mathcal{V}_j$ carries the irrep of the symmetric group $S_N$ (indexed by $\sigma$) with multiplicity $\nu_j$ and $\mathcal{U}_j$ carries the $SU(2)$ irrep of dimension $2j+1$ (indexed by $m$).

\subsection{Worked Example: Schur--Weyl Decomposition for a Two-Qubit System}

Consider the two-qubit state
\begin{equation}
    \ket{\psi} = \frac{\ket{00} + \ket{01} - \ket{10}}{\sqrt{3}}.
\label{eq:ex-state}
\end{equation}
For $N=2$, the Schur--Weyl (or Dicke) basis decomposes the Hilbert space into the triplet ($j=1$) and singlet ($j=0$) sectors:
\[
\ket{1,1} = \ket{11}, \qquad
\ket{1,0} = \tfrac{1}{\sqrt{2}}(\ket{10} + \ket{01}), \qquad
\ket{1,-1} = \ket{00}, \qquad
\ket{0,0} = \tfrac{1}{\sqrt{2}}(\ket{10} - \ket{01}).
\]
Expressing $\ket{01}$ and $\ket{10}$ in this basis,
\[
\ket{01} = \tfrac{1}{\sqrt{2}}(\ket{1,0} + \ket{0,0}), 
\qquad
\ket{10} = \tfrac{1}{\sqrt{2}}(\ket{1,0} - \ket{0,0}),
\]
one finds that the $\ket{1,0}$ contribution cancels, leading to the exact decomposition
\begin{equation}
    \ket{\psi} = 
    \frac{1}{\sqrt{3}}\,\ket{1,-1}
    + \sqrt{\frac{2}{3}}\,\ket{0,0}
    \qquad\Rightarrow\quad
    p_{j=1} = \frac{1}{3}, \quad p_{j=0} = \frac{2}{3}.
\label{eq:ex-decomp}
\end{equation}
Thus, the state carries a weight of $1/3$ in the triplet ($m=-1$) subspace and $2/3$ in the singlet.
The corresponding density operator is
\begin{equation}
    \rho_\psi = \ket{\psi}\!\bra{\psi}
    = \frac{1}{3}\ket{1,-1}\!\bra{1,-1}
    + \frac{\sqrt{2}}{3}
      \Big(\ket{1,-1}\!\bra{0,0} + \ket{0,0}\!\bra{1,-1}\Big)
    + \frac{2}{3}\ket{0,0}\!\bra{0,0}.
\label{eq:ex-rho}
\end{equation}
The off-diagonal terms represent coherences between subspaces with different total angular momentum quantum numbers ($j=1$ and $j=0$).

\subsection{General Intuition for the structure of steady-state}

Understanding the structure of the steady-state density matrix hinges on the GKLS master equation \eqref{MESup}, which governs the system evolution. The dissipator is pivotal in selectively removing coherences between states distinguished by different total angular momentum $J^2$. Notably, the lowering and raising operators $J_{\pm}$ commute with $J^2$ ($[J_{\pm}, J^2] = 0$), ensuring that the dissipative processes described by the jumps conserve the total angular momentum $J^2$. This framework allows us to express the dissipator term as
\begin{align*} \label{DissipatorAngular}
    \mathcal{L}[\cdot] = \gamma(n+1)\left(J_{-} \cdot J_{+} - \frac{1}{2}\{J_{+}J_{-}, \cdot\}\right) + \gamma n\left(J_{+}\cdot J_{-} - \frac{1}{2}\{J_{-}J_{+}, \cdot\}\right),\numberthis
\end{align*}
where the lowering operator decomposes into irreducible blocks,
\[
J_- = \bigoplus_{j}\bigoplus_{\sigma=1}^{\nu_j} L_j^{(\sigma)},
\]
where \(\nu_j\) is the multiplicity of the \(j\)-irrep and 
\(\{\ket{j,m,\sigma}\}_{m=-j}^{j}\) is the corresponding eigenbasis.
Each \(L_j^{(\sigma)}\) acts as the ladder (lowering) operator on the 
\((2j+1)\)-dimensional magnetic subspace of the \(\sigma\)-th copy,
\[
L_j^{(\sigma)}\ket{j,m,\sigma} \propto \ket{j,m-1,\sigma}.
\]
Hence, \(J_\pm\) change the magnetic quantum number \(m\) 
but preserve the total angular momentum \(j\) and the degeneracy index \(\sigma\).
. Crucially, this structure implies that the dissipator $\mathcal{L}$ projects the density matrix $\rho$ onto distinct $j$ subspaces, leading to the selective elimination of off-diagonal coherences between states with different $j$ values. In simpler terms, when the system’s density matrix is expressed in the ${|j,m,\sigma\rangle}$ basis, each block with angular momentum $j$ undergoes a thermalization process where diagonal elements thermalize and off-diagonal elements vanish. But for blocks with degeneracy $\sigma$, off-diagonal elements with the same $m$ do not vanish but still thermalize. To illustrate, consider a 3-qubit system with $j=\frac{3}{2}, \frac{1}{2}, \frac{1}{2}$ \cite{Shankar1994,sakurai2020modern}, resulting in a degeneracy $\nu_{j}=2$ (indexed by $\sigma$) for the $j=\frac{1}{2}$ subspace. Then we have
\begin{equation}\label{AngularL}
    J_{-}^{(3)} = L_{3/2} \bigoplus L^{(1)}_{1/2} \bigoplus L^{(2)}_{1/2}.
\end{equation}

\subsection{The structure}\label{Structure}

The steady-state density matrix $\rho_{SS}$ is block-diagonal in the basis $\ket{j, m, \sigma}$, where $j$ denotes the total angular momentum quantum number, $m$ its projection, and $\sigma$ the degeneracy index. This state admits the decomposition:
\begin{equation} \label{GeneralSteady-State0}
    \rho_{SS} = \bigoplus_{j} (\chi_{j} \otimes \tau_{j}),
\end{equation}
where $\chi_j$ acts on the degeneracy subspace, and $\tau_j$ is the thermal part of the steady state that acts on the magnetic angular momentum subspace for each $j$. Explicitly, we can write:
\begin{equation}
  \mathbb{I}_{\nu_{j}} \otimes \tau_j 
  \;=\; \frac{1}{
  \mathcal Z_j}\sum_{\sigma = 1}^{\nu_{j}}\sum_{m=-j}^{j} \alpha_{c}^{\,j+m}\,\ket{j,m, \sigma}\bra{j,m,\sigma},
\qquad
\mathcal Z_j=\sum_{k=0}^{2j}\alpha_{c}^{k}=\frac{1-\alpha_{c}^{2j+1}}{1-\alpha_{c}}.
\end{equation}
\par
\noindent
and
\begin{equation}
    \chi_j \otimes \mathbb{I}_{m} = \sum_{\sigma, \sigma' = 1}^{\nu_j} (\chi_j)_{\sigma,\sigma'} \sum_{m = -j}^{j} \ket{j, m, \sigma}\bra{j, m, \sigma'}.
\end{equation}
The matrix elements $(\chi_{j})_{\sigma,\sigma'}$ can be given as $(\chi_{j})_{\sigma,\sigma'} = \sum_{m=-j}^{j}\bra{j,m,\sigma'}\rho(t_0)\ket{j,m,\sigma}$.
The diagonal elements $(\chi_{j})_{\sigma,\sigma}$ represents the population of the sectors with degeneracy where "which-path" information preserved in degeneracy subspaces, and the off-diagonal elements $(\chi_{j})_{\sigma, \sigma'} (\sigma \neq \sigma')$ represent quantum coherences that survive collective dissipation. It is straightforward to see that $\chi_{j}$ is a positive semi-definite matrix $\chi_{j} \succeq 0$ with unit-trace $\tr \chi_{j} = 1.$
\par
\noindent
The total density matrix $\rho$ is given by:
\begin{equation}\label{GeneralSteady-State1}
    \rho_{SS} \;=\; \bigoplus_{j}\chi_j\otimes\tau_j
\;=\; \sum_{j}\sum_{\sigma,\sigma'=1}^{\nu_j}\sum_{m=-j}^{j}
(\chi_j)_{\sigma\sigma'}\frac{\alpha_{c}^{j+m}}{\mathcal Z_j}\,\ket{j,m,\sigma}\bra{j,m,\sigma'}. 
\end{equation}
Note that the steady-state density matrix is normalized:
\begin{equation}
\operatorname{Tr}\rho_{SS}=\sum_j \operatorname{Tr}\chi_j 
=\sum_{j} p_{j} = 1.
\end{equation}
\subsection{Examples}
\subsubsection{Two-qubit Steady State}

The extended Dicke basis for the two-qubit system is given as follows: $\mathcal{B}^{(2)} = \{\ket{00},\frac{\ket{01}+\ket{10}}{\sqrt{2}},\ket{11},\frac{\ket{01}-\ket{10}}{\sqrt{2}}\}$. 
For an arbitrary initial state at time $t_{0}$ in this basis $\rho(t_{0}) = \sum_{i,j = 1}^{4}\rho_{ij}(t_0)\ket{i}\bra{j}$ where $\ket{i},\ket{j} \in \mathcal{B}^{(2)}$, the steady-state condition $(\mathcal{L}\rho_{SS}^{(2)} = 0)$ yields the following linear equations:
\begin{equation}
\begin{aligned}
    n \gamma \rho_{1 1} &= (1 + n) \gamma \rho_{2 2}, \\
    (\gamma + 3 n \gamma) \rho_{1 2} &= 2 (1 + n) \gamma \rho_{2 3}, \\
    (1 + 2 n) \gamma \rho_{1 3} &= 0, \\
    n \gamma \rho_{1 4} &= 0, \\
    (\gamma + 3 n \gamma) \rho_{2 1} &= 2 (1 + n) \gamma \rho_{3 2}, \\
    \gamma \left( n \rho_{1 1} - (1 + 2 n) \rho_{2 2} + (1 + n) \rho_{3 3} \right) &= 0, \\
    2 n \gamma \rho_{1 2} &= (2 + 3 n) \gamma \rho_{2 3}, \\
    (1 + 2 n) \gamma \rho_{2 4} &= 0, \\
    (1 + 2 n) \gamma \rho_{3 1} &= 0, \\
    2 n \gamma \rho_{2 1} &= (2 + 3 n) \gamma \rho_{3 2}, \\
    n \gamma \rho_{2 2} &= (1 + n) \gamma \rho_{3 3}, \\
    (1 + n) \gamma \rho_{3 4} &= 0, \\
    n \gamma \rho_{4 1} &= 0, \\
    (1 + 2 n) \gamma \rho_{4 2} &= 0, \\
    (1 + n) \gamma \rho_{4 3} &= 0.
\end{aligned}
\end{equation}
The solution to this system of linear equations gives us the two-qubit steady state $\rho_{SS}^{(2)}$ given by 
\begin{align}\label{2qss}
    \rho_{SS}^{(2)} = 
    \begin{pmatrix}
    \frac{p_1}{1+\alpha_{c} + \alpha_{c}^2} & 0 & 0 & 0 \\
    0 & \frac{\alpha_{c} p_1}{1+\alpha_{c} + \alpha_{c}^2} & 0 & 0 \\
    0 & 0 & \frac{\alpha_{c}^2 p_1}{1+\alpha_{c} + \alpha_{c}^2} & 0 \\
    0 & 0 & 0 & p_2
    \end{pmatrix}
\end{align}
where $p_1 = \frac{\rho_{11}(t_0) + \rho_{22}(t_0) + \rho_{33}(t_0)}{\mathcal{N}^{(2)}}$ and $p_2 = \frac{\rho_{44}(t_0)}{\mathcal{N}^{(2)}}$ are the populations of the sectors corresponding to the sectors $j=1$, and $j=0$, and $\alpha$ is the Gibbs ratio given by $\alpha_{c}=\frac{n_{c}}{n_{c}+1}$. Here $\mathcal{N}^{(2)}$ is the normalization factor for the two-qubit initial density matrix given by $\mathcal{N}^{(2)} = \sum_{i=1}^4\rho_{ii}(t_{0})$.
The state has the angular momentum decomposition:
\begin{itemize}
    \item \textbf{$j=1$ block} (dimension 3, degeneracy $\nu_j=1$):
    \[
    \chi_{j=1} = p_1, \quad 
    \tau_{j=1} = \begin{pmatrix} \frac{1}{1 + \alpha_{c} + \alpha_{c}^2} & 0 & 0 \\ 0 & \frac{\alpha_{c}}{1 + \alpha_{c} + \alpha_{c}^2} & 0 \\ 0 & 0 & \frac{\alpha_{c}^2}{1 + \alpha_{c} + \alpha_{c}^2} \end{pmatrix} = \sum_{m=-1}^{1} \frac{\alpha_{c}^{1+m}}{1 + \alpha_{c} + \alpha_{c}^2} \ket{1,m}\bra{1,m}
    \]
Block factorization: 
    \[
    \chi_{j=1} \otimes \tau_{j=1} = p_1 \cdot \begin{pmatrix} \frac{1}{1 + \alpha_{c} + \alpha_{c}^2} & 0 & 0 \\ 0 & \frac{\alpha_{c}}{1 + \alpha_{c} + \alpha_{c}^2} & 0 \\ 0 & 0 & \frac{\alpha_{c}^2}{1 + \alpha_{c} + \alpha_{c}^2} \end{pmatrix}
    \]
    
    \item \textbf{$j=0$ block} (dimension 1, degeneracy $\nu_j=1$):
    \[
    \chi_{j=0} = p_2, \quad 
    \tau_{j=0} = 1 = \alpha_{c}^{0+0} \ket{0,0}\bra{0,0}
    \]
Block factorization: 
    \[
    \chi_{j=0} \otimes \tau_{j=0} = p_2
    \]
\end{itemize}
Full state reconstruction:
\begin{equation}
    \rho_{SS}^{(2)} = \underbrace{\begin{pmatrix} \frac{p_1}{1+\alpha_{c} + \alpha_{c}^2} & 0 & 0 \\ 0 & \frac{\alpha_{c} p_1}{1+\alpha_{c} + \alpha_{c}^2} & 0 \\ 0 & 0 & \frac{\alpha_{c}^2 p_1}{1+\alpha_{c} + \alpha_{c}^2} \end{pmatrix}}_{\chi_{1} \otimes \tau_{1}} \oplus \underbrace{\begin{pmatrix} p_2 \end{pmatrix}}_{\chi_{0} \otimes \tau_{0}}
\end{equation}
which is the original form \eqref{GeneralSteady-State0}.

\subsubsection{Three-qubit steady state}
\noindent For a three-qubit system, the extended Dicke basis looks as follows:

$\mathcal{B}^{(3)} = \left\{
\ket{000}, 
\frac{\ket{001}+\ket{010}+\ket{100}}{\sqrt{3}}, 
\frac{\ket{011}+\ket{101}+\ket{110}}{\sqrt{3}}, 
\ket{111}, 
\frac{\ket{001}+\ket{010}-2\ket{100}}{\sqrt{6}}, 
\frac{2\ket{011}-\ket{101}-\ket{110}}{\sqrt{6}},\frac{\ket{001}-\ket{010}}{\sqrt{2}},\frac{\ket{101}-\ket{110}}{\sqrt{2}}
\right\}.$
As in the previous case, in this basis, for an arbitrary initial state $\rho(t_{0}) = \sum_{i,j = 1}^{8}\rho_{ij}\ket{i}\bra{j}$ where $\ket{i},\ket{j} \in \mathcal{B}^{(3)}$, the steady-state condition \(\mathcal L[\rho_{SS}^{(3)}]=0\) gives rise to the following block form for \(\rho_{SS}^{(3)}\):
\begin{align}\label{3qss}
    \rho_{SS}^{(3)} =
    \begin{pmatrix}
    \frac{p_{1}}{1+ \alpha_{c} + \alpha_{c}^2 + \alpha_{c}^3} & 0 & 0 & 0 & 0 & 0 & 0 & 0 \\
    0 & \frac{\alpha_{c} p_{1}}{1+ \alpha_{c} + \alpha_{c}^2 + \alpha_{c}^3} & 0 & 0 & 0 & 0 & 0 & 0 \\
    0 & 0 & \frac{\alpha_{c}^2 p_{1}}{1+ \alpha_{c} + \alpha_{c}^2 + \alpha_{c}^3} & 0 & 0 & 0 & 0 & 0 \\
    0 & 0 & 0 & \frac{\alpha_{c}^3 p_{1}}{1+ \alpha_{c} + \alpha_{c}^2 + \alpha_{c}^3} & 0 & 0 & 0 & 0 \\
    0 & 0 & 0 & 0 & \frac{p_{2}}{1+\alpha_{c}} & 0 & \frac{c_2}{1+\alpha_{c}} & 0 \\
    0 & 0 & 0 & 0 & 0 & \frac{\alpha_{c} p_{2}}{1 + \alpha_{c}} & 0 & \frac{\alpha_{c} c_2}{1+\alpha_{c}} \\
    0 & 0 & 0 & 0 & \frac{c_2}{1+\alpha_{c}} & 0 & \frac{p_{3}}{1+\alpha_{c}} & 0 \\
    0 & 0 & 0 & 0 & 0 & \frac{\alpha_{c} c_2}{1+\alpha_{c}} & 0 & \frac{\alpha_{c} p_{3}}{1+\alpha_{c}}
    \end{pmatrix},
\end{align}
which again admits the form \eqref{GeneralSteady-State0}. As in the two-qubit case, the populations 
$p_{1} = \frac{\rho_{11}(t_{0}) + \rho_{22}(t_{0}) + \rho_{33}(t_{0}) + \rho_{44}(t_{0})}{\mathcal{N}^{(3)}}$, $p_{2} = \frac{\rho_{55}(t_{0}) + \rho_{66}(t_{0})}{\mathcal{N}^{(3)}}$, 
and $p_{3} = \frac{\rho_{77}(t_{0}) + \rho_{88}(t_{0})}{\mathcal{N}^{(3)}}$ are the populations corresponding to each sector characterized by $j = \frac{3}{2}, j=\frac{1}{2}$ and $j=\frac{1}{2}$. Here $\mathcal{N}^{(3)} = \sum_{i=1}^8\rho_{ii}(t_{0})$ is the normalization factor.
Here $c_2 = \rho_{57}(t_0) + \rho_{68}(t_0)$ quantifies the coherence between the energy levels of the same sector.
It is easy to see that the state \eqref{3qss} admits the form \eqref{GeneralSteady-State0} with
\begin{equation}
\small
\chi_{j=\frac{3}{2}} = \begin{pmatrix} p_1 \end{pmatrix},\quad 
\tau_{j=\frac{3}{2}} = 
\begin{pmatrix}
\frac{1}{1 + \alpha_{c} + \alpha_{c}^2 + \alpha_{c}^3} &0 &0 &0 \\
0 &\frac{\alpha_{c}}{1 + \alpha_{c} + \alpha_{c}^2 + \alpha_{c}^3} &0 &0 \\
0 &0 &\frac{\alpha_{c}^2}{1 + \alpha_{c} + \alpha_{c}^2 + \alpha_{c}^3} &0 \\
0 &0 &0 &\frac{\alpha_{c}^3}{1 + \alpha_{c} + \alpha_{c}^2 + \alpha_{c}^3}
\end{pmatrix},\quad 
\chi_{j=\frac{1}{2}} = 
\begin{pmatrix} p_2 & c_2 \\ c_2 & p_3 \end{pmatrix},\quad 
\tau_{j = \frac{1}{2}} = 
\begin{pmatrix} \frac{1}{1+\alpha_{c}} &0 \\ 0 &\frac{\alpha_{c}}{1+\alpha_{c}} \end{pmatrix}.
\end{equation}

\subsection{Proof of leakage in evolution sectors violating $j=j'$ and $\Delta_{J_z}=0$}\label{A2}

Let us begin with the master equation,
\begin{equation}\label{ME2}
    \frac{d\rho}{dt} = \gamma(n+1)\left(J_-\rho J_+ - \dfrac{1}{2}\{J_+J_-, \rho\}\right) + \gamma n\left(J_+\rho J_- - \dfrac{1}{2}\{J_-J_+, \rho\}\right).
\end{equation}
For density matrix elements $\rho_{j,m,\sigma : j',m',\sigma'} = \langle j, m,\sigma|\rho|j', m',\sigma'\rangle$, we then have
\begin{align}\label{rhojm}
    \frac{d}{dt}\rho_{j,m,\sigma;j'm',\sigma'} = &-\frac{\gamma(n+1)}{2}[A_{m} + A'_{m'}]\rho_{j,m,\sigma;j',m',\sigma'} - \frac{\gamma n}{2}[A_{m+1} + A'_{m'+1}]\rho_{j,m,\sigma;j'm',\sigma'} \nonumber\\
    &+ \gamma(n+1)\sqrt{A_{m+1}A'_{m'+1}}\rho_{j,m+1,\sigma;j',m'+1,\sigma'} + \gamma n\sqrt{A_{m}A'_{m'}}\rho_{j,m-1,\sigma;j', m'-1,\sigma'}.
\end{align}
where
\begin{equation}
\begin{aligned}
    A_{m} &\equiv(j+m)(j-m+1), \qquad A'_{m'}\equiv(j'+m')(j'-m'+1), \\
    A_{m+1} &\equiv(j-m)(j+m+1), \quad A'_{m'+1}\equiv(j'-m')(j'+m'+1).
\end{aligned}
\end{equation}
Note that we have omitted the degeneracy label $\sigma$ in the basis as it is independent of the results here.
The master equation \eqref{ME2} can be vectorized and rewritten as
\begin{equation}
    \frac{d}{dt}|\rho\rangle\rangle = \mathcal{L}|\rho\rangle\rangle,
\end{equation}
where $\mathcal{L}$ is the Liouvillian super-operator. The general solution is
\begin{equation}
    |\rho(t)\rangle\rangle = e^{\mathcal{L}t}|\rho(t_0)\rangle\rangle.
\end{equation}
The eigenvalues of $\mathcal{L}$ determine the evolution of the state with
\begin{equation*}
    \mathcal{L} = \text{diag}(\mathcal{L}^{(1)}, \mathcal{L}^{(2)}, \mathcal{L}^{(3)}, \ldots)
\end{equation*}
where $\mathcal{L}^{(k)}$ is the sub-matrix governing the evolution of coherences in the $k$-th Bohr sector. For a specific Bohr sector S$(j,j',\Delta j_z)$ with $\Delta j_z = m - m'$, the corresponding sub-matrix $\mathcal{L}_{r,s,\Delta j_z}$ would be:
\begin{equation*}
    \mathcal{L}_{r,s,\Delta j_z} =
    \begin{pmatrix}
    d_1 & c_1 & 0 & 0 & \cdots \\
    b_1 & d_2 & c_2 & 0 & \cdots \\
    0 & b_2 & d_3 & c_3 & \cdots \\
    0 & 0 & b_3 & d_4 & \cdots \\
    \vdots & \vdots & \vdots & \vdots & \ddots
    \end{pmatrix}
\end{equation*}
where for the tri-diagonal $m (d_m, c_m, b_m)$ we have $d_m = - \frac{\gamma}{2}[n(A_{m+1} + A'_{m'+1}) + (n+1)(A_m + A'_{m'})]$, $b_m = \gamma n\sqrt{A_{m+1}A'_{m'+1}}$ and $c_m = \gamma(n+1)\sqrt{A_{m+1}A'_{m'+1
}}$.
As we will show below, at most one eigenvalue is 0, corresponding to the steady state and all other eigenvalues have negative real parts, causing decay to the thermal steady state. Therefore, the density matrix evolves as
\begin{equation}
    |\rho(t)\rangle\rangle = |\rho_{\lambda=0}\rangle\rangle + \sum_{\lambda\neq0} c_\lambda e^{\lambda t} |\rho_\lambda\rangle\rangle,
\end{equation}
where $\rho_{\lambda=0}=\rho_{SS}$ is the steady state and $\rho_{\lambda\neq0}$ are the eigen-modes corresponding to non-zero eigenvalues $\lambda$. To prove that all eigenvalues have negative real parts, we use the following Gershgorin's circle theorem.
\begin{theorem}[Gershgorin's Circle Theorem]\label{Gershgorin}\cite{HornMatrixAnalysis2012}
Let $M$ be a complex $n \times n$ matrix, with entries $m_{rs}$. For each row (column) $r\in\{1,\ldots,n\}$, define the radius $R_i$ as the sum of absolute values of non-diagonal entries in that row (column)
\begin{equation}
    R_r = \sum_{r \neq s} |m_{rs}|.
\end{equation}
Let $D(m_{rr},R_r) \subseteq \mathbb{C}$ be a closed disc, Gershgorin disc, in the complex plane centered at $m_{rr}$ with radius $R_r$:
\begin{equation}
    D(m_{rr},R_r) = \{z \in \mathbb{C} : |z - m_{rr}| \leq R_r\}.
\end{equation}
Then, every eigenvalue of matrix $M$ lies within at least one of the Gershgorin discs $D(m_{rr},R_r)$.
\end{theorem}
\begin{corollary}\label{corollary}
Using the column-based Gershgorin's Circle Theorem \ref{Gershgorin}, if $\lambda$ is an eigenvalue of $M$, there exists some $r$ such that
\begin{equation}
    |\lambda - m_{rr}| \leq R_r.
\end{equation}
\end{corollary}
\noindent For our system the diagonal terms of matrix $M$ are
\begin{equation}
    m_{rr} = d_{m} = - \gamma(n+1)\dfrac{A_{m} + A'_{m'}}{2} - \gamma n\dfrac{A_{m+1} + A'_{m'+1}}{2},
\end{equation}
and the off-diagonal (coupling) terms in $m^{th}$ column are
\begin{equation}
\begin{aligned}
    c_{m-1} &\equiv \gamma(n+1)\sqrt{A_{m}A_{m'}} , \\
    b_{m} &\equiv \gamma n\sqrt{A_{m+1}A_{m'+1}}.
\end{aligned}
\end{equation}
Exploiting corollary \eqref{corollary} we have
\begin{equation}\label{inequlity1}
    |\lambda - d_{m}|\leq R_r = |c_{m-1}| + |b_{m}|.
\end{equation}
For any complex numbers $z$ and $w$, if $|z - w| \leq a$, then \cite{Churchill}
\begin{equation}\label{zr}
    \Re(z) - \Re(w) \leq |z - w| \leq a.
\end{equation}
Now applying relation \eqref{zr} to our inequality \eqref{inequlity1} we get
\begin{equation}
    \Re(\lambda) - \Re(m_{rr}) \leq |\lambda - m_{rr}| \leq |c_{m-1}| + |b_{m}|,
\end{equation}
which gives
\begin{equation}
    \Re(\lambda) \leq \Re(m_{rr}) + |c_{m-1}| + |b_{m}|,
\end{equation}
hence
\begin{equation}\label{inequality2}
    \Re(\lambda) \leq -\gamma(n+1)\left(\frac{A_m+A_{m'}}{2} - \sqrt{A_mA_{m'}}\right) - \gamma n\left(\frac{A_{m+1}+A_{m'+1}}{2}-\sqrt{A_{m+1}A_{m'+1}}\right).
\end{equation}
For $j \neq j'$ the right hand side of inequality \eqref{inequality2} is strictly negative due to the arithmetic-geometric mean inequality
\begin{equation}
    \frac{x+y}{2} > \sqrt{x}\sqrt{y}.
\end{equation}
Thus, there always exists a gap $\Delta > 0$ such that
\begin{equation}\label{gap}
    |\Re(\lambda)| \geq \gamma(n+1)\left(\frac{A_m+A_{m'}}{2} - \sqrt{A_mA_{m'}}\right) + \gamma n\left(\frac{A_{m+1}+A_{m'+1}}{2}-\sqrt{A_{m+1}A_{m'+1}}\right) = \Delta > 0.
\end{equation}
This proves that all coherences between different $j$ blocks decay exponentially with a rate at least $\Delta$.
\begin{figure}[t]
    \centering
    \begin{subfigure}[b]{0.32\textwidth}
        \includegraphics[width=\textwidth]{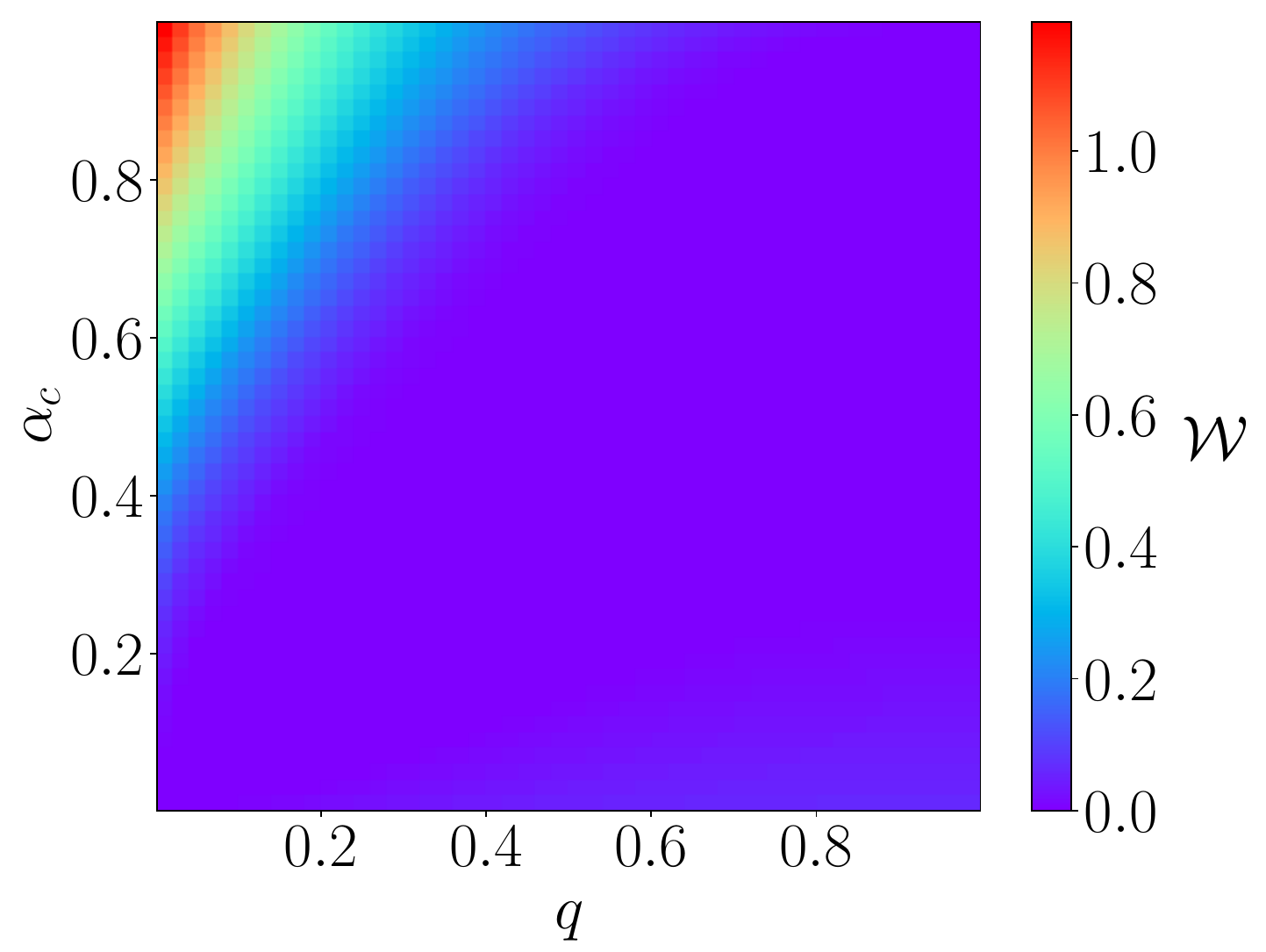}\caption{$N=4$}
        \label{4Qubits_density}
    \end{subfigure}
    \begin{subfigure}[b]{0.32\textwidth}
\includegraphics[width=\textwidth]{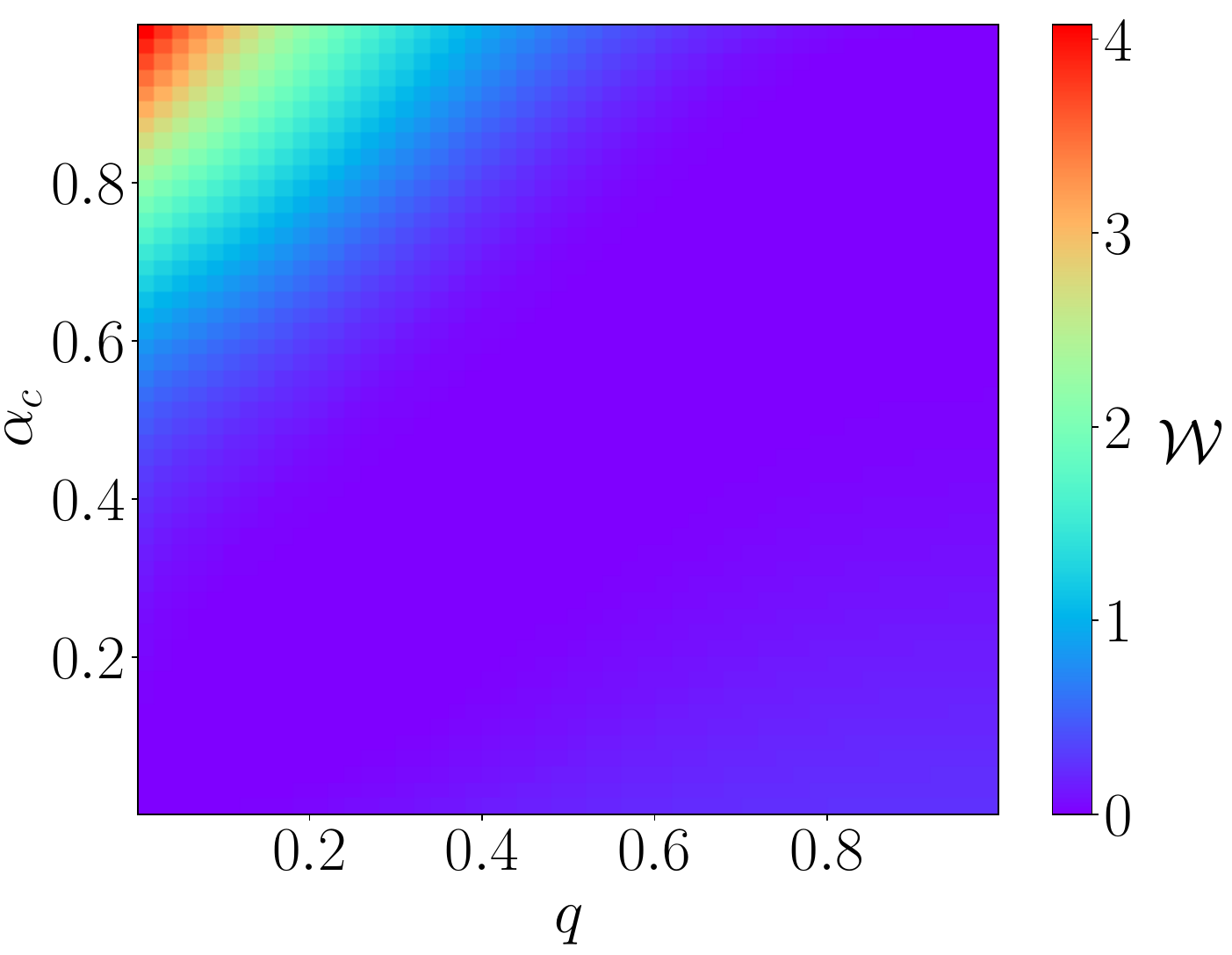}
        \caption{$N=10$}\label{10Qubits_density}
    \end{subfigure}
    \begin{subfigure}[b]{0.32\textwidth}
\includegraphics[width=\textwidth]{26Qubits_density.pdf}
        \caption{$N=26$}\label{26Qubits_density}
    \end{subfigure}
    \captionsetup{justification=justified}
    \caption{\textbf{Steady-state ergotropy under ideal collective dissipation ($\eta=1$).}
Heat map of the steady-state ergotropy $\mathcal{W}$ for $N=4,10,26$ with perfectly collective coupling versus the common-reservoir parameter $\alpha_c=e^{-\beta_R\hbar\omega_0}$. Except along the fine-tuned line $\alpha_c=q$ (initial temperature equals reservoir temperature), the steady state is nonpassive and $\mathcal{W}>0$. As $\alpha_c\to 1$ (hotter common bath), \(\mathcal{W}_{\rm SS}\) typically increases because the collective jumps \(J_\pm\) thermalize each Dicke ladder while preserving the interference-induced ladder structure, enabling ergotropy extraction from an incoherent reservoir.\justifying} 
\label{ErgoAlpha-q}
\end{figure}

\section{Activated ergotropy increase with system size}

Under perfectly collective coupling, the dissipator closes on the collective jumps \(J_\pm\) (see Sec.~\ref{SteadyStructure}), so the Liouvillian preserves total spin \(j\) and acts independently on each Dicke ladder. Within any fixed \(j\), the stationary shell populations obey the geometric law
\begin{equation}
    \frac{p^{\mathrm{st}}_{k+1}}{p^{\mathrm{st}}_{k}}=\alpha_c,\qquad
    \alpha_c \equiv \frac{n(\omega_0,T_R)}{n(\omega_0,T_R)+1}=e^{-\beta_R\omega_0}\in(0,1),
\end{equation}
with \(k\) the excitation number in the bright ladder within each $j$-block.
The steady state is thus a direct sum of \emph{thermal ladders} (one per \(j\)), with weights over different \(j\) fixed by the
initial state (Sec.~\ref{Structure}). For the ground initial condition (all weight in the \(j{=}\frac{N}{2}\) manifold), the stationary energy
\(E_{\mathrm{SS}}\) and ergotropy \(\mathcal{W}\) admit the closed forms reported in Sec.~\ref{InitialDiagonal}, which increase monotonically with \(\alpha_c\).
Two robust features are apparent.
(i) Even with a zero-temperature collective reservoir ($\alpha\!\to\!0$), the steady state exhibits nonzero ergotropy for $N>1$ whenever the initial local temperature parameter $q$ is sufficiently large, and the activated ergotropy grows with system size (see Fig. \ref{ErgoAlpha-q}).
(ii) Analytically, in the thermodynamic limit ergotropy activation is generic: for fixed $(\alpha_c,q)\in(0,1)$ (see Sec. \ref{alpha_cq} for analytical proof),
\begin{equation}
    \lim_{N\to\infty}\mathcal{W}(N,\alpha_c,q) \;>\; 0 
    \qquad \text{whenever}\quad \alpha_c \neq q,
\end{equation}
so the fine-tuned line $\alpha_c=q$ is the only case yielding vanishing ergotropy.

\noindent
Taken together, Fig.~\ref{ErgoAlpha-q} demonstrates that starting from a fully passive product Gibbs state, collective dissipation alone generates a nonpassive stationary state with finite ergotropy, whose magnitude increases systematically with $N$.

\section{Existence of population inversion for $\alpha_c \neq q$}\label{alpha_cq}

We can show that, as long as temperatures of the local baths and the common bath are different, the stationary state has non-zero ergotropy in the limit $N\rightarrow\infty$. Let us denote $\alpha_c=e^{-\beta_c}$ and $q=e^{-\beta_{q}}$. Then, in the stationary state, the population on the $m-th$ excited level in the symmetric space is given by
\begin{align}
    s(m)=\alpha_c^m\frac{(1-\alpha_c) \left(1-q^{N+1}\right)}{(1-q) \left(1-\alpha_c^{N+1}\right)(1+q)^{N}},
\end{align}
while the population on the ground state of the subspace with the lowest energetic state on level $m$ is given by
\begin{align}
    n_s(m)=\frac{\sum_{p=m}^{N-m}q^p}{(1+q)^N}\frac{1}{\sum_{p=0}^{N-2m}\alpha_c^p}=\frac{(1-\alpha_c)\left(q^m-q^{N+1-m}\right)}{(1-q) (1+q)^{N}(1-\alpha_c^{N+1-2m})}.
\end{align}
One can easily verify that for 
\begin{align}
    \frac{s(m+1)}{n_s(m)}=\alpha_c\frac{(1-q^{1+N})(\alpha_c^{m}-\alpha_c^{1+N-m})}{(q^m-q^{1+N-m})(1-\alpha_c^{1+N})}\rightarrow_{N\rightarrow\infty}\alpha_c(\frac{\alpha_c}{q})^m,
\end{align}
and for $\alpha_c>q$ there always exist $m$ big enough that the above is greater than $1$, pointing to population inversion. Analogously,
\begin{align}
    \frac{s(m-1)}{n_s(m)}=\frac{1}{\alpha_c}\frac{(1-q^{1+N})(\alpha_c^{m}-\alpha_c^{1+N-m})}{(q^m-q^{1+N-m})(1-\alpha_c^{1+N})}\rightarrow_{N\rightarrow\infty}\frac{1}{\alpha_c}(\frac{\alpha_c}{q})^m,
\end{align}
and for $q>\alpha_c$ there always exist $m$ big enough that the above is smaller then $1$, pointing to population inversion.

\section{Initial State with zero Temperature}\label{InitialDiagonal}

Before proceeding with the calculations, it is important to note that the rationale behind our study stems from the physics involved. Specifically, our motivation is grounded in the observation that the anti-commutators outlined in Eq. \eqref{DissipatorAngular} yield expressions such as $\sigma^i_+\sigma^j_+$, consequently facilitating the generation of coherence within the overall state of the QB which, in turn, will lead to nonzero ergotropy. In investigating energetic properties of the shared bath, we will exploit a symmetry of the evolution induced by the master equation \eqref{MESup}: in the so called Dicke basis, the evolution of the diagonal is independent from the evolution of coherences. Let us denote the basis of $H_{B}$ as $\{|x\rangle\}_x$, with $x\in\{0,2^n-1\}$, and a digit $x_{i}$ on every $i$ position in their binary representation. Then, $\sigma_{z}^{i}|x\rangle=(-1)^{x_{i}} |x\rangle$. Dicke states $\dicke{N}{k}$ are symmetric states with respect to a fixed total number of excitations $e(x)=\sum_{l=0}^{2^N-1}\delta_{x_{l},1}$:
\begin{align*}
   \dicke{N}{k} = \frac{1}{C^{N}_{k}} \sum\nolimits_{x:e(x)=k}\ket{x},\numberthis
\end{align*}
where $C^{N}_{k}={\binom{N}{k}}^{\frac{1}{2}}$. For example, $\dicke{5}{0} =|00000\rangle$, $\dicke{3}{1} =\frac{1}{\sqrt{3}}\Big(|001\rangle+|010\rangle+|100\rangle\Big)$.
Action of jump operators on Dicke states can be characterized as a shift with a phase factor:
\begin{eqnarray}
    J_-\dicke{N}{k}=\alpha^{N}_{k}\dicke{N}{k-1},\label{O}\\
    J_+\dicke{N}{k}=\beta^{N}_{k}\dicke{N}{k+1},\label{OO}
\end{eqnarray}
where 
\begin{align*}
    \alpha^{N}_{k}&=k\frac{C^{N}_{k}}{C^{N}_{k-1}}=k\sqrt{\frac{N-k+1}{k}},\numberthis\\    \beta^{N}_{k}&=(N-k)\frac{C^{N}_{k}}{C^{N}_{k+1}}=(N-k)\sqrt{\frac{k+1}{N-k}}.\numberthis
\end{align*}
This leads to
\begin{eqnarray}
    J_-J_+\dicke{N}{k}=\underbrace{\alpha^{N}_{k+1}\beta^{N}_{k}}_{(N-k)(k+1)}\dicke{N}{k},\label{OOO}\\
    J_+J_-\dicke{N}{k}=\underbrace{\alpha^{N}_{k}\beta^{N}_{k-1}}_{(N-k+1)(k)}\dicke{N}{k}.\label{OOOO}
\end{eqnarray}
We can now analyze the action of the dissipator on the Dicke state:
\begin{align*}\label{Evo}
    \mathcal{L}[\dicke{N}{k}\dickeD{N}{k}] = -\dicke{N}{k}\dickeD{N}{k}q_{0}
    + \dicke{N}{k-1}\dickeD{N}{k-1}q_{-} + \dicke{N}{k+1}\dickeD{N}{k+1}q_{+},\numberthis
\end{align*}
with 
\begin{align*}\label{Evo2}
    q_{+}(\gamma,n,N,k)&=\gamma n\Big(\beta^{N}_{k}\Big)^2=\gamma n (N-k) (k+1),\\
     q_{-}(\gamma,n,N,k)&=\gamma(n+1)\Big(\alpha^{N}_{k}\Big)^2=\gamma (n+1) k (N-k+1),\\
     q_{0}(\gamma,n,N,k)&=\gamma n (N-k)(k+1)+ \gamma(n+1)k(N-k+1)\numberthis.
\end{align*}
%
%
%
By exploiting Eq. (\ref{Evo}), from the stationary condition
\begin{eqnarray}
    \mathcal{L}(\rho^{st})=0,
\end{eqnarray}
we then obtain a set of $N+1$ equations for the diagonal elements of the stationary state $\{p_{0}^{st},p_{1}^{st},\dots,p_{N}^{st}\}$:
\begin{align*}\label{set}
    0 & = -p_{0}^{st}q_{0}(0)+p_{1}^{st}q_{-}(1),\\
    0 & = p_{0}^{st}q_{+}(0)-p_{1}^{st}q_{0}(1)+p_{2}^{st}q_{-}(2),\\
     0 & = p_{1}^{st}q_{+}(1)-p_{2}^{st}q_{0}(2)+p_{3}^{st}q_{-}(3),\\
     &\dots,\\
     0 & = p_{m-1}^{st}q_{+}(m-1)-p_{m}^{st}q_{0}(m)+p_{m+1}^{st}q_{-}(m+1),\\
     &\dots,\\
     0 & = p_{N-2}^{st}q_{+}(N-2)-p_{N-1}^{st}q_{0}(N-1)+p_{N}^{st}q_{-}(N),\\
     0 & = p_{N-1}^{st}q_{+}(N-1)-p_{N}^{st}q_{0}(N).\numberthis
\end{align*}
Above, while keeping the parameters $\gamma$, $n$, and $N$ constant across all equations, we explicitly indicated the dependence of $q_{-}$ and $q_{+}$ from Eq. \eqref{Evo2} solely on $k$. The set (\ref{set}), along with the conditions $0\le\alpha p_{k}^{st}\leq 1$ and $\sum_{k=0}^{N}p_{k}^{st}=1$, possesses a unique solution. Solving it becomes trivial when we observe that 
\begin{eqnarray}\label{balance}
    q_{0}(k)=q_{+}(k)+q_{-}(k),
\end{eqnarray} 
representing the expected conservation of total population. Thus, we can determine the ratio $\frac{p^{st}_{k+1}}{p^{st}_{k}}$ by equating the "left" and "right" currents across each cut:
\begin{eqnarray}
    p^{st}_k q_{+}(k)=p^{st}_{k+1} q_{-}(k+1)
\end{eqnarray}
for $k=0,1,\dots,N-1$, which gives
\begin{eqnarray}\label{q}
    \frac{p^{st}_{k+1}}{p^{st}_{k}}=\frac{n}{n+1}=: \alpha_c,
\end{eqnarray}
subject to normalization $\sum_{k=0}^{N}p_{k}^{st}=1$, returning 
\begin{eqnarray}\label{p0}
    p_{0}^{st}=\frac{1-\alpha_c}{1-\alpha_c^{N+1}}.
\end{eqnarray}
Defining $E_B(\rho):=tr\{\rho H_B\}$ as the internal energy of the system, ergotropy of a state $\rho$ is calculated as \cite{Allahverdyan_2004}
\begin{equation}\label{ergo}
    \mathcal{W}_B(\rho):=E_B(\rho)-E_B(\rho_p),
\end{equation}
where $\rho_p$ is the passive state associated with $\rho$ by $\rho_p=\min_{U}U\rho U^{\dagger}$, with the minimization over all unitaries. We notice that for $n>0$, Dicke states with superposition in energy basis are going to be activated. Each of them can be diagonalized in the energy eigenbasis by some unitary (see Fig. \ref{ErgoPic}). As a unitary does not influence the trace, this diagonalization merely concentrates all the weight $p^{st}_{k}$ on one of the degenerated energy levels in the energy subspace. Applying a unitary diagonalizing the system in every energy subspace is the first step in implementing an optimal unitary performing ergotropy extraction \eqref{ergo} (though on itself, it does not contribute directly to the change of the average energy of the system). This is accomplished by the second and final step: a permutation pushing every population $p_{k}$ (for $k\geq 2$) from its respective energy level $k$ to energy level $1$. Note that this is always possible, since there are ${N\choose{1}}=N$ degenerated energy levels with energy $1$, and only one of them is already occupied by $p_{1}$, so there is a place for $N-1$ ones from higher energy levels.

The state after the permutation is passive, hence no other unitary can lead to higher ergotropy. Based on the final permutation, we conclude that the ergotropy of the stationary state reads
\begin{align*}
    \mathcal{W}_B = \sum_{k=2}^{N} p_{k}^{st}(k-1) = N + \frac{\alpha_c}{1-\alpha_c}+\frac{N+\alpha_c}{-1+\alpha_c^{1+N}},\numberthis
\end{align*}
where we exploited Eqs. \eqref{q} and \eqref{p0}. The energy of the stationary state is obtained as
\begin{align*}\label{EBq}
    E_B = \sum_{k=1}^{N} p_{k}^{st}k = N+\frac{1}{1-\alpha_c}+\frac{N+1}{-1+\alpha_c^{1+N}}.\numberthis
\end{align*}
The first thing to note is that nonzero ergotropy necessitates having more than one cell in the QB, i.e., $N>1$. This is in fact expected because with one cell the steady state of the battery would be a Gibbs state which is a completely passive state with zero ergotropy \cite{Allahverdyan_2004}. Defining $\mathcal{R}_B\equiv E_B-\mathcal{W}_B$, as the residual passive energy in the QB after ergotropy extraction, it is seen that as both the temperature of the shared reservoir $T_R$ and the number of cells $N$ in the battery approach infinity $\mathcal{R}_B$ approaches 1. This implies that the minimum amount of quanta of energy left in the QB after ergotropy extraction can never be zero but one. This is explained by the fact that ground state lies in the first block of the initial density state (with the largest angular momentum $j$) (see Fig. \ref{ErgoPic}).

Furthermore, in this limit, both energy and ergotropy exhibit linear growth, demonstrating additive behavior. Figure \ref{Energy} illustrates the steady-state energy $E_B$ and ergotropy $\mathcal{W}_B$ plotted against the number of cells $N$ for various values of $q$. Typically with temperature going up the mixedness of the state should go up as well, which causes the ergotropy of the state to go down. But, here it is the complete opposite - the collective character of thermal energy structure completely suppresses the disadvantage of the noise: eventually, for the high temperature the state looks \textit{purer} from the perspective of the system energy Hamiltonian.
\begin{figure}[ht]
\center
\includegraphics[width=\textwidth]{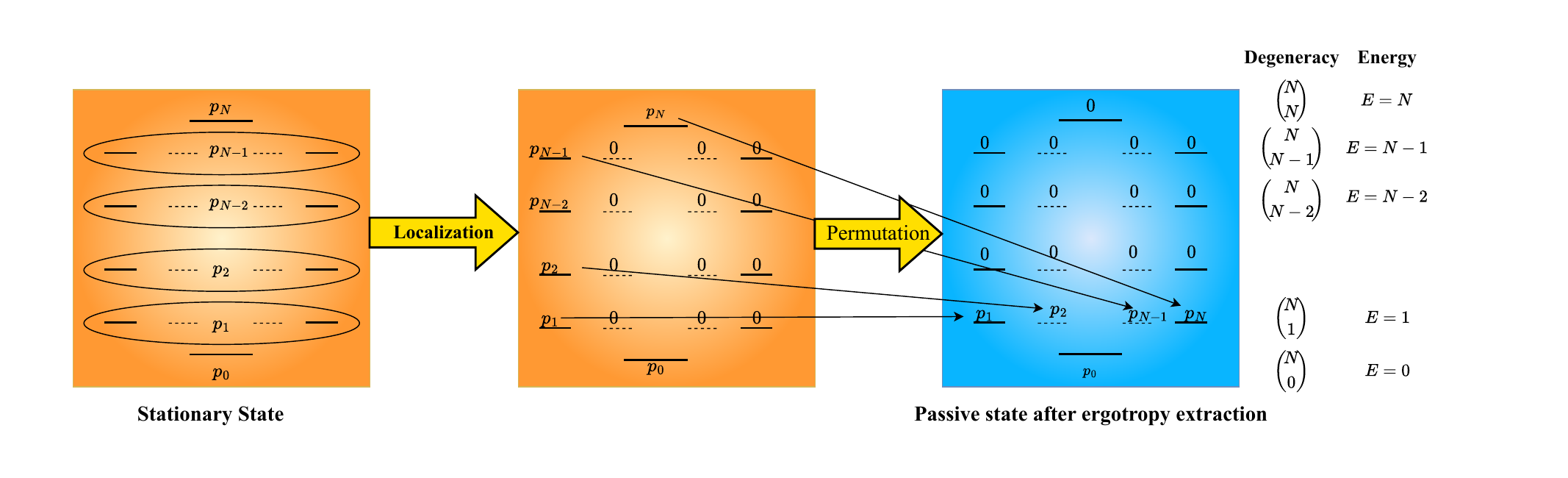}
\caption{Extraction of ergotropy $\mathcal{W}_B$ from the stationary state of the evolution initiated in the ground state (bottom line in each diagram, corresponding to state $|00\dots0\rangle$). Localized states of given energy are represented by lines, e.g. $|10\dots0\rangle$, $|01\dots0\rangle$, $\dots$, $|00\dots1\rangle$ for energy level $E=1$.   Evolution is restricted to the symmetric subspace and results in thermalization in the Dicke basis (left, with superpositions of the type $|10\dots0\rangle + |01\dots0\rangle + \dots + |00\dots1\rangle$ symbolized by oval contours). Then, localization of energy is obtained by a unitary acting independently in energy subspaces (middle), and a final permutation drives the system to the passive state.\justifying}\label{ErgoPic}
\end{figure}
\begin{figure}[t]
    \centering
    \begin{subfigure}[b]{0.32\textwidth}
        \includegraphics[width=\textwidth]{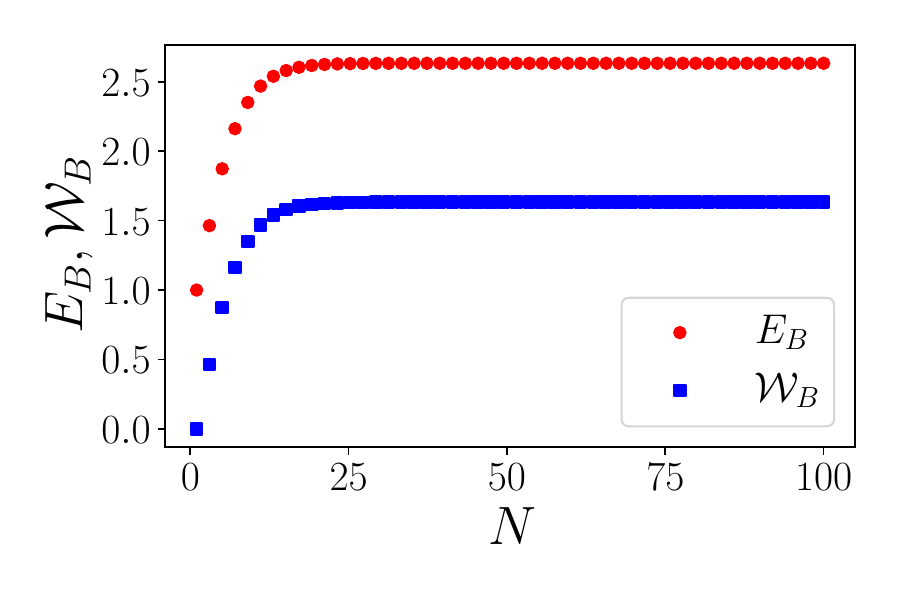}\caption{$\alpha_c=0.7$}
        \label{Energy-q=0.7}
    \end{subfigure}
    \begin{subfigure}[b]{0.32\textwidth}
\includegraphics[width=\textwidth]{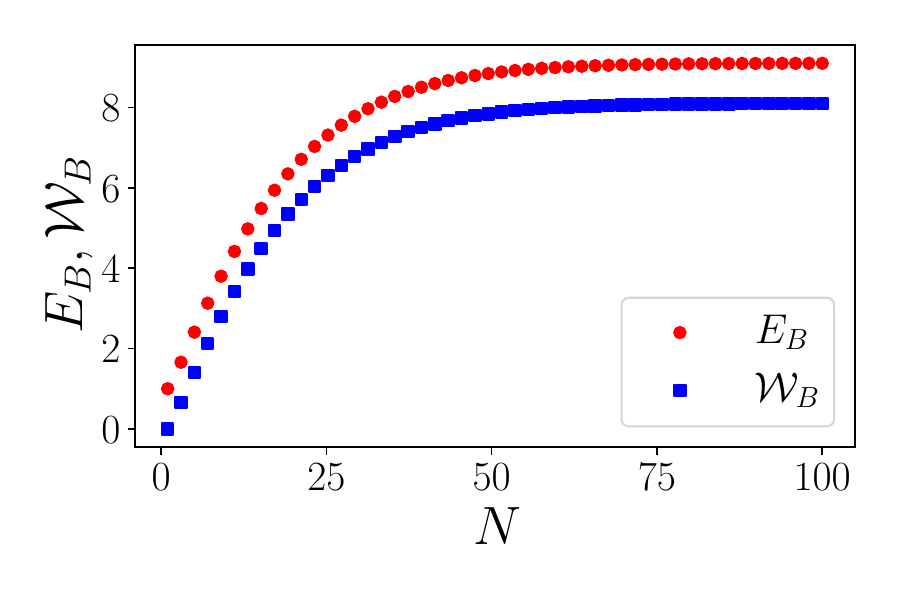}
        \caption{$\alpha_c=0.9$}\label{Energy-q=0.9}
    \end{subfigure}
    \begin{subfigure}[b]{0.32\textwidth}
\includegraphics[width=\textwidth]{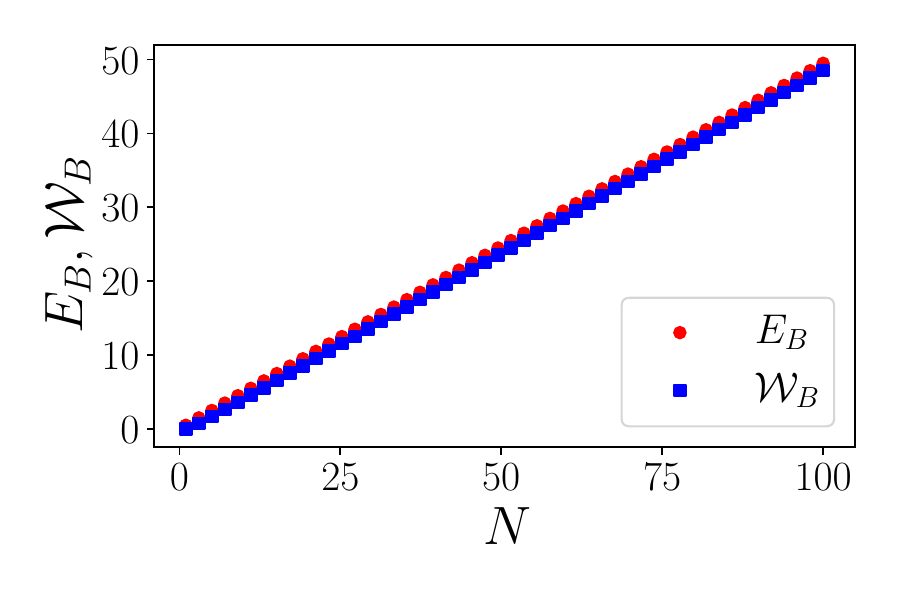}
        \caption{$\alpha_c=1$}\label{Energy-q=1}
    \end{subfigure}
    \captionsetup{justification=justified}
    \caption{Energy $E_{\mathrm{SS}}$ and ergotropy $\mathcal{W}$ of the QB in the steady state versus the number of cells $N$ inside the QB for various $\alpha_c$. As $\alpha_c\rightarrow1$ ($T_R\rightarrow\infty$) both ergotropy $\mathcal{W}$ and energy $E_{\mathrm{SS}}$ exhibit linear growth, illustrating their additivity behaviors. The collective character of thermal energy structure  suppresses completely the disadvantage of the noise. Here we have $\omega=1$ and $\gamma_c=0.05$.\justifying} 
    \label{Energy}
\end{figure}

\section{Optimal form of the initial state}\label{OptimalInitial}

In this section, we provide numerical evidence for the optimal form of the initial state of the dynamics, from the perspective of maximizing ergotropy in the stationary regime, with costs of initial state preparation taken into account. 
We start with the state of the form
\begin{align}
    \rho=\otimes_{i=1}^{N}\rho_{\beta_{q}},
\end{align}
where $\rho_{\beta_{q}}$ stands for local qubit Gibbs states with inverse temperature $\beta_{q}$ and
\begin{align}
    \rho_{\beta_{q}}=
\frac{1}{Z}\begin{pmatrix}
    1 & 0\\
    0 & e^{-\beta_{q}}
\end{pmatrix},
\end{align}
with $Z=1+e^{-\beta_{q}}$. A Haar-random unitary $U$ acting on the entire system is then applied to construct the initial state of the evolution:
\begin{align}
    \rho(t_0)=U\rho_{\beta{q}} U^{\dagger}
\end{align}
and ergotropic balance  $\Delta\mathcal{W}$ in the stationary limit is calculated taking into account energetic cost of the initial rotation
\begin{align}
    \Delta\mathcal{W}=\mathcal{W}(\lim_{t\rightarrow\infty}\rho(t))-\mathcal{W}(\rho(t_0))
\end{align}
(note that ergotropy of the product of Gibbs states is zero). In Fig. \ref{DeltaW} we present some more values of ergotropic balance for a system of $N=4$ qubits, with a series of random unitaries, local initialization inverse temperatures $\beta_{q}$, and an inverse temperature $\beta_c=\{0.1, 1\}$ of the common reservoir. 
\begin{figure}[htbp]
    \centering
    \begin{subfigure}{0.49\textwidth}
        \centering
        \includegraphics[width=\textwidth]{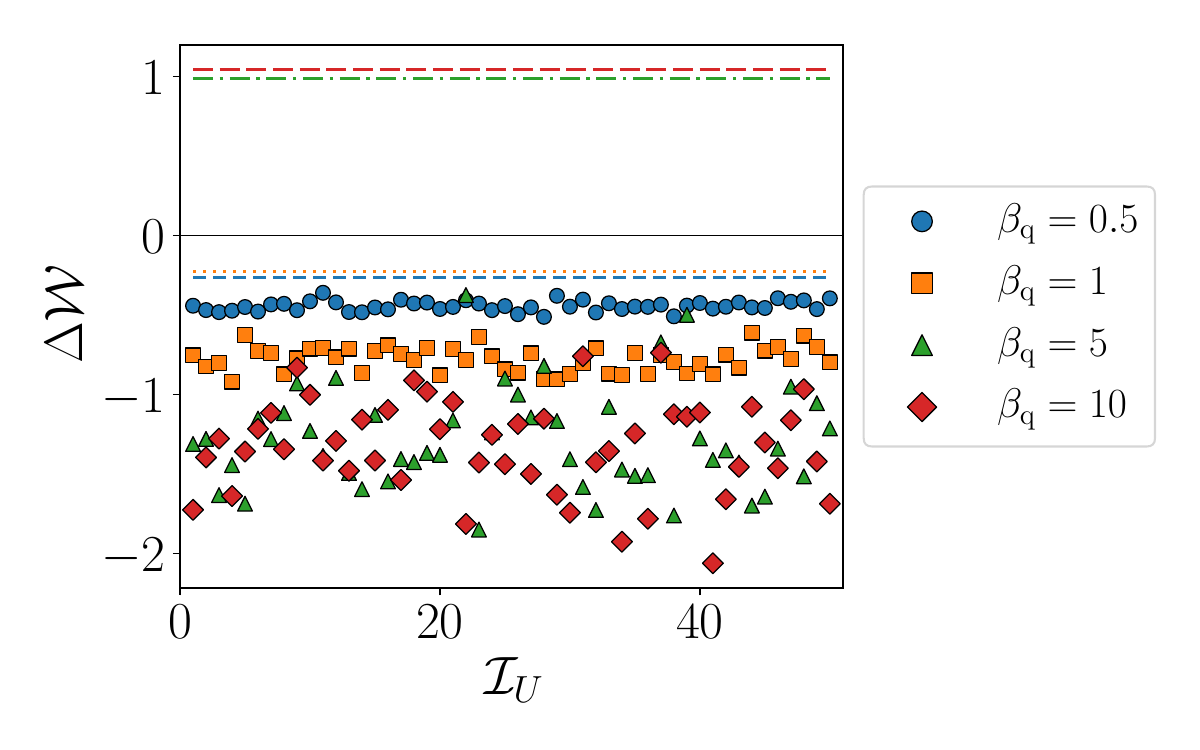}
        \caption{Ergotropic balance for $\beta_{c} = 0.1$}
        \label{fig:DW1}
    \end{subfigure}
    \hfill
    \begin{subfigure}{0.49\textwidth}
        \centering
        \includegraphics[width=\textwidth]{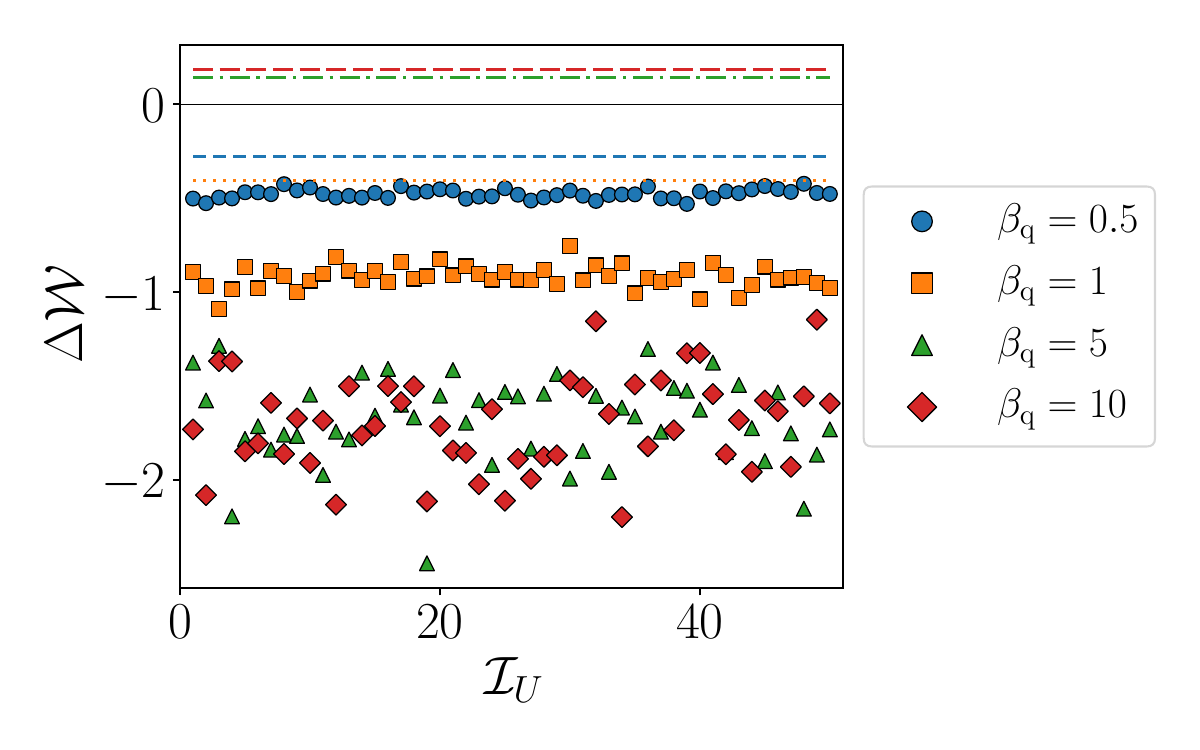}
        \caption{Ergotropic balance for $\beta_{c} = 1$}
        \label{fig:DW2}
    \end{subfigure}
    \caption{Further comparison of ergotropic balance $\Delta \mathcal{W}$ against the ordinality of the random unitary rotation $\mathcal{I}_{U}$ for different values of $\beta_{q} \in \{0.5, 1, 5, 10\}$ for collective bath temperatures $\beta_{c} = 0.1$ and $\beta_{c} = 1$.\justifying}
    \label{DeltaW}
\end{figure}

\section{Robustness against local noise}
\begin{figure}[t]
    \centering
    \begin{subfigure}[b]{0.32\textwidth}
        \includegraphics[width=\textwidth]{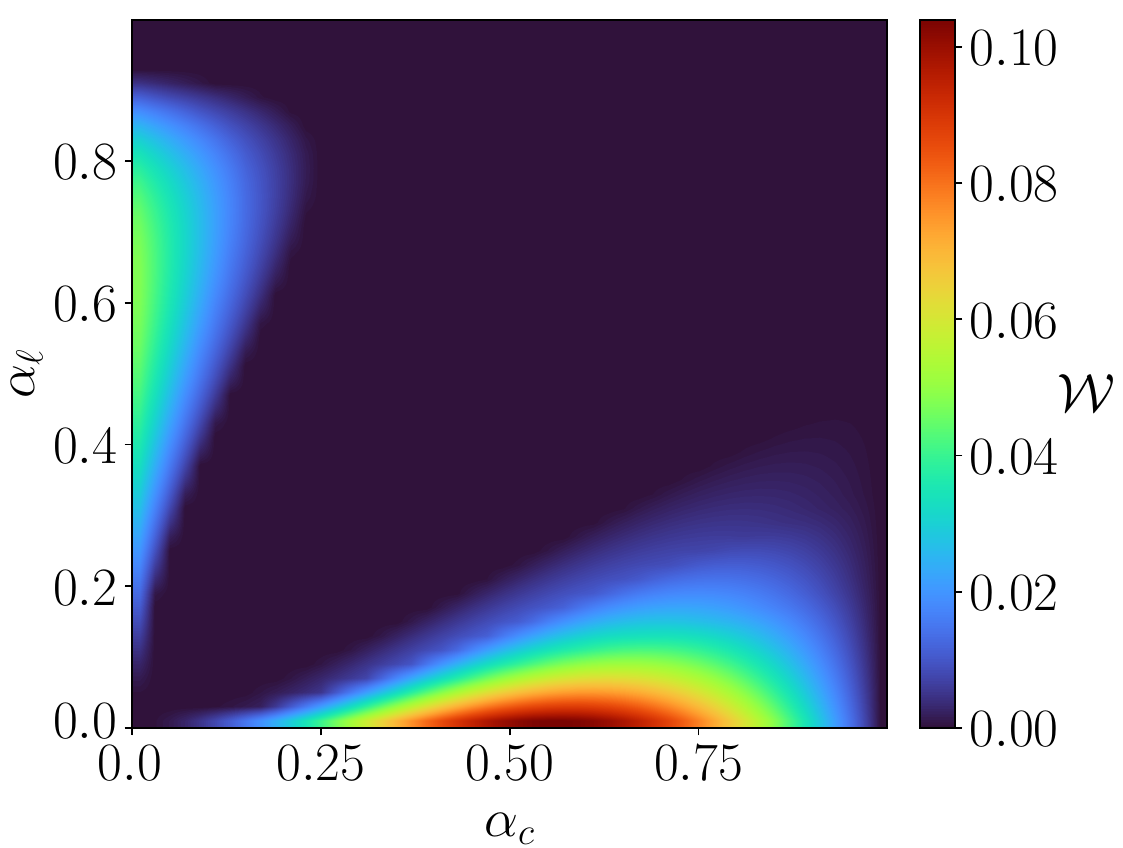}\caption{$N=4$}
        \label{N4_eta0.90}
    \end{subfigure}
    \begin{subfigure}[b]{0.32\textwidth}
\includegraphics[width=\textwidth]{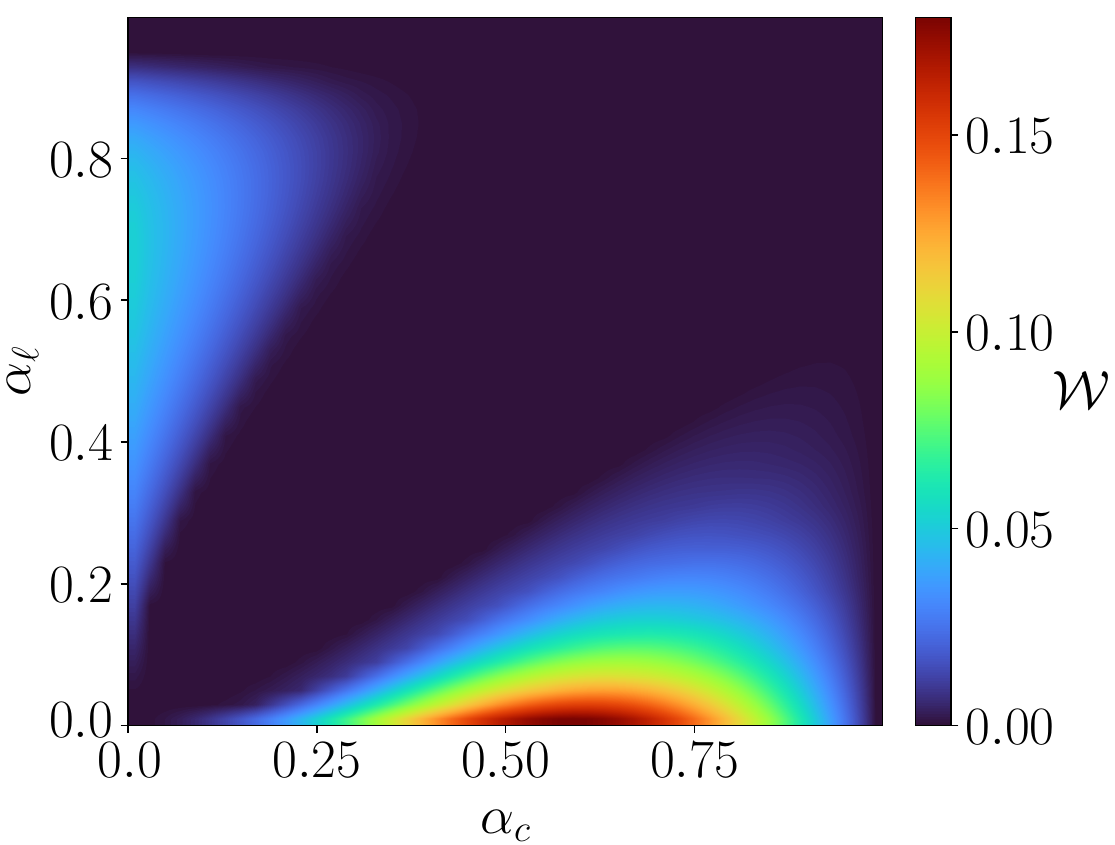}
        \caption{$N=5$}\label{N5_eta0.90}
    \end{subfigure}
    \begin{subfigure}[b]{0.32\textwidth}
\includegraphics[width=\textwidth]{N7_eta0.90.pdf}
        \caption{$N=7$}\label{N7_eta0.90}
    \end{subfigure}
    \captionsetup{justification=justified}
    \caption{\textbf{Finite-temperature optima under partial collectivity.}
    Steady-state ergotropy $\mathcal{W}$ for $N=4,5,7$, $\gamma_r=1$ (where $\gamma_r \equiv \gamma_\ell/\gamma_c$) with collective fraction $\eta=0.9$ as a function of the collective and local Bose ratios $(\alpha_c,\alpha_\ell)$. Unlike the ideal case, the maximum shifts to a finite \(\alpha_c^\star<1\) and a nonzero \(\alpha_\ell^\star>0\), reflecting a trade-off between interference-enabled collective pumping (favored by larger \(\alpha_c\)) and which-path information introduced by local channels (growing with \(\alpha_\ell\)) that degrades interference and drives the system toward passivity.\justifying} 
    \label{ErgoEta}
\end{figure}

\subsection{Partial collectivity (\(\eta=0.9\)) and finite-temperature optima}
With imperfections, we model the dynamics by the convex interpolation \(\mathcal L_\eta=\eta\,\mathcal L_{\rm coll}+(1-\eta)\,\mathcal L_{\rm local}\) (Sec.~“Robustness to imperfections” in the main text), where \(\mathcal L_{\rm coll}\) acts via \(J_\pm\) at Bose factor \(\alpha_c=e^{-\beta_c}\) and \(\mathcal L_{\rm local}\) via \(\{\sigma_\pm^{(i)}\}\) at \(\alpha_\ell=e^{-\beta_\ell}\).
The local channel breaks permutation symmetry, mixes total-spin sectors, and renders the evolution primitive \cite{Evans1977,Frigerio1978}, leading to a unique full-rank stationary state. Figure~\ref{ErgoEta} shows that, at fixed \(\eta=0.9\), the steady-state ergotropy \(\mathcal{W}\) is \emph{non-monotonic} in \((\alpha_c,\alpha_\ell)\): it peaks at a finite \(\alpha_c^\star<1\) and a nonzero \(\alpha_\ell^\star>0\).
This arises from a competition:
(i) increasing \(\alpha_c\) enhances collective upward pumping along bright ladders (which favors nonpassivity),
while (ii) increasing \(\alpha_\ell\) injects which-path information and local detailed balance that erode the interference-protected structure and drive the state toward passivity.
Consequently, when $\eta<1$ there is a trade-off: too small $\alpha_c$ underpopulates the bright ladders, while too large $\alpha_c$ in the presence of local noise overheats and undermines the interference-protected structure; likewise, a small but nonzero $\alpha_\ell$, with the help of interference caused by the collective reservoir, can facilitate transport across excitation sectors, but excessive local noise suppresses $\mathcal{W}$ (see Sec. \ref{overheating} for analytical detail). The lobes in Fig.~\ref{ErgoEta} at moderate–high $\alpha_c$ and small–moderate $\alpha_\ell$ captures precisely this compromise. As $N$ increases the interference within the QB increases paving the way for more energy transfer.

\subsection{Dependence of the activation lobe on the collective fraction \texorpdfstring{$\eta$}{eta}}

\subsubsection*{Activation-point shift for $\eta<0.8$}

Figure~\ref{ErgoEta_Shift} shows $\mathcal W$ versus $(\alpha_c,\alpha_\ell)$ for $\eta\in\{0.4,0.6,0.8\}$. A pronounced activation lobe appears along the $\alpha_c$-axis at small $\alpha_\ell$ in all panels.
As $\eta$ decreases below $0.8$, two effects are evident:
\begin{itemize}
\item \textbf{Peak shift.} The maximum moves to larger $\alpha_c$ (hotter common reservoir), indicating that reduced collectivity requires stronger collective pumping to populate the bright ladders.
\item \textbf{Nearly constant width along $\alpha_c$.} The \emph{span of the lobe in $\alpha_c$} remains approximately unchanged across the three panels, while its extent in $\alpha_\ell$ and its amplitude vary. Thus the dominant impact of lowering $\eta$ is to shift the optimum $\alpha_c^\star$ rather than to compress the lobe along the collective-temperature axis demonstrating robustness against imperfection.
\end{itemize}

This behavior reflects the balance in $\mathcal L_\eta=\eta\,\mathcal L_{\mathrm{coll}}+(1-\eta)\,\mathcal L_{\mathrm{local}}$: decreasing $\eta$ weakens the net collective pumping (shifting $\alpha_c^\star$ rightward) but does not substantially alter the \emph{range of $\alpha_c$} over which collective pumping outcompetes local mixing for small–moderate $\alpha_\ell$.
\begin{figure}[t]
    \centering
    \begin{subfigure}[b]{0.32\textwidth}
        \includegraphics[width=\textwidth]{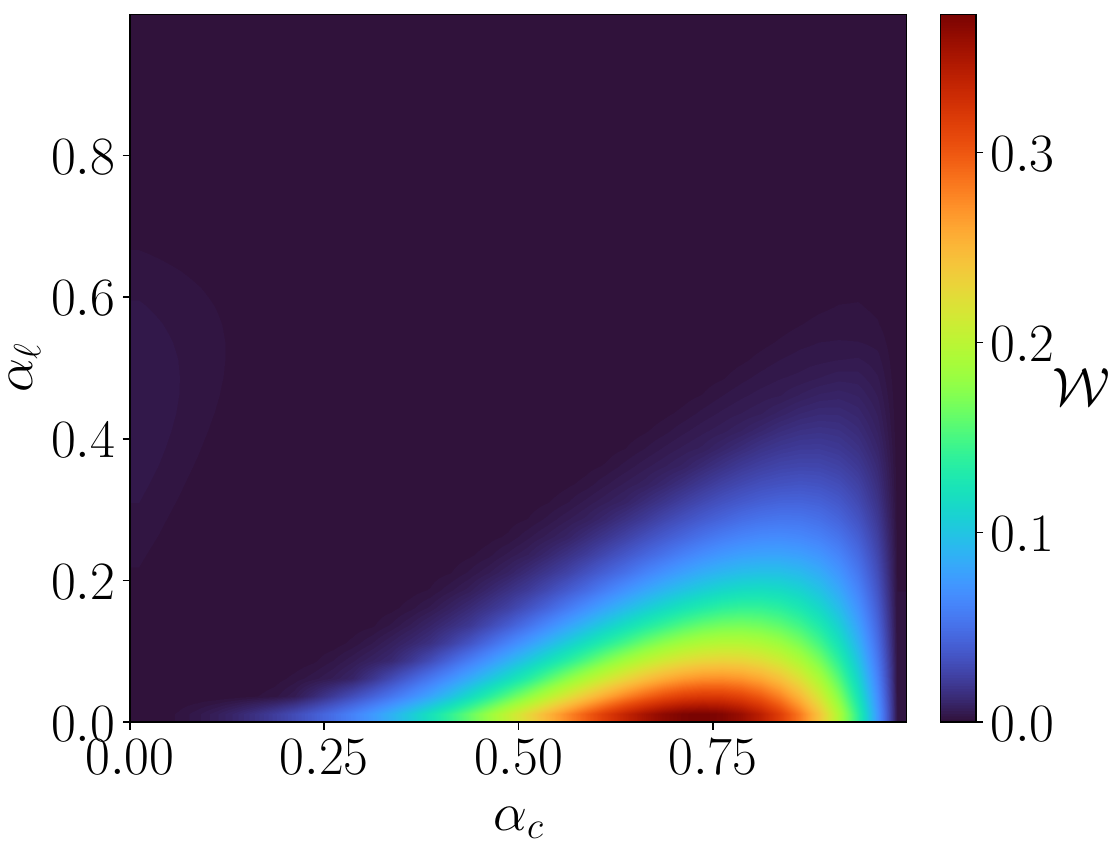}\caption{$\eta=0.4$}
        \label{N7_eta0.40}
    \end{subfigure}
    \begin{subfigure}[b]{0.32\textwidth}
\includegraphics[width=\textwidth]{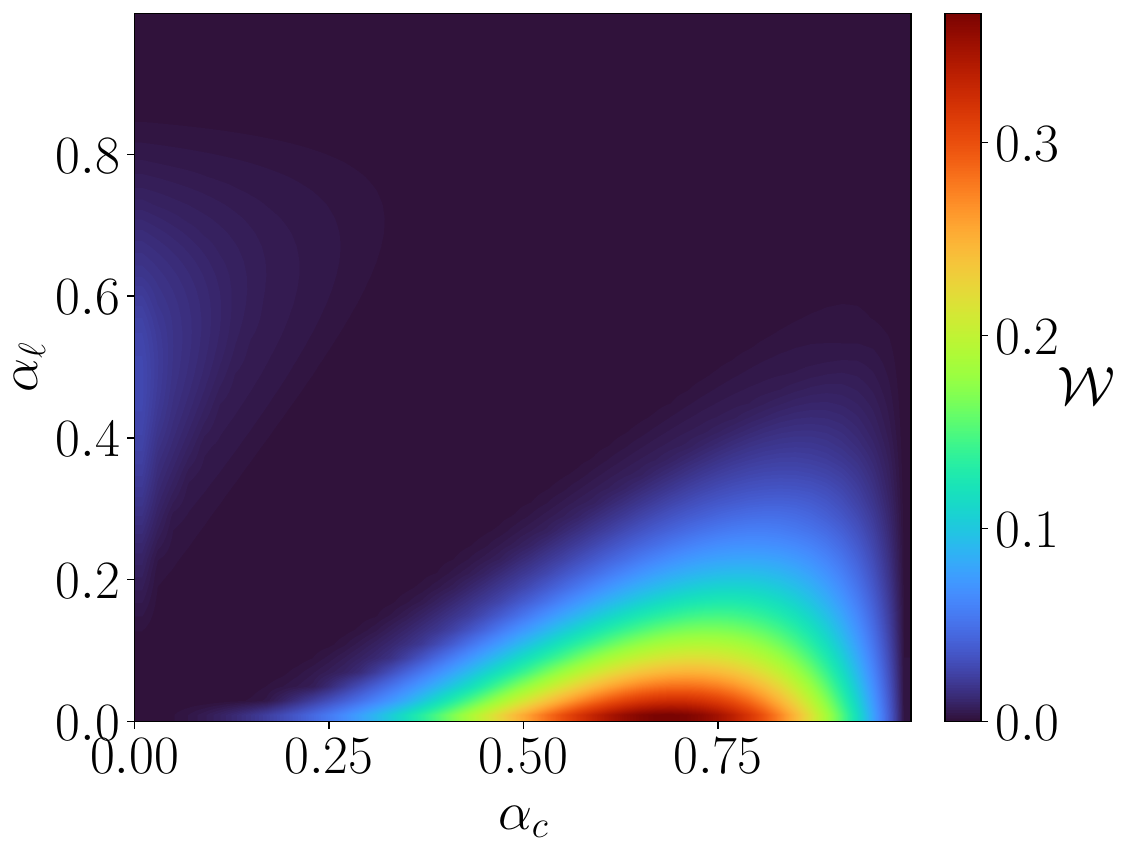}
        \caption{$\eta=0.6$}\label{N6_eta0.60}
    \end{subfigure}
    \begin{subfigure}[b]{0.32\textwidth}
\includegraphics[width=\textwidth]{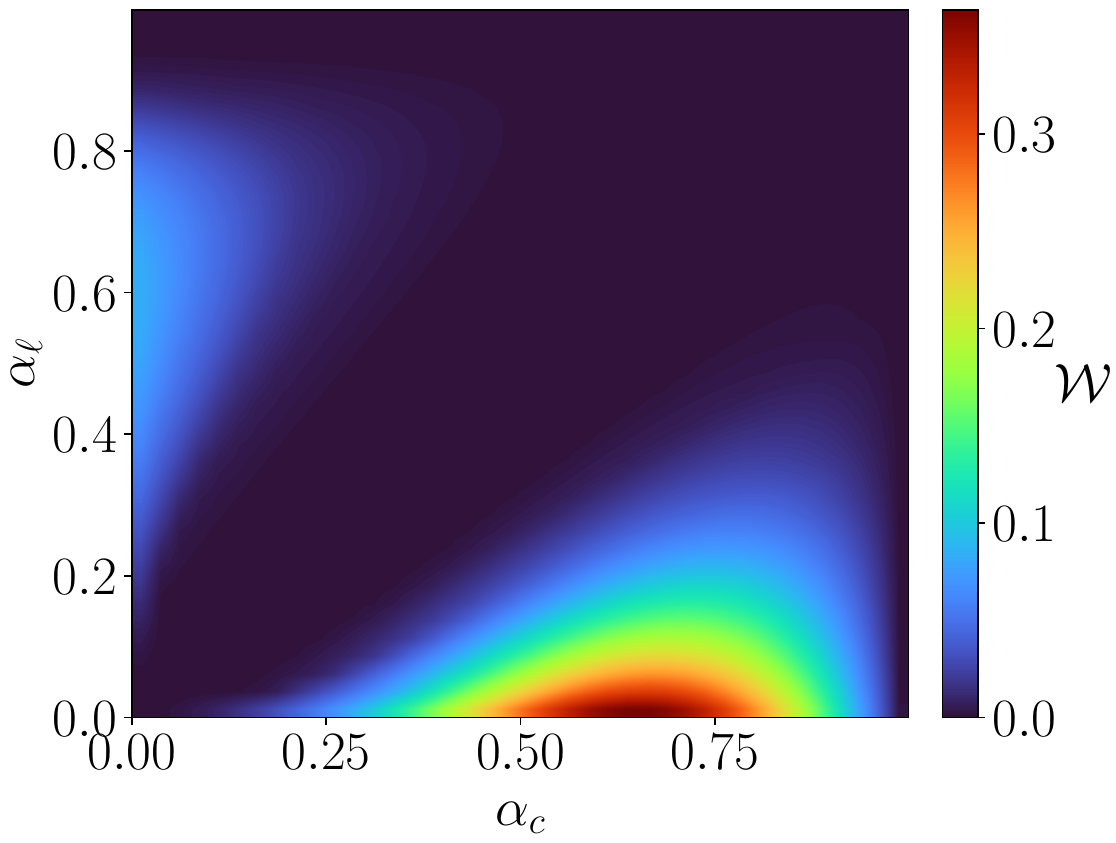}
        \caption{$\eta=0.8$}\label{N7_eta0.80}
    \end{subfigure}
    \captionsetup{justification=justified}
    \caption{\textbf{Activation shift with nearly constant $\alpha_c$ width as collectivity is reduced.}
    Steady-state ergotropy $\mathcal W$ for $N=7$, $\gamma_r=1$ versus $(\alpha_c,\alpha_\ell)$ for (a) $\eta=0.4$, (b) $\eta=0.6$, and (c) $\eta=0.8$.
    Lower $\eta$ shifts the activation maximum to higher $\alpha_c$ and lowers its amplitude, while the lobe’s width along the $\alpha_c$ axis remains approximately constant. 
    Hence, reduced collectivity mainly requires a hotter common reservoir to compensate for weaker collective pumping; the tolerance to $\alpha_c$ is largely preserved, whereas the dependence on $\alpha_\ell$ and the peak height are more sensitive to $\eta$.\justifying} 
    \label{ErgoEta_Shift}
\end{figure}

\subsubsection{No activation-point shift for $\eta\gtrsim0.8$}

Figure~\ref{ErgoEta_NShift} displays the steady-state ergotropy $\mathcal W_{\mathrm{SS}}$ as a function of the collective and local Bose ratios $(\alpha_c,\alpha_\ell)$ for three values of the collective fraction, $\eta\in\{0.8,0.9,0.95\}$. In all panels, a pronounced \emph{activation lobe} appears along the $\alpha_c$-axis at small $\alpha_\ell$, reflecting ergotropy generated by the collective channel $\mathcal L_{\mathrm{coll}}$ (via $J_\pm$) when which-path noise from the local channel is weak.
As $\eta$ decreases (from right to left), two systematic effects are visible:
\begin{itemize}
\item \textbf{Tolerance to local noise shrinks.} The lobe narrows in the $\alpha_\ell$ direction, indicating a reduced range of local temperatures that still preserves interference and nonpassivity when the dynamics contains a larger local component.
\item \textbf{The optimal collective temperature is only weakly $\eta$-dependent.} The location of the maximum, $\alpha_c^\star$, shifts only slightly (here within $\alpha_c^\star \approx 0.6$) as $\eta$ varies from $0.95$ to $0.8$ again reflecting the robustness against imperfection. Thus, once collectivity is reasonably high ($\eta\gtrsim 0.8$), the activation point is primarily set by the collective ladder structure.
\end{itemize}
These trends follow from the competition in $\mathcal L_\eta=\eta\,\mathcal L_{\mathrm{coll}}+(1-\eta)\,\mathcal L_{\mathrm{local}}$: collective pumping along Dicke ladders (favored by larger $\alpha_c$) versus local, which-path–resolving jumps that mix symmetry sectors and drive the state toward passivity as $\alpha_\ell$ or $1-\eta$ increase.
\begin{figure}[t]
    \centering
    \begin{subfigure}[b]{0.32\textwidth}
        \includegraphics[width=\textwidth]{N7_eta0.80.pdf}\caption{$\eta=0.8$}
        \label{N7_eta0.80}
    \end{subfigure}
    \begin{subfigure}[b]{0.32\textwidth}
\includegraphics[width=\textwidth]{N7_eta0.90.pdf}
        \caption{$\eta=0.9$}\label{N7_eta0.90}
    \end{subfigure}
    \begin{subfigure}[b]{0.32\textwidth}
\includegraphics[width=\textwidth]{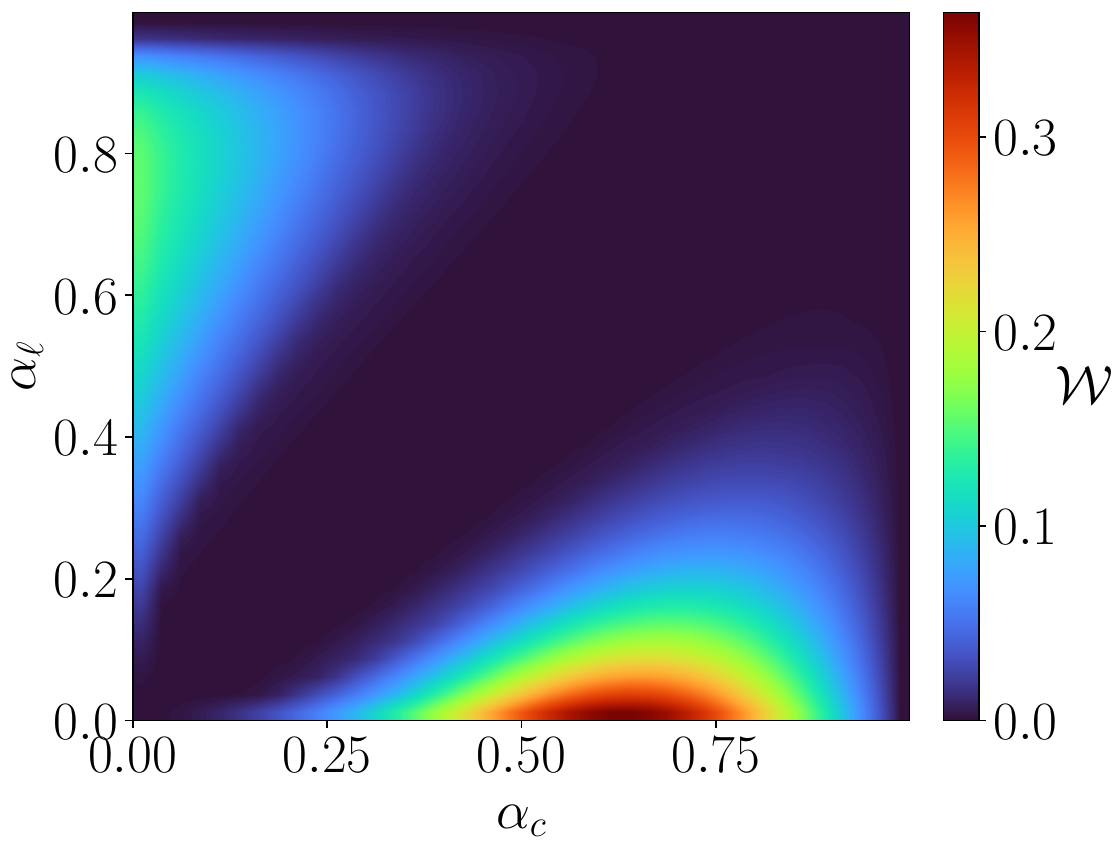}
        \caption{$\eta=0.95$}\label{N7_eta0.95}
    \end{subfigure}
    \captionsetup{justification=justified}
    \caption{\textbf{Dependence of the activation lobe on the collective fraction \(\eta\).}
    Steady-state ergotropy \(\mathcal W\) for $N=7$, $\gamma_r=1$ versus collective and local Bose ratios \((\alpha_c,\alpha_\ell)\) for (a) \(\eta=0.8\), (b) \(\eta=0.9\), and (c) \(\eta=0.95\).
    A pronounced activation lobe appears along the \(\alpha_c\)-axis at small \(\alpha_\ell\), evidencing ergotropy generated by the collective channel.
    As \(\eta\) decreases, the lobe on \(\alpha_\ell\) narrows (reduced tolerance to local which-path noise) and its peak value drops, while the optimal collective setting \(\alpha_c^\star\) moves only slightly within a narrow window (\(\alpha_c^\star \approx 0.6\)).\justifying} 
    \label{ErgoEta_NShift}
\end{figure}

\subsection{Effect of the local-to-collective rate ratio $\gamma_r$}

Figure~\ref{ErgoEta_gamma_r} shows $\mathcal W_{\mathrm{SS}}$ for $N=7$, $\gamma_r\in\{0.01,0.1,1\}$ where $\gamma_r \equiv \frac{\gamma_\ell}{\gamma_c}$. Two systematic trends are evident:
\begin{itemize}
\item \emph{Nearly no shift of the optimum:} the optimal collective setting $\alpha_c^\star$ is weakly affected as $\gamma_r$ grows.
\item \emph{Selective sensitivity to $\alpha_\ell$:} the lobe retains an approximately constant span along $\alpha_c$, while its extent in $\alpha_\ell$ shrinks; thus small–moderate local temperatures remain compatible with activation.
\end{itemize}
A further observation from Fig.~\ref{ErgoEta_gamma_r} is the contrasting role of the local channel. For weak local dissipation ($\gamma_\ell\ll\gamma_c$; see Figs.~\ref{N7_eta0.90_r0.01}–\ref{N7_eta0.90_r0.1}), the activation lobe along the $\alpha_\ell$-axis is broad and resilient: suppression sets in only at relatively large local temperatures $\alpha_\ell$. As $\gamma_\ell$ increases (compare $\gamma_r=0.01\to 1$; Fig.~\ref{N7_eta0.90_r1}), the same lobe is quenched much earlier—moderate $\alpha_\ell$ already erases the activation. This trend is consistent with the overheating analysis in Sec.~\ref{overheating}; notably, the $\alpha_c$ lobe’s location and width are almost unaffected. Together, these observations quantify the robustness window: ergotropy activation is sustained when the collective channel remains comparable to or stronger than the local one ($\gamma_r\lesssim1$), and it degrades smoothly as $\gamma_r$ increases.

\begin{figure}[t]
    \centering
    \begin{subfigure}[b]{0.32\textwidth}
        \includegraphics[width=\textwidth]{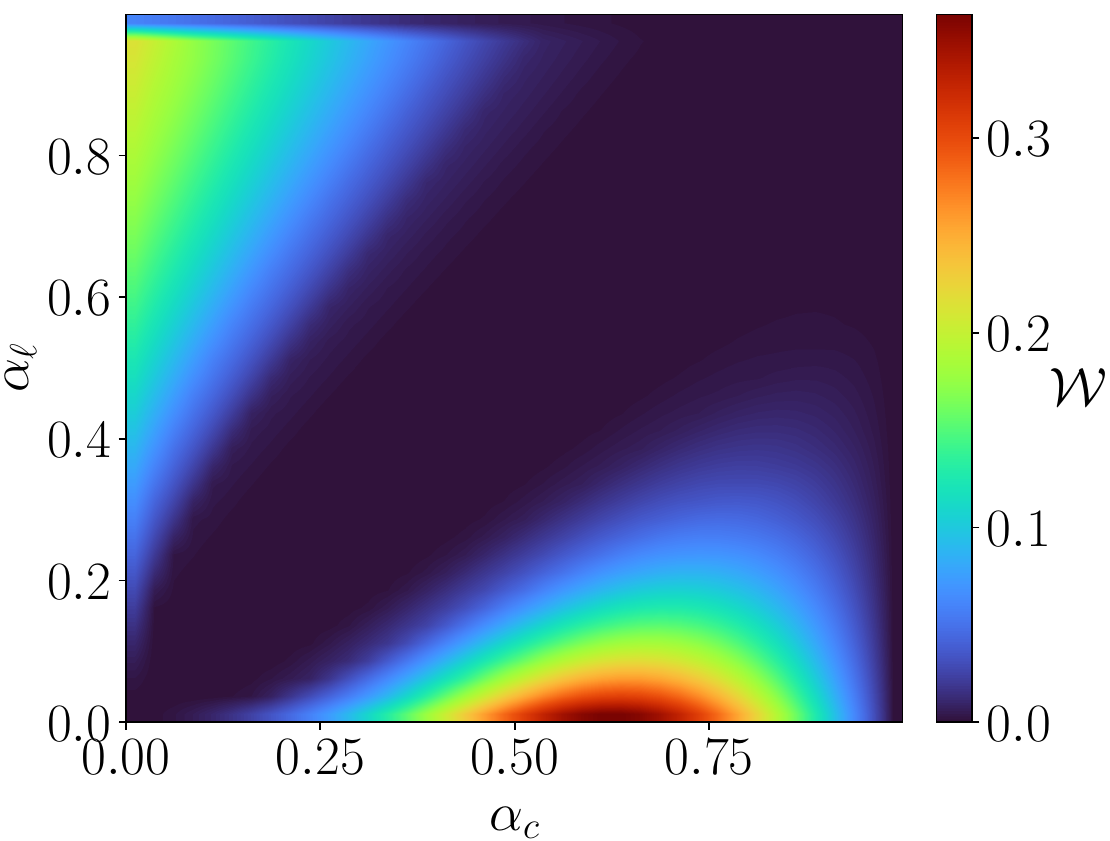}\caption{$\gamma_r=0.01$}
        \label{N7_eta0.90_r0.01}
    \end{subfigure}
    \begin{subfigure}[b]{0.32\textwidth}
\includegraphics[width=\textwidth]{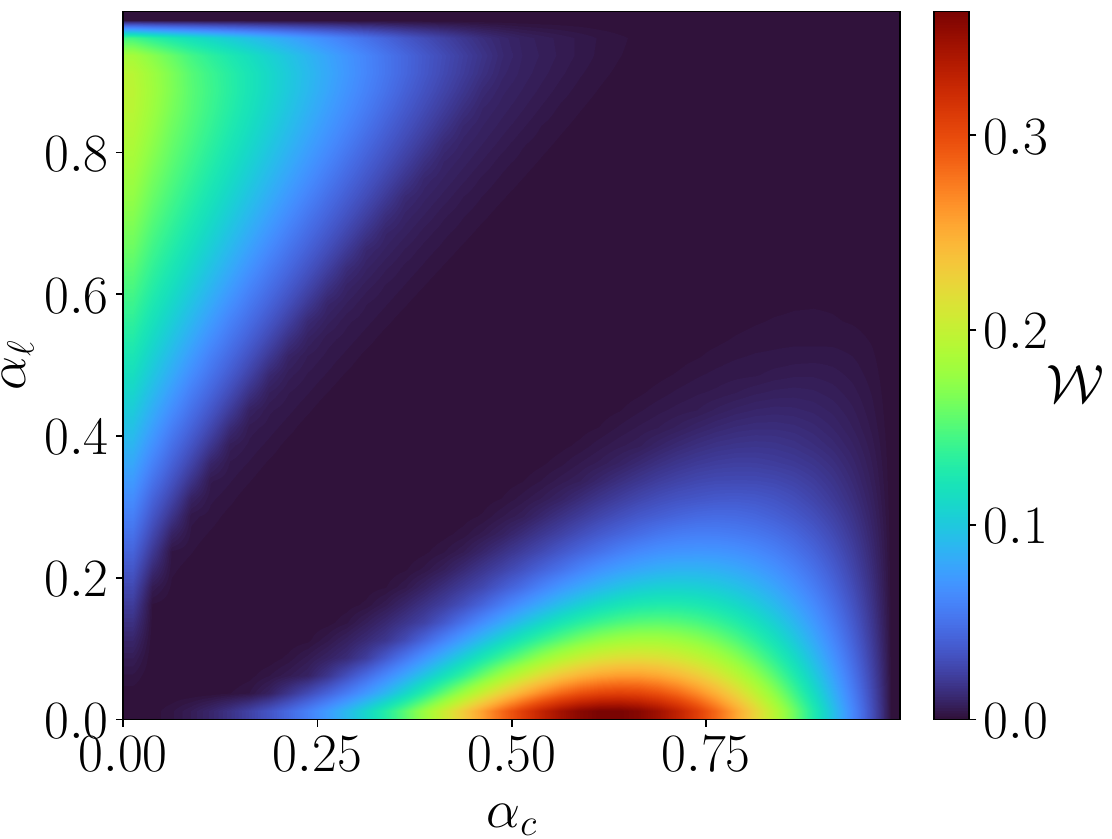}
        \caption{$\gamma_r=0.1$}\label{N7_eta0.90_r0.1}
    \end{subfigure}
    \begin{subfigure}[b]{0.32\textwidth}
\includegraphics[width=\textwidth]{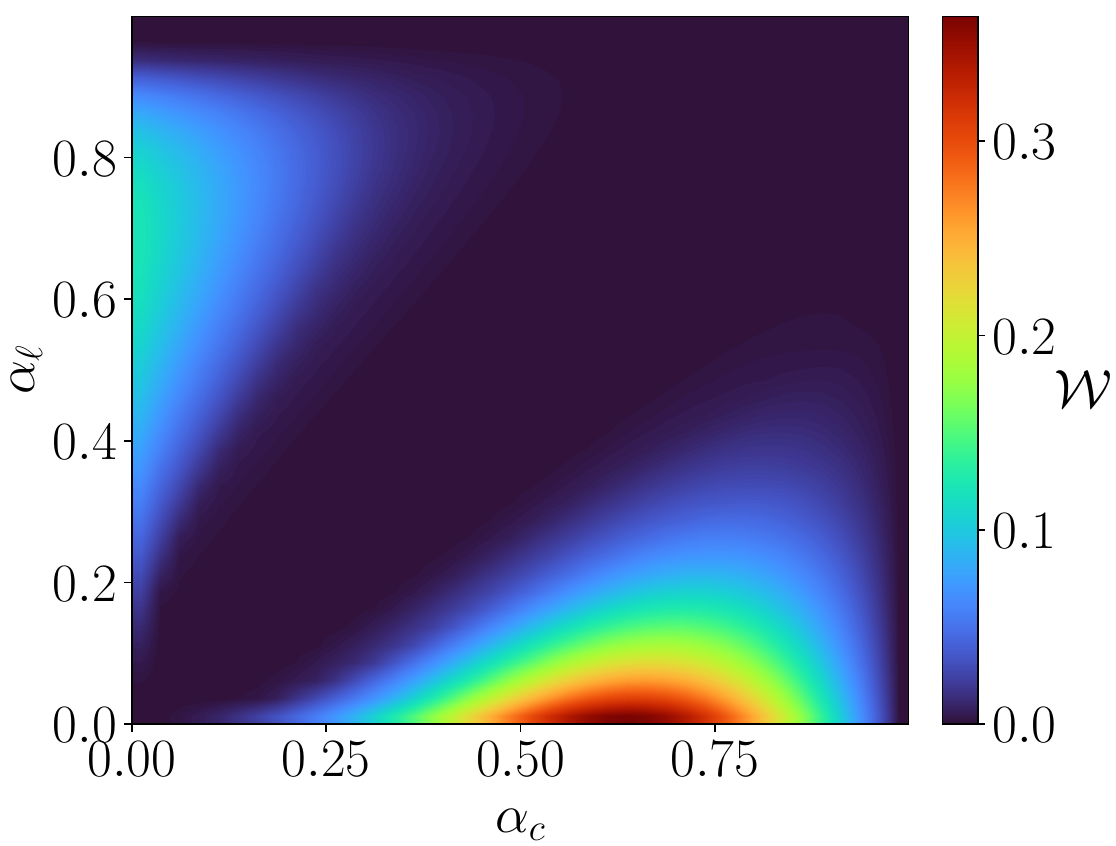}
        \caption{$\gamma_r=1$}\label{N7_eta0.90_r1}
    \end{subfigure}
    \captionsetup{justification=justified}
    \caption{\textbf{Competition between collective and local dissipation.}
    Steady-state ergotropy $\mathcal W$ versus the collective and local Bose ratios $(\alpha_c,\alpha_\ell)$ for $\eta = 0.9$, $N=7$ and three values of the dissipation ratio:
    (a) $\gamma_r=0.01$, (b) $\gamma_r=0.1$, and (c) $\gamma_r=1$.\justifying} 
    \label{ErgoEta_gamma_r}
\end{figure}

\section{The overheating mechanism}\label{overheating}

In this section we give a derivation of the leakage caused by the local (site-resolved) dissipator out of the fully symmetric Dicke sector and show how this induces a finite-temperature optimum of the collective bath when $\eta<1$.

\paragraph*{Setting and notation.}
Let $N$ two-level systems evolve under the GKLS generator
\begin{equation}
\mathcal L \;=\; \eta\,\gamma_c\!\left[(n_c{+}1)\mathcal D[J_-]+n_c\,\mathcal D[J_+]\right]
\;+\; (1{-}\eta)\,\gamma_\ell \sum_{i=1}^N\!\left[(n_\ell{+}1)\mathcal D[\sigma_-^{(i)}]+n_\ell\,\mathcal D[\sigma_+^{(i)}]\right]
\;-\; i[H,\cdot],
\label{Liov}
\end{equation}
with $\mathcal D[L]\rho=L\rho L^\dagger-\tfrac12\{L^\dagger L,\rho\}$, $0\le \eta\le1$, $\gamma_c,\gamma_\ell>0$, and Bose factors $n_x=\frac{\alpha_x}{(1-\alpha_x)}\in[0,\infty)$. The Hamiltonian $H$ is bounded and permutation covariant (e.g. $H_B{+}$Lamb shift), so it preserves total spin $j$ and does not contribute to leakage between SU(2) irreps. Let $P_{\rm sym}$ denote the projector onto the fully symmetric SU(2) irrep $j=\frac{N}{2}$ (the “bright” sector) $P_{\rm sym}=\mathbb I_{\{j(j+1)\}}(J^2)$, and $Q_{\rm sym}=\mathbb I-P_{\rm sym}$ its complement. Define the excitation-number operator
\begin{equation}
\hat{K} \;:=\; \sum_{i=1}^N \sigma_+^{(i)}\sigma_-^{(i)}, \qquad 0\le K\le N,
\label{excitation-K}
\end{equation}
where $K=\langle\hat{K}\rangle_\rho$ and the \emph{leakage functional} generated by $\dot\rho=\mathcal L[\rho]$,
\begin{equation}
\Lambda(\rho(t)) \;:= -\frac{d}{dt}\tr\!\big(P_{\rm sym}\rho(t)P_{\rm sym}\big)
=-\tr\!\big(P_{\rm sym}\mathcal L[\rho(t)]P_{\rm sym}\big).
\end{equation}
Thus $\Lambda(\rho(t))$ is the instantaneous loss rate of bright-sector weight \emph{at time $t$}.
\begin{lemma}\label{PJcommute}
Let $J_\pm=\sum_{i=1}^N \sigma_\pm^{(i)}$ and $J^2:=J_x^2+J_y^2+J_z^2$.
Let $P_{\rm sym}$ be the spectral projector of $J^2$ onto the eigenspace with
$j=\tfrac N2$ (the fully symmetric subspace). Then
\[
[P_{\rm sym},\,J_\pm]=0.
\]
\end{lemma}
\begin{proof}
The SU(2) commutation relations imply
$[J_z,J_\pm]=\pm J_\pm$ and $[J_+,J_-]=2J_z$, from which one obtains
\[
[J^2,\,J_\pm]=\big[J_x^2+J_y^2+J_z^2,\,J_\pm\big]=0.
\]
Hence $J_\pm$ commute with $J^2$ and leave every spectral subspace of $J^2$
invariant. By functional calculus, every bounded Borel function $f(J^2)$
commutes with $J_\pm$; in particular the spectral projector$P_{\rm sym}=\mathbb I_{\{j(j+1)\}}(J^2)$ onto the $j=\tfrac N2$ eigenspace
commutes with $J_\pm$. Therefore $[P_{\rm sym},J_\pm]=0$.
\end{proof}
Define the bright weight $p_{\rm sym}=\tr(P_{\rm sym}\rho P_{\rm sym})\in[0,1]$ and the normalized bright component $\rho_{\rm b}=\frac{P_{\rm sym}\rho P_{\rm sym}}{p_{\rm sym}}$. It should be noted that only the local channel changes the $p_{\rm sym}$ since
\begin{equation}
    \tr(P_{\rm sym}\mathcal{D}[J_{\mp}](\rho)P_{\rm sym}) = \tr(P_{\rm sym}J_{\mp}\rho J_{\pm} P_{\rm sym}) - \frac{1}{2}\tr(P_{\rm sym}J_{\pm}J_{\mp}P_{\rm sym}\rho P_{\rm sym}) - \frac{1}{2}\tr(P_{\rm sym}\rho P_{\rm sym}J_{\pm}J_{\mp}) = 0.
\end{equation}
Therefore, we have
\begin{equation}\label{lambdaP}
\Lambda(\rho(t)) \;= -\tr\!\big(P_{\rm sym}\mathcal L_\ell[\rho(t)]P_{\rm sym}\big).
\end{equation}
\begin{lemma}[Lowering-induced leakage]
\label{LemLower}
On the bright ladder $j=\frac{N}{2}$,
\begin{equation}
\sum_{i=1}^N \tr\!\big(P_{\rm sym}\,\mathcal D[\sigma_-^{(i)}]\rho P_{\rm sym}\big)
\;=\; -\frac{p_{\rm sym}}{N}\,\big\langle \hat{K}(\hat{K}{-}1)\big\rangle_{\rho_b}.
\label{Plower}
\end{equation}
\end{lemma}
\begin{proof}
Permutation covariance on the bright irrep implies
$P_{\rm sym}\sigma_-^{(i)}P_{\rm sym}=\tfrac{1}{N}P_{\rm sym}J_-P_{\rm sym}$ for each $i$. Hence
\begin{align*}
\sum_i \tr\!\big(P_{\rm sym}\sigma_-^{(i)}\rho\,\sigma_+^{(i)}P_{\rm sym}\big)
= \frac{1}{N}\tr\!\big(P_{\rm sym}J_-\,\rho\,J_+P_{\rm sym}\big)
= \frac{1}{N}\tr\!\big(J_-\,P_{\rm sym}\rho P_{\rm sym}\,J_+\big)
= \frac{p_{\rm sym}}{N}\,\langle J_+J_-\rangle_{\rho_{\rm b}},
\end{align*}
where in the last inequality we used Lemma \ref{PJcommute}.  Since $\sum_i^N \tr(P_{\rm sym}\sigma_+^{(i)}\sigma_-^{(i)}\rho\,P_{\rm sym})=\tr(\hat{K}P_{\rm sym}\rho\,P_{\rm sym})=p_{\rm sym}\langle \hat{K}\rangle_{\rho_b}$,
\[
-\frac{1}{2}\sum_{i=1}^N \tr\left(P_{\rm sym}\left\{\sigma_+^{(i)}\sigma_-^{(i)}, \, \rho\right\}P_{\rm sym}\right) = -p_{\rm sym}\langle \hat{K}\rangle_{\rho_b},
\]
where we used the fact that $[\hat{K},P_{\rm sym}]=0$ since $\hat{K}=J_z + \frac{N}{2}\mathbb I$. On the bright irrep $j=\frac{N}{2}$, one has the Dicke identity $J_+J_-=\hat{K}(N-\hat{K}+1)$ (operator identity when restricted to $P_{\rm sym}$). Substituting yields Eq. \eqref{Plower}.
\end{proof}
\begin{lemma}[Raising-induced leakage]\label{LemRaise}
On the bright ladder $j=\frac{N}{2}$,
\begin{equation}
\sum_{i=1}^N \tr\!\big(P_{\rm sym}\,\mathcal D[\sigma_+^{(i)}]\rho P_{\rm sym}\big)
\;=\; -\frac{p_{\rm sym}}{N}\,\big\langle \hat{H}(\hat{H}-1)\big\rangle_{\rho_b}.
\label{Praise}
\end{equation}
\end{lemma}
with $\hat{H}:=N-\hat{K}$ the "hole"-number operator.
\begin{proof}
Identical to Lemma~\ref{LemLower}, using
$P_{\rm sym}\sigma_+^{(i)}P_{\rm sym}=\tfrac{1}{N}P_{\rm sym}J_+P_{\rm sym}$ and the identity $J_-J_+=(\hat{K}{+}1)(N{-}\hat{K})$ on the bright irrep.
\end{proof}
\noindent\emph{Identities on $j=\tfrac N2$.} On the fully symmetric irrep one has $\hat{K}=J_z+\tfrac N2$ and $J_+J_- = \hat{K}(N-\hat{K}+1)$, $J_-J_+=(\hat{K}+1)(N-\hat{K})$. Moreover, by permutation covariance and Schur’s lemma,
$P_{\rm sym}\sigma_\mp^{(i)}P_{\rm sym}=\tfrac1N P_{\rm sym}J_\mp P_{\rm sym}$ for every site $i$.
\begin{proposition}\label{proposition}
The functional leakage from the bright sector is described as
\begin{equation}\label{functional}
\boxed{\ 
\Lambda(\rho)
=\frac{(1-\eta)p_{\rm sym}\gamma_\ell}{N}\left[(n_\ell{+}1)\,\big\langle \hat{K}(\hat{K}{-}1)\big\rangle_{\rho_b}
+ n_\ell\,\big\langle \hat{H}(\hat{H}{-}1)\big\rangle_{\rho_b} \right].
\ }
\end{equation}
\end{proposition}
\begin{proof}
Inserting Eqs. \eqref{Plower}–\eqref{Praise} into Eq. \eqref{lambdaP} we readily get Eq. \eqref{functional}.
\end{proof}
Proposition \ref{proposition} states that the leakage from the bright sector increases with the excitation density via the term $\langle \hat{K}(\hat{K}-1)\rangle{\rho_b}$ when the state is close to the top of the ladder. By particle–hole symmetry, it likewise increases with hole density via $\big\langle \hat{H}(\hat{H}{-}1)\big\rangle_{\rho_b}$ when the state is close to the bottom of the ladder.

\textbf{Jensen's inequality}:
Since $\mathrm{Var}(\hat{K})=\mathbb{E}[\hat{K}^2]-\mathbb{E}[\hat{K}]^2\ge 0$, where $\mathbb{E}[X]$ is the expectation value (mean) of the random variable $X$,
\[
\mathbb{E}[\hat{K}(\hat{K}-1)]
=\mathbb{E}[\hat{K}^2]-\mathbb{E}[\hat{K}]
\;\ge\; \mathbb{E}[\hat{K}]^2-\mathbb{E}[\hat{K}]
=\mathbb{E}[\hat{K}]\;\big(\mathbb{E}[\hat{K}]-1\big),
\]
and similarly with $K$ replaced by $N{-}K$:
\[
\mathbb{E}\!\big[(N{-}\hat{K})(N{-}\hat{K}{-}1)\big]
=\mathbb{E}\!\big[(N{-}\hat{K})^2\big]-\mathbb{E}[N{-}\hat{K}]
\;\ge\; (N-\mathbb{E}[\hat{K}])^2-(N-\mathbb{E}[\hat{K}]).
\]

\noindent Since $K, H\in\{0,\dots,N\}$, where $H=\langle \hat{H}\rangle_{\rho_b}$, Jensen’s inequality gives
\begin{equation}
\big\langle \hat{K}(\hat{K}{-}1)\big\rangle_{\rho_b} \ge \langle \hat{K}\rangle_{\rho_b}(\langle \hat{K}\rangle_{\rho_b}{-}1),\qquad
\big\langle \hat{H}(\hat{H}{-}1)\big\rangle_{\rho_b} \ge \langle \hat{H}\rangle_{\rho_b}(\langle \hat{H}\rangle_{\rho_b}{-}1),
\label{Jensen}
\end{equation}
hence $\Lambda(\rho)$ is strictly increasing as the excitation density $K$ grows for $K\gtrsim1$ when near the top of the ladder, and similarly increases as the hole density $H$ grows when near the bottom of the ladder, while near the bottom: small $K$ and large $H$. Near the top: large $K$ and small $H$. In sum, by particle–hole symmetry, the first term dominates and grows as the excitation density increases toward the top ($K\rightarrow N$), while the second term dominates and grows as the hole density increases toward the bottom ($H\rightarrow N$).

In the “underheated” regime $\langle \hat{H}\rangle_{\rho_b}=\Theta\!\left(\,N\right)$ (scales linearly with $N$), with proportionality also linear in $\gamma_\ell$ and $n_\ell$, and "overheated" regime $\langle \hat{K}\rangle_{\rho_b}=\Theta\!\left(\,N\right)$, with proportionality also linear in $\gamma_\ell$ and $2n_\ell+1$ as well. Therefore Eq. \eqref{functional}, in these regimes, scales respectively as
\begin{equation}\label{scaling}
\begin{aligned}
\Lambda(\rho) &= \Theta\!\left((1-\eta)p_{\rm sym}\gamma_\ell n_\ell\langle \hat{H}\rangle_{\rho_b} \right) = (1-\eta)p_{\rm sym}\gamma_\ell n_\ell\Theta\!\left(\,N\right),\\
\Lambda(\rho) &= \Theta\!\left((1-\eta)p_{\rm sym}\gamma_\ell(n_\ell{+}1)\langle \hat{K}\rangle_{\rho_b} \right) = (1-\eta)p_{\rm sym}\gamma_\ell(n_\ell{+}1)\Theta\!\left(\,N\right).
\end{aligned}
\end{equation}
These scalings explain the trends in Fig.~\ref{ErgoEta_gamma_r}. 
Because the leakage rate grows linearly with $\gamma_\ell$ and with the “local temperature factor” $(2n_\ell{+}1)$ [Eq.~\eqref{scaling}], the $\alpha_\ell$–activation ridge is robust for small $\gamma_r$ (suppressed only at large $\alpha_\ell$), but it is quenched already at moderate $\alpha_\ell$ when $\gamma_r$ is increased. 
By contrast, the location and width of the $\alpha_c$ lobe are nearly unchanged: the collective detailed balance within the bright sector is unaffected, and $\gamma_\ell$ enters mainly by reducing $p_{\rm sym}$ via leakage rather than by shifting the collective optimum.

\section{Transient charging dynamics near the optimal collective temperature}
\begin{figure}[t]
    \centering
    \begin{subfigure}[b]{0.32\textwidth}
        \includegraphics[width=\textwidth]{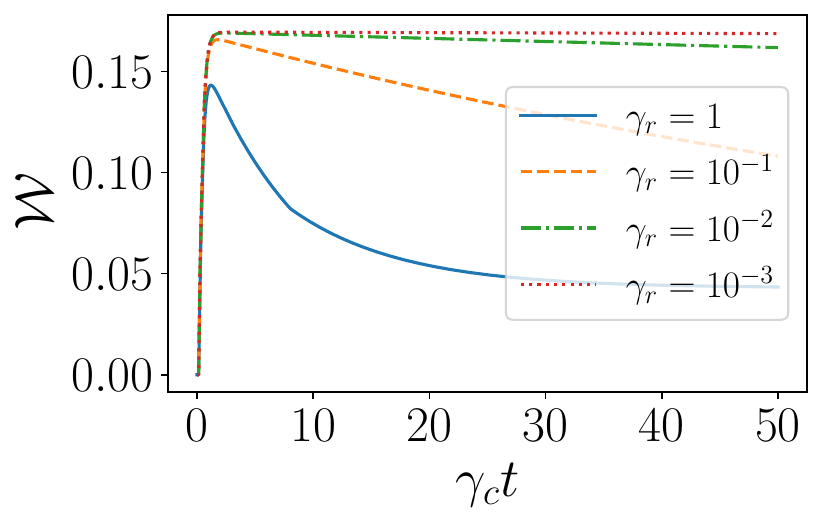}\caption{$N=3$}
        \label{Ergotropy3_gamma_r1}
    \end{subfigure}
    \begin{subfigure}[b]{0.32\textwidth}
\includegraphics[width=\textwidth]{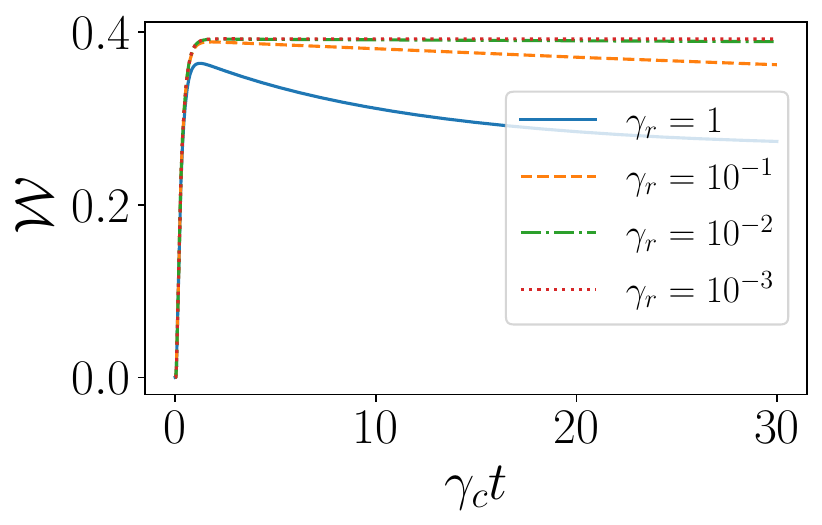}
        \caption{$N=6$}\label{Ergotropy6_gamma_r1}
    \end{subfigure}
    \begin{subfigure}[b]{0.32\textwidth}
\includegraphics[width=\textwidth]{Ergotropy10_gamma_r1.pdf}
        \caption{$N=10$}\label{Ergotropy10_gamma_r1}
    \end{subfigure}
    \captionsetup{justification=justified}
    \caption{\textbf{Charging dynamics near the activation point.}
    Ergotropy $\mathcal W(t)$ versus rescaled time $\gamma_c t$ for $N=3,6,10$ and dissipation ratios $\gamma_r$. The collective bath is set close to its optimal value $\alpha_c^\star$; $\alpha_\ell=0$ is fixed across curves.\justifying} 
    \label{Ergo_gamma_r}
\end{figure}
Figure~\ref{Ergo_gamma_r} shows the time evolution of the ergotropy $\mathcal W(t)$ of the QB when the collective bath is set close to its optimal value $\alpha_c^\star$ (the activation point), for three system sizes $N\in\{3,6,10\}$ and several ratios $\gamma_r$ (with $\alpha_\ell=0$ fixed across curves). Two robust dynamical regimes appear from the competition
\begin{equation}
    \mathcal L \;=\; \eta\gamma_c\!\left[(n_c{+}1)\,\mathcal D[J_-]+n_c\,\mathcal D[J_+]\right]
    \;+\; (1-\eta)\gamma_\ell \sum_{i=1}^N\!\left[(n_\ell{+}1)\,\mathcal D[\sigma_-^{(i)}]+n_\ell\,\mathcal D[\sigma_+^{(i)}]\right].
\end{equation}
\paragraph*{Early-time collective pumping (interference-enabled).}
For $\gamma_c t\ll 1$, the evolution is dominated by the collective jumps $J_\pm$. Because these jumps are \emph{indistinguishable} across emitters, path amplitudes add and produce constructive interference (superradiant matrix elements). Population is driven \emph{up} the bright Dicke ladders with rates $(j-m)(j+m+1)$ [equivalently $(N{-}k)(k{+}1)$], yielding an initial rise of $\mathcal W(t)$ largely \emph{insensitive} to $\gamma_\ell$ at fixed $\alpha_\ell=0$—hence the near collapse of curves at early times. More precisely: if we write the generator as $\mathcal{L} = \eta\gamma_c \mathcal{L}_{\text{coll}} + (1-\eta)\gamma_\ell \mathcal{L}_{\text{loc}}$ and expand
\begin{equation}
\rho(t) = \rho(t_0) + t \mathcal{L}[\rho(t_0)] + \frac{t^2}{2} \mathcal{L}^2[\rho(t_0)] + \cdots \,,
\end{equation}
then, plotted versus the rescaled time $x = \gamma_c t$,
\begin{equation}
\rho(t) = \rho(t_0) + x \eta\mathcal{L}_{\text{coll}}[\rho(t_0)] + x (1-\eta){\gamma_r} \mathcal{L}_{\text{loc}}[\rho(t_0)] + O(x^2).
\end{equation}
Starting from ground state, $\mathcal{L}_{\text{loc}}[\rho(t_0)]\approx0$ for $\alpha_\ell=0$ while $\mathcal{L}_{\text{coll}}$ keeps pumping energy into the battery for $\alpha_c^*$. Thus the leading slope of $\mathcal{W}(t)$ at small times is set by $\mathcal{L}_{\text{coll}}$, and the dependence on $\gamma_r$ for $x = \gamma_c t \ll 1$ is negligible, so the early-time traces of $\mathcal{W}(t)$ for different $\gamma_r$ are nearly identical---hence the curves nearly collapse.

\paragraph*{Crossover, interference erosion, and steady state.}
At later times the local channel injects which-path information (site-resolved jumps $\sigma_\pm^{(i)}$) that \emph{degrades the interference} sustaining collective transport and mixes permutation sectors. The competition sets two coarse time scales,
\begin{equation}\nonumber
\tau_{\rm coll}^{-1}\sim \gamma_c\,\Psi(\alpha_c,N),\qquad
\tau_{\rm loc}^{-1}\sim \gamma_\ell\,\Phi(\alpha_\ell,N),
\end{equation}
with $\Psi$ the interference-enabled collective pumping and $\Phi$ the local which-path mixing. Large $\gamma_r$ pushes the dynamics into a regime where interference is quickly eroded and the asymptotic ergotropy drops; small $\gamma_r$ leaves collective pumping dominant and the plateau value essentially unchanged.

\paragraph*{Size dependence.}
Sensitivity to $\gamma_r$ at steady state diminishes rapidly with $N$:
\begin{itemize}
\item For $N=3$ [Fig.~\ref{Ergotropy3_gamma_r1}], $\mathcal W(t)$ exhibits an overshoot and decay when $\gamma_r=1$, while for $\gamma_r\le 10^{-2}$ it approaches a larger steady value monotonically.
\item For $N=6$ [Fig.~\ref{Ergotropy6_gamma_r1}], the post-peak decay is weaker and steady-state values for $\gamma_r\le 10^{-1}$ are very close.
\item For $N=10$ [Fig.~\ref{Ergotropy10_gamma_r1}], the curves nearly coincide: enhanced collective matrix elements scale with $N$, so $\tau_{\rm coll}\ll \tau_{\rm loc}$ in the shown range, and the steady-state ergotropy becomes almost insensitive to $\gamma_r$.
\end{itemize}
These trends agree with the steady-state maps: near $\alpha_c^\star$ the \emph{location} of the activation window is set by the collective detailed balance (an interference-protected feature, hence only weakly dependent on $\gamma_r$), while the \emph{height} decreases when local which-path processes undermine interference. Growing $N$ amplifies interference-enabled collective rates and suppresses the influence of $\gamma_r$ on both transient dynamics and the asymptote.

\twocolumngrid

\bibliography{References}

\end{document}